\documentclass[10pt]{article}
\usepackage{jheppub}
\usepackage[utf8]{inputenc}
\usepackage[english]{babel}
\usepackage[T1]{fontenc}
\usepackage{amsthm}
\usepackage{amsmath}
\usepackage{amssymb}
\usepackage{hyperref}
\usepackage{appendix}
\usepackage{bbm}
\usepackage{bm} 
\usepackage{comment}
\usepackage{xcolor} 
\usepackage{graphicx}
\usepackage{scalerel}
\usepackage[normalem]{ulem}

\newtheorem{Theorem}{Theorem}
\newtheorem{Definition}{Definition}
\newtheorem{Proposition}{Proposition}

\newtheorem{Remark}{Remark}

\newcommand{\red}[1]{\textcolor{red}{#1}}

\newcommand{\orange}[1]{\textcolor{orange}{#1}}
\newcommand{\magenta}[1]{\textcolor{magenta}{#1}}
\newcommand{\bb}[1]{\boldsymbol{#1}}
\newcommand{\mc}{\mathcal}
\newcommand*{\paral}{\stretchrel*{\parallel}{\perp}}

\title{The importance of being empty: Hopfield neural networks with diluted examples}
\author[a,1]{Elena Agliari,}
\author[a,1]{Alberto Fachechi,}
\author[b]{Domenico Luongo}
\affiliation[a]{Department of Mathematics G. Castelnuovo, Sapienza Università di Roma, Italy}
\affiliation[b]{Scuola Normale Superiore, Pisa, Italy}
\affiliation[1]{GNFM-INdAM, Gruppo Nazionale di Fisica Matematica (Istituto Nazionale di Alta Matematica), Italy}
\emailAdd{agliari@mat.uniroma1.it}

\abstract{We consider Hopfield networks, where neurons interact pair-wise by Hebbian couplings built over $i$. a set of definite patterns (ground truths), $ii$. a sample of labeled examples (supervised setting), $iii$. a sample of unlabeled examples (unsupervised setting). 
We focus on the case where ground-truths are Rademacher vectors and examples are noisy versions of these ground-truths, possibly displaying some blank entries (e.g., mimicking missing or dropped data), and we determine the spectral distribution of the coupling matrices in the three scenarios, by exploiting and extending the Marchenko-Pastur theorem.
By leveraging this knowledge, we are able to analytically inspect the stability of the ground truths, as well as the generalization capabilities of the networks. 
In particular, as corroborated by long-running Monte Carlo simulations, the presence of blank entries can have benefits in some conditions, suggesting strategies based on data sparsification; the robustness of these results in structured datasets is confirmed numerically. Finally, we demonstrate that the Hebbian matrix, built on sparse examples, can be recovered as the fixed point of a gradient descent algorithm with dropout, over a suitable loss function.
}

\begin{document}
\maketitle


\section{Introduction}

In the last decades, science and technology have been revolutionized by the advent of modern Artificial Intelligence (AI), whose success is mainly due to the availability of big data at small cost as well as the high computational capabilities of parallel processing systems. This very success, however, poses significant challenges in implementing Machine Learning models able to accomplish the desired tasks, especially given the typical size of benchmark datasets and the extreme complexity of modern AI solutions.
In particular, several strategies have been investigated in the literature in order to reduce the sensitivity of the generalization power of data-driven decision-making systems and high-complexity neural networks on the quality of training data. These strategies aim to optimize the performances of such models through approaches such as data cleaning and augmentation \cite{little2019statistical,shorten2019survey,feng2021survey}, the usage of specific cost-functions \cite{huber1992robust,barron2019general} or algorithms \cite{breiman2001random}, regularization, dropout and noise injection \cite{prechelt2002early,srivastava2014dropout,bishop1995training}, or more sophisticated methods and training procedures \cite{goodfellow2014explaining,chen2020simple,devlin2019bert}. 
Nevertheless, while investigations in this direction are undoubtedly crucial, they are often rather empirical in nature, and in practical applications an extensive fine-tuning of parameters is often required. In fact, a comprehensive and well-controllable theoretical framework, equipped with rigorous mathematical techniques, for interpreting their working principle is still under construction.
\par\medskip
Recently, the Hopfield model -- a prototypical energy-based neural network --  has been revised \cite{EmergencySN,Leuzzi_2022,alemanno2023supervised,negri2024memorization,negri2024random,catania2025theoretical} and consistently analyzed to address the emergence of generalization in a simplified but solid-grounded setting (see also \cite{Fontanari-1990,Litinskii,Theumann}). In particular, in \cite{regularizationdreaming} an unsupervised reconstruction task of hidden patterns with a $L_2$-regularized loss function has been mapped to Hopfield-like neural networks, with the regularization parameter shaping the attraction basins of stored vectors \cite{fachechi2019dreaming}. Consequently, the emergence of the generalization phase can be interpreted as the result of the coalescence of attraction basins associated to training examples, while overfitting occurs for choices of the regularization parameters resolving them. 
In this paper, 
we advance this investigation, focusing on the effect of missing data on the retrieval and generalization capabilities of such networks. The absence of some data entries is assumed as completely random, that is, the probability of missingness is independent of both observed and unobserved variables (this might occur, for instance, when a sensor fails randomly during data collection, or can be made {\it a posteriori} as in dropout-like techniques \cite{srivastava2014dropout}). 
Specifically, we consider the following framework: we have a set of $K$ \textit{patterns}, each corresponding to a different ground-truth feature, represented by a binary vector of length $N$ and denoted as $\boldsymbol \xi^{\mu} \in  \{ -1, +1\}^N$, for $\mu=1,...,K$. The available dataset comprises $M$ different measurements of these features, denoted as $\boldsymbol \xi^{\mu,A}$, for $\mu=1, ...,K$ and $A=1, ...,M$. 
These measurements are affected by a random error, meaning that some entries can be flipped, namely $\xi_i^{\mu,A} = -\xi_i^{\mu}$ or can be missing, namely $\xi_i^{\mu,A} =0$, for some $i, \mu, A$, and are combined into the Hebbian matrix. 
Following \cite{leonelli2021effective,AABD-NN2022,AAKBA-EPL2023}, we envisage two ways to build up the Hebbian matrix, that mirror, respectively, a supervised and unsupervised protocol yielding to the coupling matrices $\bm \Gamma^{s}$ and $\bm \Gamma^{u}$; in the first setting, knowledge of class partitioning of the examples is used to approximate the unknown $\mu$-th pattern through the empirical average of the examples belonging to the $\mu$-th class, and these averages $\bar{\bb \xi}^\mu$ are then stored in the Hebbian matrix: $ \Gamma^s _{i,j}\propto \sum_\mu \bar {\xi}^\mu_i\bar { \xi}^\mu_j$; conversely, in the unsupervised scenario, such knowledge is unavailable, thus single entries are all combined in the coupling matrix as $\Gamma^u_{i,j} \propto \sum_{\mu,A} \xi^{\mu,A}_i \xi^{\mu,A}_j$, so that each training example enters in the energetic landscape with its own attraction basin. As a first result, we analytically inspect the algebraic properties of these matrices, obtaining their spectral distribution. Next, we are able to assess the generalization capabilities of the models in terms of the first and second momenta of these spectral distributions; this approach is also shown to encompass the popular signal-to-noise analysis. In this way, it is possible to analytically highlight that, in the supervised case, dilution always impairs the generalization capabilities, with a stronger effect when the initial datum is farther from the target; conversely, in the unsupervised case, our theoretical results suggest that dilution can enhance the performance as long as the quality of data points is relatively high and the load $K/N$ is relatively large. Remarkably, these predictions are confirmed by Monte Carlo simulations in the random setting and for structured datasets. This supports the claim that the beneficial effect of dilution is a general phenomenon occurring in attractor neural networks with Hebbian couplings built on empirical data, leading to superior generalization capabilities and suggesting the development of targeted dropout strategies.
\par\medskip
The paper is organized as follows. 
In Sec.~\ref{sec:model} we present the model and the datasets, distinguishing between basic storing, supervised and unsupervised protocols. In Sec.~\ref{sec:1step}, we introduce key quantities for assessing the model performance and we provide preliminary insights on its generalization capabilities. Next, in Sec.~\ref{sec: spectra} we apply the Marchenko-Pastur theorem to determine the asymptotic spectral distribution of the coupling matrices in all the three protocols, namely $\bm \Gamma$, $\bm \Gamma^s$, and $\bm \Gamma^u$, and relate the momenta of these distributions to pattern stability and generalization capabilities, thus providing quantitative predictions. 
Our non-trivial findings concerning the generalization capabilities in the unsupervised setting are further inspected numerically in Sec.~\ref{sec:numerics}, where the benefits of sparsification in the training data is deepened. Finally, Sec.~\ref{sec:final} contains our conclusions and outlooks. 
Lengthy proofs of theorems, technical details and further insights on the mechanism underlying dataset dilution are collected in the Appendices.

\section{Setting and definitions} \label{sec:model}
\subsection{Attractor neural networks}
We consider a fully-connected associative neural network (ANN) made of $N$ binary neurons whose state is denoted by a binary variable $\sigma_i \in \{ -1, +1\}$ for $i=1, ..., N$; the overall neural configuration is represented by the vector $\boldsymbol \sigma = (\sigma_1, ..., \sigma_N) \in \{-1, +1\}^N$. In general, $\boldsymbol \sigma$ is a dynamic variable and it is iteratively updated according to the \textit{parallel} evolution rule
\begin{equation} \label{eq:dynamics}
    \boldsymbol \sigma^{(n+1)} = \textrm{sgn} (\boldsymbol J \cdot\boldsymbol \sigma^{(n)} ),
\end{equation}
where the superscript $n$ counts the number of updating steps and $\boldsymbol J$ is a $N \times N$ real matrix encoding the synaptic efficiencies.
The previous equation is equivalent to say that each single neuron at the $(n+1)$-th step takes the value $\sigma^{(n+1)}_i = \textrm{sgn} (\sum_j J_{i,j} \sigma^{(n)}_j )$. The matrix entry $J_{i,j}$ therefore represents the strength by which the $j$-th neuron influences the state of the $i$-th neuron. 
This iterative updating can be interpreted in terms of a (infinitely deep) multi-layer neural network, where the $n$-th layer encodes for the configuration $\boldsymbol \sigma^{(n)}$ and the information propagates from one layer to the next according to the rule \eqref{eq:dynamics}. 
Specifically, $\boldsymbol{\sigma}^{(0)}$ is interpreted as the \textit{input} of our neural network, and we are interested in the related output $\boldsymbol{\sigma}^{(T)}$ for some large $T$ (\emph{vide infra}) or, more precisely, in the input-output mapping performed by the machine.
%
To this purpose, it is worth recalling that (see e.g., \cite{Picco,Coolen}), if $\boldsymbol{J}$ is a symmetric matrix with non-negative diagonal entries 
, the dynamics described in eq.~\eqref{eq:dynamics} admits the bounded Lyapunov function $L(\boldsymbol{\sigma}) = -\sum_{i} | \sum_j J_{i,j} \sigma_j |$, which implies that the dynamics converges to a {fixed point}, i.e. a configuration invariant under the dynamics \eqref{eq:dynamics}, or to a period-2 {limit cycle}.
%
%
%
As the dynamics converges to a fixed point $\boldsymbol{\sigma}^{(\infty)} := \lim_{n \to \infty} \boldsymbol\sigma^{(n)}$, we will interpret it as the \emph{output} of the ANN.
 
\subsection{Synthetic dataset, ground-truths and dilution}\label{sec:2.2}
Let us suppose to have $K$ patterns of information, also referred to as \textit{archetypes} or \textit{ground-truths}, codified by binary vectors of length $N$ and denoted by $\boldsymbol{\xi}^\mu =(\xi_1^{\mu}, \dots, \xi_N^{\mu})  \in \{-1,+1\}^N$, $\mu = 1,\dots, K$. This is the information that we want to retrieve and we state that a pattern, say $\boldsymbol \xi^{\mu}$, is retrievable by the network if the configuration $\boldsymbol \sigma = \boldsymbol \xi^{\mu}$ is a fixed point for the dynamics \eqref{eq:dynamics}, with a non-vanishing attraction basin $\mathcal B_{\mu}$. The configurations belonging to the attraction basin $\mathcal B_{\mu}$ are interpreted as input configurations that prompt the retrieval of the pattern $\boldsymbol \xi^{\mu}$. Thus, the larger the cardinality $|\mathcal B_{\mu}|$ and the more attractive the pattern $\boldsymbol \xi^{\mu}$ is. On the other hand, if $|\mathcal B_{\mu}| =1$, $\boldsymbol \xi^{\mu}$ is an unstable fixed point, which is useless for retrieval purposes. We also introduce the \textit{load}, defined as $\alpha_N := K/N$, representing the amount of information that we aim to store in the network, given a certain amount of resources (i.e., of neurons).
In the limit of large $N$, we will focus on the so-called ``high-storage'' regime, where $(0, \infty) \ni \alpha: = \lim_{N \to \infty} \alpha_N$.
\par\medskip
As far as the statistical properties of patterns are concerned, we will conduct our analysis under a controllable setting, that is, we will assume that $\{\xi^\mu_i\}_{\mu, i}$ is a class of i.i.d. random variables, with symmetric Rademacher distribution $Rad\left(0\right)$, thus:
\begin{equation} \label{eq:P_xi}
    \mathcal{P}(\xi^{\mu}_{i})=\frac{1}{2} (\delta_{\xi^{\mu}_{i},+1}+\delta_{\xi^{\mu}_{i}, + 1}).
\end{equation}

Recalling \eqref{eq:dynamics},  
we notice that, in order for each of the patterns $\boldsymbol \xi$ to be a fixed point, the interaction matrix $\bb J$ must be a function of the whole set of patterns. As long as the patterns $\boldsymbol \xi$ are directly accessible, a standard choice of the coupling matrix is provided by the following\footnote{We use the symbol $\bb J$ to denote a generic coupling matrix and the symbol $\bb \Gamma$ for coupling matrices exhibiting a Hebbian structure. Different prescriptions for the coupling matrix (leading to different ANN models) can be considered, possibly yielding to larger affordable loads, such as Kohonen's projection matrix \cite{Kohonen-1972} $\boldsymbol{X} (\boldsymbol{X}^T \boldsymbol{X})^{-1} \boldsymbol X^T$ which reaches the capacity upper bound $\alpha =1$ \cite{Gardner} and the more general dreaming kernel $\boldsymbol{X} (\textbf I + t) (\textbf I + t \boldsymbol{X}^T \boldsymbol{X})^{-1} \boldsymbol{X}^T$, where $t \in \mathbb R^+$ is a tuneable parameter and which saturates to $\alpha=1$ as $t \to \infty$ \cite{FAB-NN2019, 
AABF-JStat2019,LAD,regularizationdreaming}.}
\begin{equation} \label{basicstoringmatrix}
    \Gamma_{i,j} = \frac{1}{N} \sum^K_{\mu=1} \xi^\mu_i \xi^\mu_j, \quad \forall i,j \in \{1,\dots,N\},
\end{equation}
which can be rewritten in matrix form as $\boldsymbol \Gamma =N^{-1} \boldsymbol{X}  \boldsymbol{X}^T,$ where $\boldsymbol{X}$ is a $N \times K$ matrix, with entries $  X_{i,\mu} = \xi^\mu_i$, $\forall i\in\{1,\dots,N\}, \ \forall \mu\in\{1,\dots,K\}$. This choice for $\boldsymbol J$ is inspired by the empirical Hebb's rule \cite{Hebb-1949, hopfield1982neural}, and ensures the retrieval up to a load $\alpha \approx 0.138$ at zero temperature \cite{AGS}. This framework represents an ideal setting, which we refer to as {\it basic storing}. 
\par\medskip
In real applications, we expect to be able to rely only on experimental, noisy realizations of archetypes, that we call \emph{examples}.
In order to have mathematical control, we define $M\in\mathbb{N}^+$ synthetic examples $\{\boldsymbol{\xi}^{\mu,A}\}_{A=1,\dots,M}$ for each archetype $\boldsymbol{\xi}^{\mu}$, as $\xi^{\mu,A}_i = \xi^\mu_i \chi^{\mu,A}_i,$ with $i\in \{1,\dots,N\}$, $\mu\in\{1,\dots,K\}$, and $A\in\{1,\dots,M\}$, and where $\chi^{\mu,A}_i$ encodes the \textit{error} in the $i$-th bit of the $A-$th example related to the $\mu-$th class. The $\{\chi^{\mu,A}_i\}_{\mu, A, i}$ are modeled as i.i.d. random variables, independent of the archetypes, with probability
\begin{equation} \label{eq:P_chi}
    \mathcal{P}(\chi^{\mu, A}_{i})=(1-d)\frac{1+r}{2}\delta_{\chi^{\mu, A}_{i},+1}+(1-d)\frac{1-r}{2}\delta_{\chi^{\mu, A}_{i},-1}+d\delta_{\chi^{\mu, A}_{i},0}.
\end{equation}
Thus, $d \in (0,1)$ is the probability that the entry of any given example is $0$ and this is interpreted as a missing datum; in other words, $d$ represents the expected fraction of blank entries and will be referred to as \emph{dilution} or \emph{sparseness} (e.g., see \cite{agliari2012multitasking, agliari2013parallel, AABCT-JPA2013a, AABCT-JPA2013b,AABR-JStat2023}). 
Hence, with probability $1-d$ the example entry is non-zero, but still it can be affected by an error (i.e., $\xi^{\mu,A}_i \neq \xi^\mu_i$) with probability $\frac{1-r}{2}$. For this reason, $r$ is named \textit{quality} of the dataset \cite{Fontanari-1990,AABD-NN2022,alemanno2023supervised}. With this definition, the two-entry correlation in each example is $\mathbb E[ \xi^{\mu,A}_i \xi^{\mu,A}_j]=\delta_{ij}$, and, given two examples $\boldsymbol \xi^{\mu,A}$ and $\boldsymbol{\xi}^{\nu,B}$, they are independent if $\mu \neq \nu$ and correlated if $\mu = \nu$, due to the fact that they stem from the same ground-truth $\boldsymbol \xi^\mu$. 
%
%
\par\medskip
Given a sample of examples, there may be a teacher who classifies examples according to the corresponding archetype and we refer to this setting as {\it supervised}. Exploiting such information, we can determine class-wise means as
\begin{equation*}
   \boldsymbol{\Bar \xi}^\mu := \frac{1}{M} \sum^M_{A=1} \boldsymbol{\xi}^{\mu,A} \in \left[-1, + 1\right]^N, ~ \mu=1, ...,K
\end{equation*}
and identify these as empirical representations of the ground-truth $\boldsymbol \xi^\mu$. Hence, in this case, we define the supervised Hebbian matrix as
\begin{equation} \label{supervisedhebbmatrix}
    \Gamma^s_{i,j} = \frac{1}{N} \sum^K_{\mu=1} \Bar\xi^\mu_i \Bar\xi^\mu_j,\quad \forall i,j \in \{1,\dots,N\},
\end{equation}
 or, in matrix form, as $\boldsymbol{\Gamma}^s =N^{-1} \Bar{\bb{X}} \Bar{\bb{X}}^T$, where $\Bar{\bb{X}}$ is a $N \times K$ matrix, storing all the means of the examples, thus $\Bar X_{i,\mu} = \Bar\xi^\mu_i$ for all $i\in\{1,\dots,N\}$ and $\mu\in\{1,\dots,K\}$.\\

In the end, we deal with the {\it unsupervised} setting, where class labels remain disclosed and we cannot cluster examples. Thus, the least biased interaction matrix (see App.~\ref{app:dropout} and \cite{regularizationdreaming}) is given by the following unsupervised Hebbian matrix that reads as
\begin{equation} \label{unsupervisedhebbmatrix}
    \Gamma^u_{i,j} = \frac{1}{P} \sum_{\mu, A} \xi^{\mu,A}_i \xi^{\mu,A}_j,\quad \forall i,j \in \{1,\dots,N\},
\end{equation}
where $P=MN$, namely $\boldsymbol{\Gamma}^u = P^{-1} \bb{X}^\text{ex} \bb{X}^{\text{ex}^T}$, where $\bb{X}^\text{ex}$ is a $N \times KM$ matrix, storing all the single examples $X^\text{ex}_{i,(\mu, A)} = \xi^{\mu,A}_i$, for $i=1,\dots,N$, $\mu =1,\dots,K$, and $A=1,\dots,M$.
\par\medskip
Before concluding this section, we emphasize that the setting considered here is slightly different than that explored in \cite{AABR-JStat2023}, where ground-truth patterns themselves were affected by dilution and the noise in examples could preserve the arrangement of blank entries or possibly replace them with random binary entries. In that scenario, the dilution in ground-truth patterns was shown to give rise to ``parallel learning'', namely the simultaneous learning of multiple patterns. 

Further, it is worth specifying that the sparsity in the examples can be an intrinsic feature of the empirical realizations of the ground patterns, or it can derive from external operation, e.g., mimicking the dropout technique, as discussed in the App. \ref{app:dropout}. In the former, we expect that dilution affects both training and validation examples, while, in the latter, we expect that it acts on the level of training examples only, and thus it is only involved in the construction of the coupling matrix, while validation examples will only be characterized by the quality parameter $r$. In this paper, we will pursue the second picture. 


\section{1-step stability, generalization and signal-to-noise} \label{sec:1step}

In this section, we inspect the generalization capabilities of a neural network endowed with the supervised and unsupervised Hebbian interaction matrix. 
Let us consider a certain target pattern $\boldsymbol{\xi}^\mu$ and another configuration denoted as $\boldsymbol x$, lying on the boundary of the Hamming ball $\mathcal B_R (\boldsymbol \xi^{\mu})$ centered in $\boldsymbol{\xi}^\mu$ and with (normalized) radius $R$. These boundaries can be realized by perturbing $\boldsymbol{\xi}^\mu$ as $\boldsymbol x = \boldsymbol \eta \odot \boldsymbol{\xi}^\mu$, 
with $\eta_i\underset{i.i.d.}\sim Rad(p), ~\forall i \in\{1,...,N\}$. Indeed, choosing $p=1-2R$,
\begin{equation*}
    d_H({\boldsymbol x},\boldsymbol{\xi}^\mu) = \frac{1}{2N}\sum_i (1-\eta_i)\underset{N \gg 1}{\approx}\frac12 (1-p) = R.
\end{equation*} 
Then, we can take $\boldsymbol x$ as initial configuration $\boldsymbol \sigma^{(0)} = \boldsymbol x$ and apply \eqref{eq:dynamics} to get $\boldsymbol \sigma^{(1)}$. Next, multiplying both sides of the evolution relation by $\xi_i^{\mu}$, we obtain 
\begin{equation*}
    \sigma_i^{(1)} \xi_i^{\mu} = \textrm{sgn} (\sum_{j=1}^N   J_{i,j}  \xi_j^{\mu} \eta_j \xi_i^{\mu}), ~ i=1,..., N.
\end{equation*}
A positive argument of the sign function means that, in a single step, the state of the $i$-th neuron has remained (or has flipped to get) equal to the target $\xi_i^{\mu}$. 
It is therefore convenient to introduce the random variable
\begin{equation}
		\label{eq:pattern_attract_mod}
		\Delta_i (\boldsymbol \xi^{\mu}, \boldsymbol \eta)=\sum_{j=1}^N J_{i,j} \xi_j^{\mu} \eta_j \xi_i^{\mu},
	\end{equation}
    which shall be specialized as follows: 
\begin{itemize}
    \item $p=1$. In this case, the initial state coincides with $\boldsymbol \xi^\mu$, and $\Delta_i (\boldsymbol \xi^\mu, \boldsymbol \eta)\equiv \Delta_i (\bb \xi^\mu)$ tells us whether the state of the neuron $i$ is invariant $(\Delta_i (\boldsymbol \xi^\mu)>0)$ or not $(\Delta_i (\boldsymbol \xi^\mu)<0)$. We will refer to $\Delta_i(\bb \xi^\mu) $ as the {\it local stability} of the pattern $\bb \xi^\mu$.
    \item $p=r$. In this case, the initial state $\bb x=\bb \eta \odot \bb \xi^\mu$ shares the same statistics (apart from dilution) as the training points $\bb\xi^{\mu,A}$, so that we refer to it as a test (or validation) example. Then, $\Delta_i (\boldsymbol \xi^\mu, \boldsymbol \eta)$ tells us whether, starting from a test example, the neuron $i$-th gets aligned ($\Delta_i (\boldsymbol \xi^\mu, \boldsymbol \eta)>0$) or not ($\Delta_i (\boldsymbol \xi^\mu, \boldsymbol \eta)<0$) with the corresponding entry of the pattern $\bb\xi^\mu$. The former means that the model is exactly reconstructing the pattern at site $i$, and we accordingly refer to $\Delta_i (\bb\xi^\mu,\bb\eta)$ as the {\it local generalization}.
\end{itemize}

Operatively, it is more convenient to work with global and bounded observables and, to this aim, we introduce the configuration overlaps, also known as Mattis magnetizations (related to the pattern $\boldsymbol \xi^{\mu}$), evaluated at time steps $n=0,1$: 
\begin{eqnarray}
\label{eq:m1}
    m^{(0)}(\boldsymbol 
    \xi^{\mu},\boldsymbol \eta) &=&  \frac{1}{N}\sum_{i=1}^N \xi_i^{\mu} \sigma_i^{(0)} = \frac{1}{N} \sum_{i=1}^N \eta_i,\\
    \label{eq:m2}
m^{(1)}(\boldsymbol \xi^{\mu},\boldsymbol \eta) &=&\frac{1}{N} \sum_{i=1}^N \xi_i^{\mu} \sigma_i ^{(1)} = \frac{1}{N} \sum_{i=1}^N \textrm{sgn} [\Delta_i(\boldsymbol \xi^{\mu},\boldsymbol \eta)].
\end{eqnarray}
If $p=1$, the initial condition lies in the pattern itself (e.g. $m^{(0)}=1$), and $m^{(1)}$ corresponds to the overlap of the network configuration w.r.t. $\bb\xi^\mu$ after a single update: if $m^{(1)}<1$, the neural dynamics drives the system away from the initial condition, and thus the pattern is not stable. For this reason, at $p=1$ we will refer to $m^{(1)}(\bb \xi^\mu)$ as (global) stability of $\bb\xi^\mu$. Conversely, if $p=r$ the initial condition is a test examples, and we say that $\boldsymbol \xi^\mu$ is attracting ${\boldsymbol\sigma}^{(0)}$ if $m^{(1)}>m^{(0)}$, thus justifying the term (global) generalization. 
\par\medskip
At the analytical level, linking local and global observables and obtain closed-form expressions for the 1-step Mattis magnetizations is hard, unless working assumptions are introduced. In the following, we will carry out our computations under the definition below:
\begin{Definition}[Gaussian approximation] \label{def:GA}
    Within the Gaussian approximation (GA), we assume that the $\Delta_i$'s are i.i.d. Gaussian random variables, i.e. $\Delta_i\sim_{i.i.d.} \mathcal N(\mu_1,\mu_2-\mu_1^2)$. A discussion on the validity of this approximation is provided in \cite{agliari2024spectral}. 
\end{Definition}
\begin{Remark}
    The quantities $\mu_1$ and $\mu_2$ in the previous definition clearly stands for the first two non-centered momenta of the $\Delta_i$ distribution. A crucial point is that, within the GA and in the thermodynamic limit, these two momenta can be analytically estimated and turn out to be functions of the model parameters (that is, $\alpha$, $M$, $r$ and $d$) and of the probability $p$ only.
\end{Remark}
Within the GA and in the thermodynamic limit, the 1-step overlap w.r.t. the reference configuration $\boldsymbol \xi^{\mu}$ does not depend on $\boldsymbol{\eta}$ but directly on the parameter $p$; then we can define it as
\begin{equation}
    m^{(1)}\left(\boldsymbol \xi^{\mu}, p\right) := \lim_{N \to \infty} m^{(1)}(\boldsymbol x, \boldsymbol \eta).
\end{equation}
Using again the GA and in the thermodynamic limit, we have
\begin{align} \label{eq:general_m1_bis}
    m^{(1)}\left(\boldsymbol \xi^{\mu}, p\right) &= \lim_{N \to \infty}\frac1N \sum_{i=1}^N \textrm{sgn} [ \Delta_i(\boldsymbol \xi^{\mu},\boldsymbol \eta)] = 
2\mathcal{P}(\Delta_i \geq 0) - 1 = \textrm{erf} \Bigg(\frac{\mu_1}{\sqrt{2(\mu_2 - \mu_1^2)}}\Bigg).
\end{align}

\begin{Remark}
    This approach is closely related to the signal-to-noise technique \cite{Amit}. In particular, the requirement that the local pattern stability in the basic storing setting is positive for all $i=1,\dots,N$ precisely reduces to the stability criterion $h_i(\boldsymbol \xi^\mu) \xi^\mu_i >0$, again for all $i=1,\dots,N$, with $h_i(\boldsymbol \xi^\mu)$ being the internal field acting on the $i$-th neuron when the system is prepared in the configuration $\bb\xi^\mu$. In the signal-to-noise method, the left-hand-side of the previous inequality is split into a \emph{signal} and \emph{noise} contributions which tend, respectively, to stabilize and perturb the neuron configuration, and one aims to find the condition under which the former prevails over the latter. This procedure is perfectly equivalent to compare the first and second momenta of the local stability. In fact, requiring $m^{(1)}(\boldsymbol \xi)$ to be larger than a certain threshold, specifically $m^{(1)}(\boldsymbol \xi) >  \textrm{sign}( \frac{1}{\sqrt{2}}) \approx 0.68$, we get $\mathbb E(\Delta_i )/\sqrt{\text{Var}(\Delta_i)}>1$, which is nothing but the requirement $\mu_1 > |(\mu_2 - \mu_1^2)|$,
    with the numerator and denominator in the last inequality interpreted as, respectively, signal and noise. We will deepen this point in App. \ref{app:snr}.
\end{Remark}
\begin{figure}[tb]
    \centering
    \includegraphics[width=\linewidth]{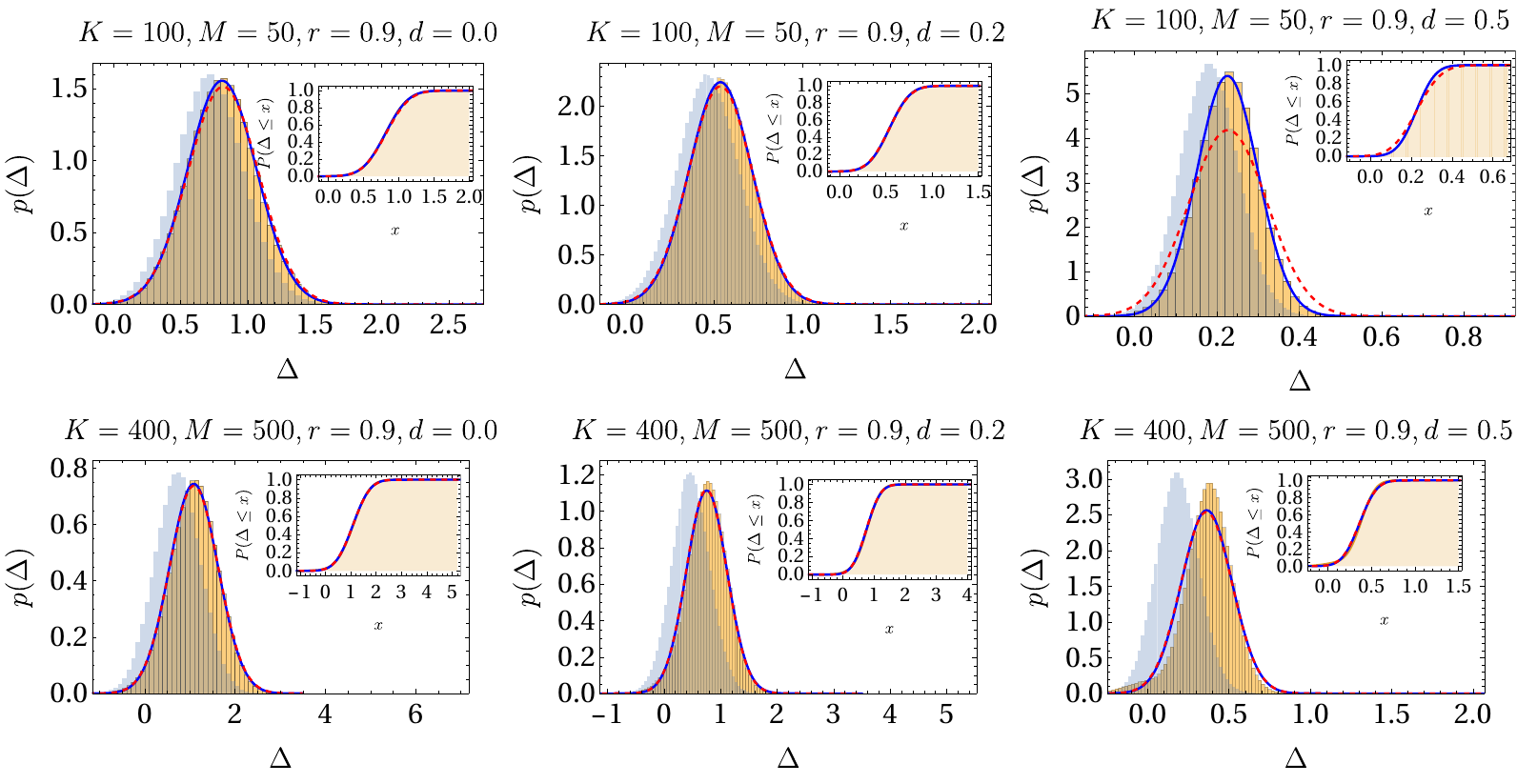}
    \caption{{\bfseries Empirical distributions for the local generalization in the unsupervised setting.} 
    The plots show the empirical distributions for the 1-step local generalization in the unsupervised setting for $K=100$ and $M=50$ (first row) and $K=400$ and $M=500$ (second row), with the quality fixed to $r=0.9$. The values of the dilutions are $d=0$ (first column), $d=0.2$ (second column) and $d=0.5$ (third column). Specifically, the main plots report the histograms of the 1-step local generalization, including or removing the diagonal from the coupling matrix (resp. the yellow and blue histograms). The blue and red dotted curves are Gaussian distributions whose parameters are resp. fitted to the data (blue curve) or estimated by the theoretical predictions with the spectral theory as given by Prop. \ref{attractiveness_prop}. The insets finally report a similar comparison performed over the (again empirical, fitted and theoretical) cumulative distribution functions. In this case, the setting without the diagonal is not reported.
    }
    \label{fig:empirical}
\end{figure}
Before proceeding further, we present some numerical results on the empirical distributions of the 1-step local generalization within the unsupervised scheme\footnote{ Here we focus on the unsupervised setting as this is expected to be (and turns out to be, see Sec.~\ref{sec: spectra}) the most intriguing one among those studied in this work.},
see Fig. \ref{fig:empirical}. First, we notice that the presence of the diagonal (yellow histograms) in the coupling matrix shifts the local activity towards positive values, thereby increasing the probability $P(\Delta \ge0)$ and, consequently, enhancing global generalization.\footnote{This is actually not a surprise since even in the usual Hopfield model the presence of the diagonal enhance, for instance, the local stability of the patterns, and the global stability is $m=\text{erf}((1+ \alpha)/\sqrt{2 \alpha}))$ rather than $m=\text{erf}(1/\sqrt{2 \alpha}))$ without the diagonal, see also \cite{Vinci-2025}.} The plots also highlight the range of validity of the GA (blue curves in the plots representing Gaussian fits to the data): for sufficiently small $K$ the approximation works well, however, at large $K$ and high degree of dilution in the dataset, the local generalization deviates for the Gaussian behavior. This is especially evident by looking at the left tail of the empirical distribution, reflecting an underestimation of the frequency of negative occurrences. As a result, the global generalization tends to be overestimated. Despite this mismatch, the GA remains useful as a qualitative understanding of the retrieval capabilities. As we will discuss later, the theoretical predictions based on this approximation still successfully capture the key phenomenology associated with dataset dilution.
\par\medskip
Further, recalling that we operate at zero temperature and even a low but non-zero value for $\Delta$ can allow reconstruction, 
these plots evidence how dilution affects the network performance (see also the discussion in App. \ref{app:snr}). In fact, when $\alpha$ is small, the broadness of the distribution — caused by interference from stored patterns other than the target — is relatively limited, such that $\Delta$ remains positive and the target reconstruction is ensured. In this regime, increasing $d$ can have detrimental effects: it shifts the peak of the distribution linearly w.r.t. $1-d$ to the left, while the standard deviation scales only with the square root of $1-d$, potentially leading to reconstruction failure. When $\alpha$ is large, the situation becomes more complex, with $d$ and $r$ interplaying. Specifically, as $d$ increases, both the mean and standard deviation of $\Delta$ decrease, but at different rates that depend on $r$, leading to distinct regimes. When $r$ is small, increasing $d$ causes the peak to move leftward faster than the spread shrinks, thereby hindering reconstruction. Conversely, when $r$ is large, increasing $d$ leads to a decrease in the coefficient of variation, which facilitates retrieval. These remarks are confirmed by the analysis carried on in the next section.

%

\section{Spectra and momenta of coupling matrices} \label{sec: spectra}
Having expressed the global stability and generalization as a nonlinear function of the first two momenta $\mu_{1,2}$ of the local quantities (see eq.~\ref{eq:general_m1_bis}), we now proceed to link the latter to the spectral properties of the coupling matrices $\mathbf \Gamma$, $\mathbf \Gamma^s$ and $\mathbf \Gamma^u$ (see eqs.~\eqref{basicstoringmatrix},\eqref{supervisedhebbmatrix}-\eqref{unsupervisedhebbmatrix}), which turn out to be captured by a modified Marchenko-Pastur law.
In particular, in the basic storing and supervised settings, we can directly apply the Marchenko-Pastur theorem as the coupling matrices are of Wishart form,\footnote{A Wishart matrix is a random matrix of the form $\frac{1}{K}\bb Y \bb Y^T$, where $\bb Y$ has i.i.d. entries. In this paper, we will often refer to matrices with Gram's form as Wishart-like matrices to stress their structure of covariance matrices, even though the entries are not independent.} while in the unsupervised setting we will rely on some approximations. 
Specifically, the supervised setting constitutes a relatively simple generalization of the basic storing setting, as the empirical means $\bar{\bb \xi}^\mu$ are promoted to representatives of the class associated to $\bb \xi^\mu$, and
\begin{equation*}
    \mathbb{E}[\sqrt\alpha \Bar\xi^\mu_i] = 0, \quad \mathbb{E}[(\sqrt\alpha \Bar\xi^\mu_i)^2] = \alpha (1-d) \Big((1-d) r^2 + \frac{1-(1-d)r^2}{M}\Big).
\end{equation*}
Conversely, in the unsupervised setting, the coupling matrix is not of Wishart form, due to the presence of the kernel $\frac1M \sum_{A} \chi^{\mu,A}_i \chi^{\mu,A}_j$. This implies that the factors $\bb X^{\textrm {ex}}$ display identically distributed but not independent entries, as the examples related to the same archetype are correlated. For this reason, in the unsupervised setting we introduce the Approximate Factorization Method (AFM), to compute \textit{approximately} the spectral distribution of this coupling matrix and capture the crucial scalings between the relevant parameters, which we define as following
\begin{Definition}[Approximate Factorization Method] \label{def:afm}
    Let $\{z^\mu_{i,j}\}_{\mu,i,j}$ be a class of random variables, where each couple $z^\mu_{i,j}$ and $z^\nu_{k,l}$ are independent if $\mu \neq \nu$ or $i\neq k$, $i \neq l$, $j \neq k$ and $j \neq l$. Under the AFM we assume that there exists a class of i.i.d random variables $\{\phi^\mu_i\}_{\mu,i}$, such that
    \begin{equation}
        z^\mu_{i,j} \sim \phi^\mu_i \phi^\mu_j, \quad \forall \mu, \forall i \neq j.
    \end{equation}
\end{Definition}
Clearly, such a claim does not hold in general, yet in the unsupervised setting, where we are interested in the large (but finite) $M$ behavior, this can be a meaningful approximation. Specifically, in our context applying the AFM to the Hebbian kernel means that we assume the existence of a family of i.i.d. random variables $\{\phi^\mu_i \}$ (with suitable first and second momenta) such that, for all $i\neq j$, $\frac1{NM} \sum_{\mu,A} \xi^\mu_i \xi^\mu_j \chi^{\mu,A}_i \chi^{\mu,A}_j\to \frac1N \sum_\mu \xi^\mu_i \xi^\mu_j \phi^\mu_i \phi^\mu_j $ a.s as $M\to\infty$. A more detailed discussion about the validity and application of this approximation is provided in App.~\ref{AFMapp}. With these ingredients, we can state the following:
\begin{Proposition}\label{prop:allspectra}
    Let $\bb J$ be a coupling matrix, and $\mu_{\bb J}$ the empirical spectral distribution, namely
    \begin{equation} 
        \mu_{\boldsymbol{J}}(x) = \frac{1}{N} \# \{j \leq N : \lambda^{\boldsymbol{J}}_j < x \}.
    \end{equation}
    Let $MP(\alpha,\alpha\sigma,s)$ being the modified Marchenko-Pastur distribution with associated probability measure
    $$
    \mathrm d \mu (\lambda) = (1-\alpha) \delta (\lambda-s)\mathrm d \lambda+ \alpha \mathrm d\mu _{bulk}(\lambda),
    $$
    with
    $$
   \mathrm d\mu_{bulk}(\lambda)=\frac1{2\pi}\frac{\sqrt{(\lambda_+ -\lambda)(\lambda-\lambda_-)}}{\alpha \sigma (\lambda-s)}\bb 1_{\lambda \in [\lambda_-,\lambda_+]}\mathrm d \lambda
    $$
    and $\lambda_\pm = \sigma (1\pm \sqrt\alpha)^2+s$
    Then, in the thermodynamic limit $N \to \infty$:
    \begin{itemize}
        \item in the basic storing setting ($\bb J=\bb \Gamma$), $\mu _{\bb \Gamma}$ converges to a modified Marchenko-Pastur distribution $MP(\alpha,\alpha,0)$, that is $\sigma=1$ and $s=0$;
        \item in the supervised setting ($\bb J = \bb \Gamma ^s$), $\mu _{\bb \Gamma^s}$ converges to a modified Marchenko-Pastur distribution $MP(\alpha,\alpha \sigma^s,0)$, that is $s=0$ and
        \begin{equation}\label{sigmasup}
           \sigma = \sigma^s(r,d,M) =(1-d) \Big((1-d) r^2 + \frac{1-(1-d)r^2}{M}\Big); 
        \end{equation}
        \item in the unsupervised setting ($\bb J = \bb \Gamma^u$) and within validity of AFM, $\mu _{\bb \Gamma^u}$ converges to a modified Marchenko-Pastur distribution $MP (\alpha,\alpha \sigma^u, \alpha(1-d-\sigma^u))$, with 
        \begin{equation}\label{sigmau}
            \sigma=\sigma^u (r,d,M)=  \sqrt{(1-d)^4r^4+(1-d)^2\frac{1-(1-d)^2r^4}{M}}.
        \end{equation}
    \end{itemize}
\end{Proposition}
\begin{Remark}
The result provided in Prop.~\ref{prop:allspectra} for the supervised setting is an improvement compared to \cite{agliari2024spectral}, as it takes into account the number of examples $M$ without assuming $M\to\infty$.
\end{Remark}
The proof of Prop. \ref{prop:allspectra} 
can be found in App. \ref{app:proof_1}, along with numerical checks to corroborate the validity of these results even at relatively small network and dataset sizes (i.e., $N \sim 10^2$, $M \sim 10$). 
\par\medskip
Once the spectra of the Hebbian matrices are fully characterized, the last theoretical point is to bridge between these spectra (and in particular the momenta of Marchenko-Pastur distribution) and the first two momenta of the 1-step stability and generalization which, under the GA, captures the (short-term, i.e. 1-step) response of the network under neural dynamics, see eq.~\eqref{eq:m1}-\eqref{eq:m2}. This is the content of the following proposition.

\begin{Proposition}\label{attractiveness_prop}
    In the thermodynamic limit and under the GA, the first and second momenta of the local generalization related to the target pattern $\bb\xi^\mu$ in the supervised and unsupervised cases read
        \begin{eqnarray}
            \mu^{s/u}_1 &=& \frac{p}{2\alpha} \big(\kappa^2(\alpha, \sigma^{s/u}, s^{s/u}) + \kappa^2(\alpha, 1, 0) - \kappa^2(\alpha, \sigma^{s/u}_\pm, s^{s/u}_-)\big), \label{eq:mu1su} \\
            \mu^{s/u}_2 &=& (1-p^2) \kappa^2(\alpha, \sigma^{s/u}, s^{s/u})\nonumber \\&+& \frac{p^2}{6\alpha} \big( \kappa^3(\alpha, \sigma^{s/u}_+, s^{s/u}_+) + \kappa^3(\alpha, \sigma^{s/u}_-, s^{s/u}_-) - 3 \kappa^3(\alpha, 1, 0)\big),\label{eq:mu2su}
        \end{eqnarray}
    where $\kappa^n\left( \alpha,\sigma, s\right)$ is the $n-$th moment of a $MP\left(\sigma, \alpha\sigma, s\right)$, $\sigma^{s/u}_\pm$ and $s^{s/u}_\pm$ 
    are defined in App.~\ref{app: proof_4}.\footnote{We recall that the asymptotic distribution is derived by means of the AFM for the unsupervised setting. In addition, we set $s^s=0$ to avoid redundancies in the statement of the proposition. Finally, notice that $\mu^{s/u}_i$ do not depend on the specific concept $\bb\xi^\mu$, because of the symmetry induced by Eq. \eqref{eq:P_xi}.}
\end{Proposition}

The proof of the Proposition is detailed in App. \ref{app: proof_4}, while here we directly show how $m^{(1)}$, determined by plugging Eqs. (\ref{eq:mu1su}-\ref{eq:mu2su}) into Eq. \eqref{eq:general_m1_bis}, varies w.r.t. the system parameters, and we compare these analytical estimates (subjected to the GA, the large $N$ limit and, for the unsupervised setting, the AFM) with numerical estimates. 
\begin{figure}
\centering   \includegraphics[width=0.8\textwidth]{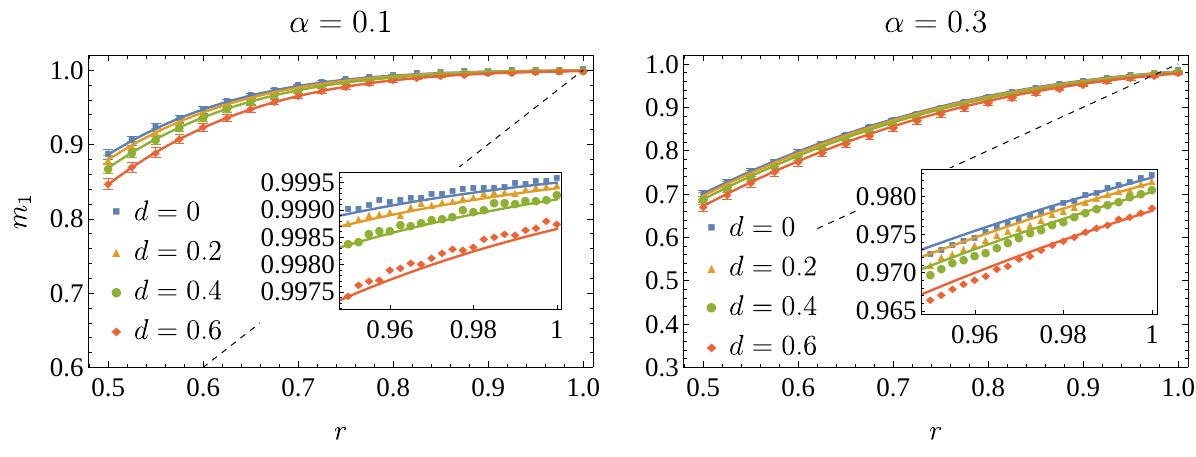}\\
    \centering  \includegraphics[width=0.8\textwidth]{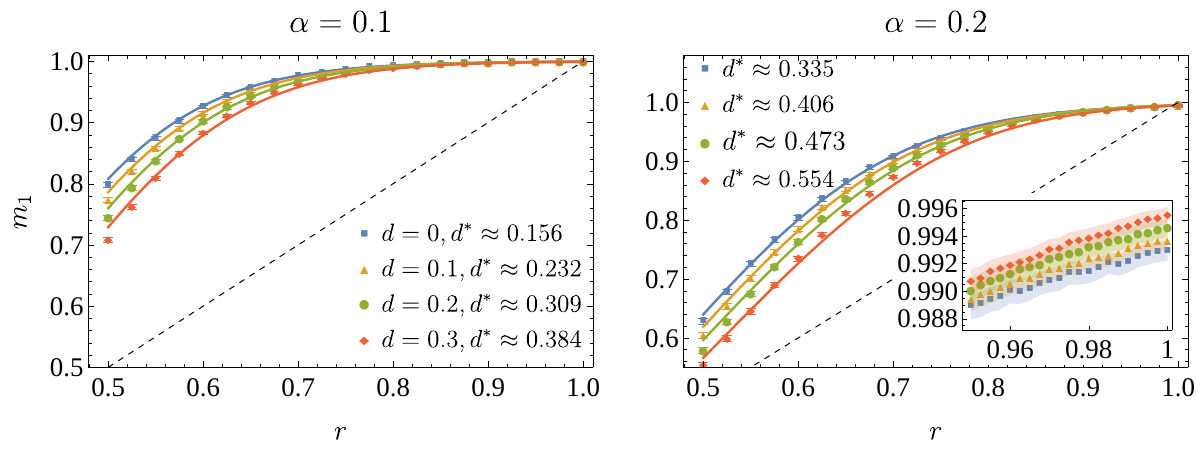}
    \caption{ {\bfseries Generalization in the supervised and unsupervised setting.} 
    The plots show a comparison between the numerical estimates (dots) for the quantity $m^{(1)}(\boldsymbol \xi, r)$ at a dilution level $d$ and the theoretical expressions (solid curves) for the best-fitting value $d^*$ -- obtained comparing the predictions in Prop. \ref{attractiveness_prop} and the numerical findings -- in the supervised (first row) and unsupervised (second row) setting. For the sake of readability, in the lower right plot we reported only the best-fitting ($R^2\approx 1$) values $d^*$ of the dilution parameters, which always refer (from top to bottom) to $d=0,0.1,0.2,0.3$ in the numerical simulations. Numerical results are averaged over $100$ different realizations fixing $N=500, M=100$. We compare results for $\alpha=0.1$ (left) and $\alpha=0.2$ (right). For the latter, we also report a zoom on high values of the dataset quality ($r\simeq 1$, inset plot). Again, error bars are not reported, due to the low magnitude of the relative errors of the numerical simulations.}
    \label{GEN}
\end{figure}
Again, we find substantial agreement between the theoretical predictions and the numerical data in the supervised setting, with the dilution having mild (but negative) effects on the generalization capabilities of the model, as reported in Fig. \ref{GEN}, first row. As for the unsupervised setting (second row, same figure), the theoretical predictions exhibit a deviation from the data due to the non-Gaussian behavior of the local generalization. This implies an underestimation of the effects of dilution. However, we can fit the numerical data with the analytical formulas and find a best-fitting value $d^*$, as reported in Fig. \ref{GEN}. 
Remarkably, our findings suggest that diluting the training dataset at high load $\alpha$ can yield a positive effect, as highlighted in the inset in Fig. \ref{GEN}, lower right plot.
\begin{figure}[tb]
    \centering
    \includegraphics[width=0.9\textwidth]{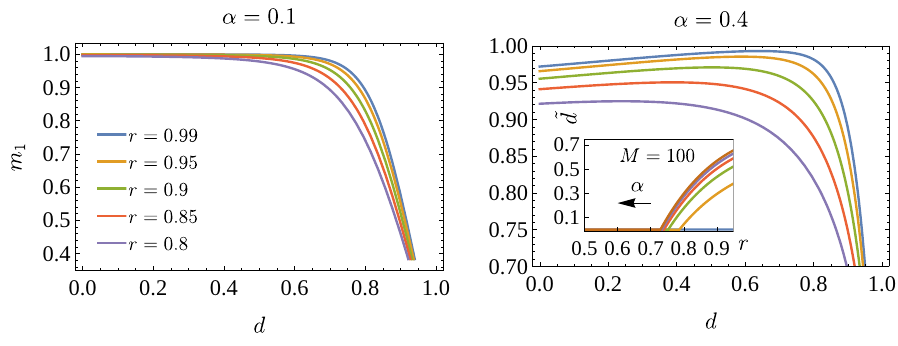}
    \includegraphics[width=0.9\textwidth]{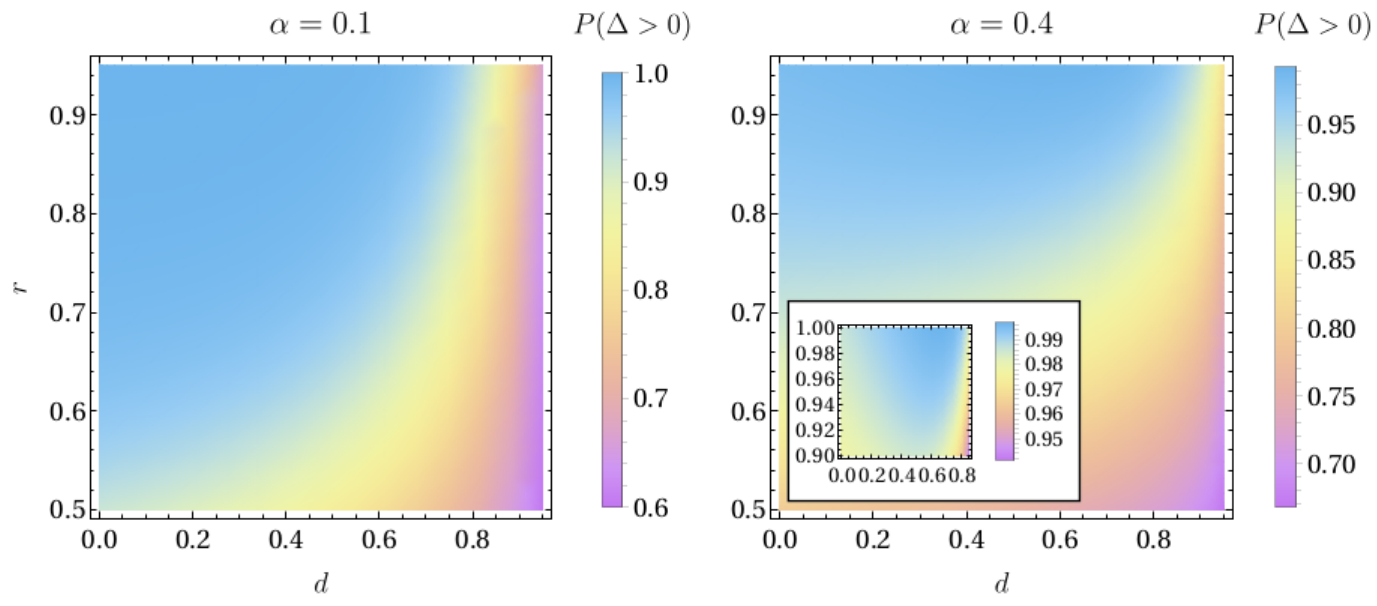}
    \caption{{\bfseries A resuming picture of the capacity of generalization in unsupervised setting.} 
    In the first row, we report the behavior of the 1-step Mattis magnetization as a function of the dilution parameter $d$ for various values of $r$ and $\alpha=0.1$ (left) and $\alpha=0.4$ (right). The inset in the upper right plot highlights the behavior of $\tilde d$ (i.e., the values of the dilution parameter at which, for given $r$, the 1-step magnetization develops a maximum) as a function of the dataset quality, for different values of the load $\alpha$ (going from right side to the left side, $\alpha=0.1,0.2,0.3,0.4,0.5,0.6$). In the second row, we report the heat maps for the probability that the attractiveness of the ground-truth is positive, namely $\mathcal P(\Delta\ge 0) = \frac12 (1+m^{(1)})$ -- see also App. \ref{app:snr}, again for $\alpha=0.1$ (left) and $\alpha=0.4$. In the lower right plot, the inset reports a zoom to the portion of parameters plane $(r,d)\in (0.9\div 1,0\div 1)$ for highlighting the non-monotonous behavior of the 1-step magnetization as a function of $d$. In all the plots, the number of examples per class is fixed to $M=100$.}
    \label{fig:theo_cmaps}
\end{figure}
This effect was already highlighted in the previous section, when examining the distribution of the local generalization, and will be further emphasized hereafter through the analysis of the behavior of the theoretical 1-step magnetization $m^{(1)}$ (exploiting Eq.~\eqref{eq:general_m1_bis} and Prop.
\ref{attractiveness_prop}) as a function of $d$, see Fig. \ref{fig:theo_cmaps}. 
In the upper left plot ($\alpha=0.1$), the magnetization is a non-increasing function of $d$ regardless of the dataset quality $r$ (while only results for high quality are reported, the same situation holds for lower quality). Indeed, for a load below the Hopfield model's critical storage capacity ($\alpha_c\approx0.14$), the unsupervised setting is expected to effectively reconstruct the ground-truths without diluting the dataset, hence no beneficial effects are observed. On the other hand, increasing $\alpha$, some non-trivial effects may emerge, depending on the interplay between dilution and noise due to pattern interference, as we show hereafter.
In the upper right plot ($\alpha=0.4$), for high values of quality, the 1-step magnetization turns out to be a non-monotonous function of $d$, exhibiting a maximum at some nonzero value of the dilution parameter (say $\tilde d$), with reconstruction capabilities progressively lost at higher $d$. In the inset of the same plot, we report the dependence of $\tilde d$ on $r$ for $M=100$ and for various values of the load $\alpha$. 
In the lower plots of Fig. \ref{fig:theo_cmaps}, we report the color maps for the probability $\mathcal P(\Delta\ge0)$, which -- under the GA (Def. \ref{def:GA}) and in the thermodynamic limit -- is exactly $\frac12 (1+m^{(1)})$, see also App. \ref{app:snr}. In the figure, we explore the plane $(d,r)$ for all possible values of the dilution parameter and for reasonable high dataset quality. Again, at low load ($\alpha=0.1$, lower left plot) and $r$ high enough, mild dilution has no effects on the generalization capabilities of the model, conversely, high $d$ results in a reduction of the probability $\mathcal P(\Delta \ge0)$, leading to a breakdown of reconstruction capabilities. As for the high load ($\alpha=0.4$, lower right plot), dilution is beneficial for sufficiently high $r$, as can be also seen by inspecting the zoom reported in the inset: 
 sparsing the training dataset here increases the probability of $\Delta\ge0$, yielding better generalization capabilities w.r.t. non-diluted case. 
\par\medskip
So far, we used random matrix theory tools to outline a qualitative picture for the effects of dilution strategies on the training dataset, both in the supervised and unsupervised settings. We emphasize that, in developing the theoretical framework for the latter scenario, we $i)$ derived the spectral distribution by means of the AFM \eqref{def:afm}; $ii)$ analytically estimated the generalization capabilities from 1-step dynamics under the GA \eqref{def:GA} and in the thermodynamic limit. Even though the qualitative validity of our theoretical framework is evident, a quantitative check based on long-run Monte Carlo simulations is in order and will be addressed in the next section.
%

\section{Numerical simulations for the unsupervised setting} \label{sec:numerics}

In this section, we further investigate the role of the dilution in the unsupervised setting by performing extensive numerical simulations, allowing the system to relax toward a fixed point under neural dynamics \eqref{eq:dynamics}. We will focus on the stability and  the generalization capabilities from unseen data in a reconstruction task. In doing this, dilution will only be inserted in the training data (that is, the examples from which we construct the coupling matrix $\bb \Gamma^u$): as already stated, such a procedure can be interpreted as a dropout-like technique which is implemented by making ``silent'' a fraction $d$ of the bits in the examples. This means that, when analyzing the generalization capabilities of the model, validation data are {\it not} diluted. 
\par\medskip
The numerical experiments are performed in this way. Depending on the feature under consideration (stability or generalization), we extract an initial condition $\bb \sigma^{(0)}$ associated to a given pattern $\bb \xi^\mu$ for some $\mu = 1,\dots ,K$, and perform the dynamics \eqref{eq:dynamics} until the stability criterion $h_i^{(n)} \sigma_i^{(n)} >0$ for all $i=1,\dots,N$ is met, where $h_i^{(n)} = \sum_j J_{i,j}\sigma_j^{(n)}$ is the local internal field acting on the $i-$th spin. The related configuration is therefore a fixed point, denoted as $\bm \sigma^{(\infty)}$. Then, we compute the Mattis magnetization $m_f:= m_{\mu}^{(\infty)} =\frac 1N \sum_{i=1}^N \xi^\mu _i \sigma_i ^{(\infty)}$. As for the initial conditions $\bb \sigma^{(0)}$, we have:
\begin{itemize}
    \item {\bfseries Stability analysis: } the initial condition $\bb \sigma^{(0)}$ is perfectly aligned with the pattern $\bb \xi^\mu$. For a given realization of the training set $\{\bb \xi^{\mu,A}\}_{\mu=1,...,K}^{A=1,...,M}$, we perform the stability analysis for each of the $K$ ground-truth $\bb \xi^\mu$, and in the end we average over 50 different realizations of the training set. The results, reported in Fig. \ref{fig:num_stability}, show the dependence on $d$ of the magnetization $m_f$ for various values of $K$, $M$ and $r$;
    \item {\bfseries Generalization analysis: } the initial condition is a validation example, that is, a configuration $\bb x=\bb \eta \odot\bb \xi^\mu $ with $\eta_i\sim_{i.i.d.} Rad(r)$. The number of different validation examples for each pattern is 500, then we perform the same analysis for all of the $K$ ground-truths, and we further average over $50$ different samples of these ground-truths. The results, reported in Fig. \ref{fig:num_generalization}, show the final magnetization $m_f$ as a function of $m_0 = \frac1N \sum_{i=1}^N x_i \xi^\mu_i$ for different values of $K$, $M$ and $d$.
\end{itemize}
\par\medskip
Let us start with the pattern stability, whose results are organized in Fig. \ref{fig:num_stability} as follows: the columns refer to different values of the number of classes in the training set ($K=50,250,500$, corresponding to $\alpha_N = 0.05,0.25,0.5$), while the rows refer to different values of the dataset quality ($r=0.6,0.9$). Each plot reports the results for $M=10,20,200$. By inspecting the plots, we see that, at low load ($\alpha=0.05$) the effect of the dilution is harmful at relatively low number of training points ($M=10,20$), as the final magnetization decays with $d$, but does not affect the stability of the patterns in the presence of large redundancy in the dataset ($M=200$). The situation is robust w.r.t. the quality parameter $r$, with the only quantitative differences in the final magnetization $m_f$. Increasing the load to $\alpha =0.25$ (above the Hopfield model critical threshold), the situation slightly changes, as -- especially for high-quality and large-size training datasets -- the stability of the ground-truths is enhanced by supplying the network with pretty diluted (that is, $d\gtrsim 0.4$) examples. Take for instance the case $M=200$: increasing the quality from $r=0.6$ (where ground-truths are always fixed point for the neural dynamics) to $r=0.9$, we see that the stability of the hidden patterns decreases (in average) for non-diluted training samples; in this case, introducing dilution is crucial to regain stability of the ground-truths. This is more evident by further increasing the load to $\alpha=0.5$: dealing with high-quality examples would require both large $M$ and dilution $d$ in order to ensure stability of the hidden patterns. Thus, these findings corroborate our expectations about the positive role of dilution at high $\alpha$ (and high $M$).

\begin{figure}[tb]
    \centering
   \includegraphics[width=\linewidth]{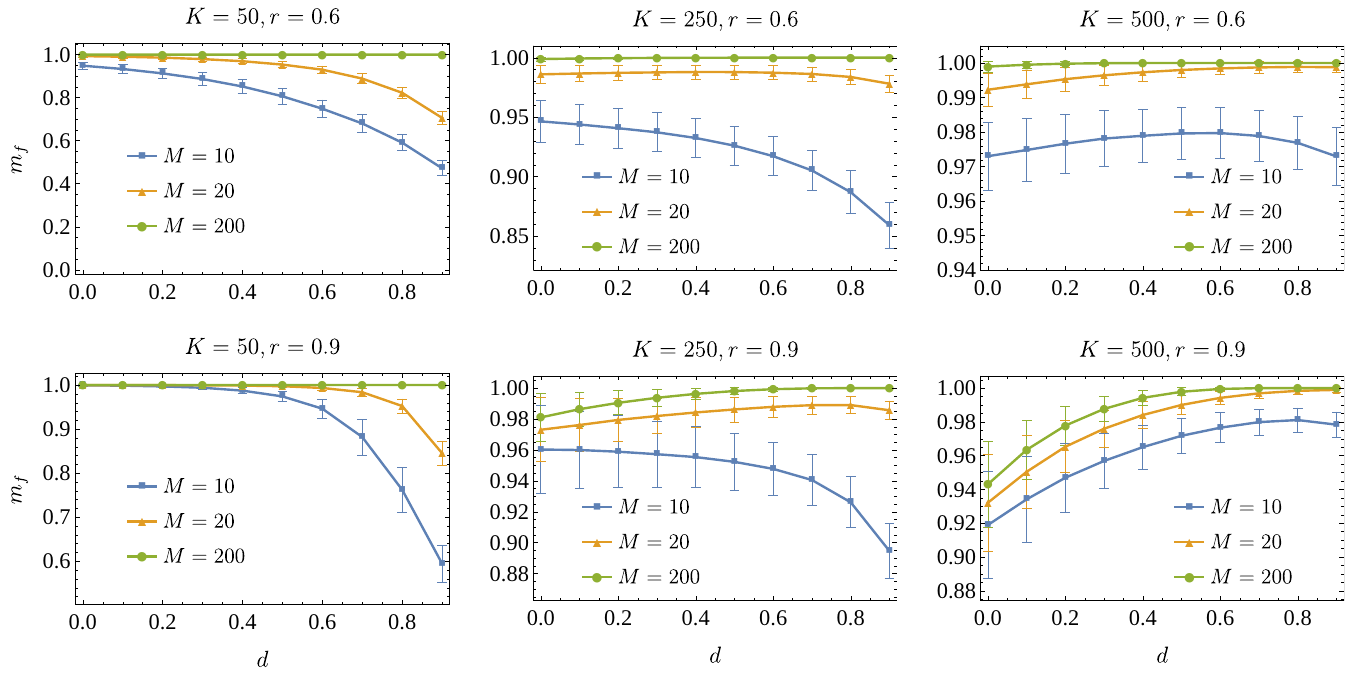}
    \caption{{\bfseries Pattern stability as a function of the dilution in the unsupervised setting.} The plots report the behavior of the final magnetization $m_f$ of the neural dynamics \eqref{eq:dynamics} with the initial condition being the one of the ground-truth ($\bb\sigma^{(0)}=\bb \xi^\mu$). We report the dependence of $m_f$ on the dilution parameter $d$ for various values of $K$ ($50$, left; $250$, center; $500$, right), the dataset quality $r$ ($0.6$, first row; $0.8$, second row; $0.9$, third row), and the number of examples per class $M$ ($10$, blue square dots; $20$, yellow triangles; $200$, green circles). The network size is fixed to $N=1000$. The results are averaged over 500 different realizations of the couplings matrix for each point.}
    \label{fig:num_stability}
\end{figure}
\begin{figure}[tb]
    \centering
    \includegraphics[width=0.95\linewidth]{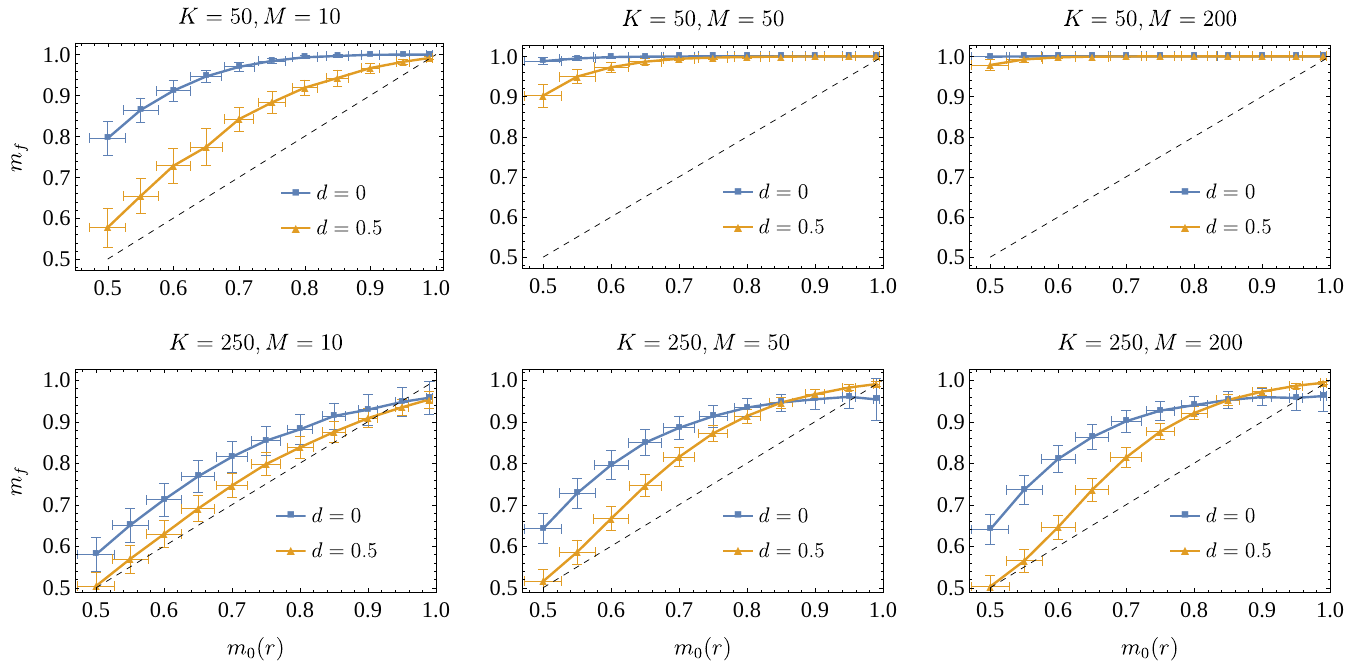}
    \caption{{\bfseries Generalization in the unsupervised setting}. The plots report the results of the pattern attractiveness, namely the final magnetization $m_f$ under the relaxation to fixed points of the neural dynamics \eqref{eq:dynamics} starting from a validation example (with initial magnetization $m_0 (r)$). We report the dependence of $m_f$ on the initial overlap $m_0 (r)$ for various values of $K$ ($50$, first row; $200$, second row; $250$, third row), the number of examples per class $M$ ($10$, left; $50$, center; $200$, right), and the dilution parameter $d$ ($0$, blue squares; $0.5$, yellow ones). The network size is fixed to $N=1000$. The results are averaged over 500 different realizations of the couplings matrix for each point.}
    \label{fig:num_generalization}
\end{figure}
\par\medskip
We now focus on the effects of dilution on the generalization capabilities of these networks. As we said above, the starting configuration is an example with the same quality of the training dataset, but {without} the dilution, as the latter is inserted by hand solely in the construction of the coupling matrix. We perform experiments analogous to the preceding ones, but now we consider the retrieval maps of the model by fixing the external parameters, that is we plot the final magnetization $m_f$ as a function of the starting correlation $m_0(r)$ with the pattern associated to the validation example. At low $K$, the situation is similar to the previous analysis: for low number of examples, the non-diluted network exhibits better generalization performances, and the two models have comparable pattern reconstruction capabilities in the large dataset case (in particular, the generalization performances are high, as $m_f\sim 1$ for $M=200$ both for $d=0$ and $d=0.5$). Turning to relatively large loads, a large $M$ and high dilution $d$ lead to the $m_f$ vs $m_0$ curves of the two models (with and without dilution) crossing. This indicates that there is a value of $r$ for which the generalization capabilities of the diluted model surpass those of the non-diluted one. 
The situation is evident in the case $K=250$ and $M=200$: there is a (relatively) wide region of the quality parameter where the final magnetization of the diluted model is close to 1, thus signaling that the system (almost) correctly reproduces the hidden archetype associated with that data point. For further details, we refer to App. \ref{app:mechanism}.
\par\medskip
As a final analysis, we fix the load $\alpha=0.1, 0.4$ and the number of examples per class $M=50,100,200$, span the values of $r \in [0.5,1]$ and $d \in [0,1]$ and check the effects of dilution on the generalization capabilities of the network after thermalization towards fixed points of the neural dynamics. To do this, we consider the ``network gain'' $\delta(d,r)$ of the final Mattis magnetization at dilution $d$ w.r.t. the case without dilution. In formula:
\begin{equation}
    \delta(d,r)= \frac{m_f (d,r)-m_f (0,r)}{m_f (0,r)}.
\end{equation}
The results are reported in Fig. \ref{fig:heat_maps}. In agreement with the analysis reported in Fig. \ref{fig:num_generalization}, we see that, at low load ($\alpha=0.1$) the effect of dilution on generalization capabilities is (in the best case) irrelevant, as the majority of points in the heat map exhibits zero gain in the final magnetization after the dilution is switched on. Moreover, at very poor quality dataset ($r\sim 0.5$), dilution results in a strong reduction of the reconstruction capabilities of the network, with a loss of at most $50\%$ w.r.t. non-diluted case. Thus, again, if the number of categories in the dataset is low, the introduction of blank entries in the training examples constitute an obstruction to pattern reconstruction, and this behavior is observed even increasing the number of examples per class. In contrast, the high-load case ($\alpha=0.4$) exhibits a qualitatively different situation. While in low quality dilution retains a harmful effect on generalization capabilities, for (very) high quality datasets (with $r \in [0.9,1.0]$) the presence of blank entries in the training examples leads to a gain of around $10\%$ in the final magnetization, and it improves with the number of points per class $M$. 

\begin{figure}[tb]
    \centering
    \includegraphics[width=0.95\linewidth]{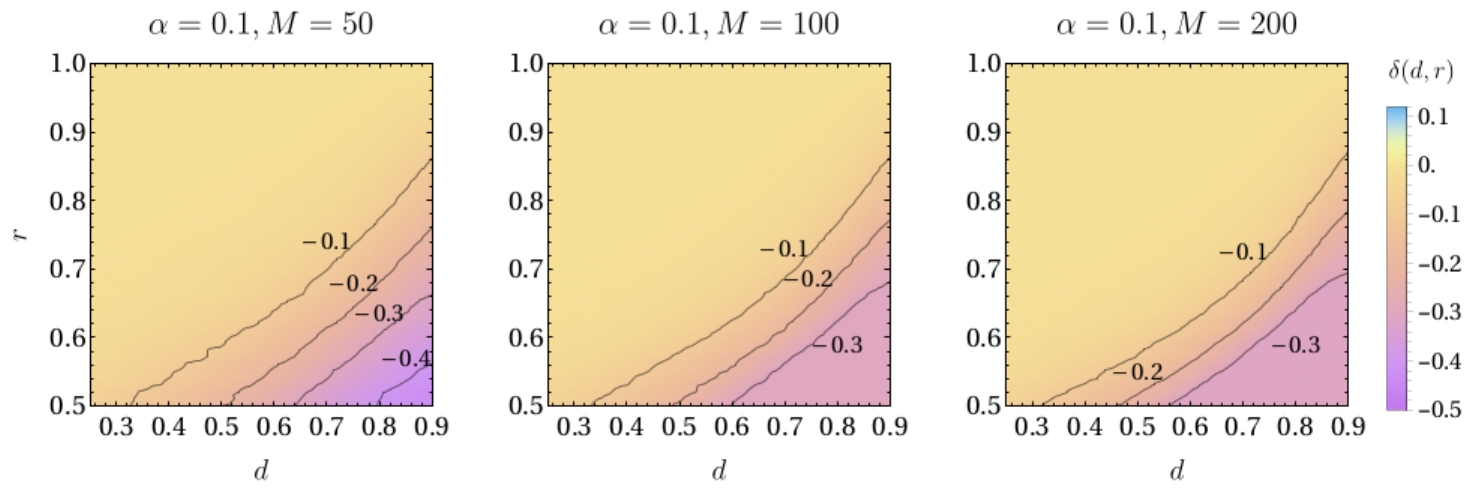}
    \includegraphics[width=0.95\linewidth]{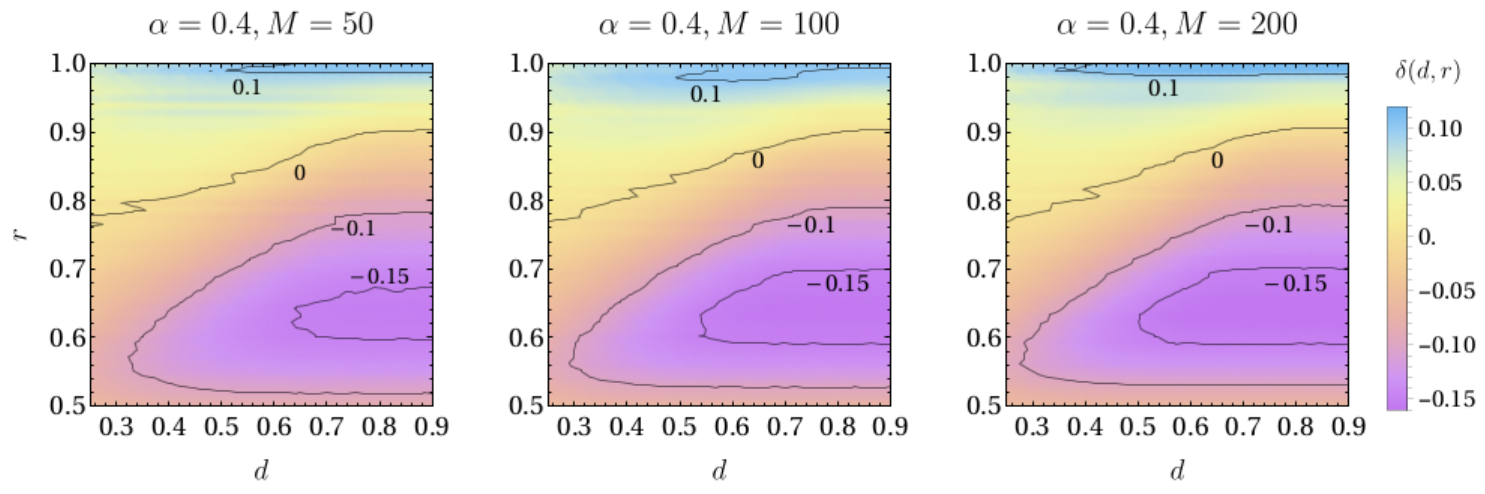}
    \caption{{\bfseries Generalization in the unsupervised setting in the $(r,d)$ plane}. The plots report the results of the pattern attractiveness, namely the final magnetization $m_f$ under the relaxation to fixed points of the neural dynamics \eqref{eq:dynamics} starting from a validation example. We report the dependence of $m_f$ on the dilution parameter $d$ and the dataset quality $r$, for $\alpha=0.1$ (first row) and $\alpha=0.4$ (second row), and $M=50$ (left column), $M=100$ (center column) and $M=200$ (right column). The black solid lines are the level curve in the heat maps. The network size is fixed to $N=1000$. The results are collected over 20 different realizations of the couplings matrix.}
    \label{fig:heat_maps}
\end{figure}

\par\medskip
\subsection{An experiment with a structured dataset}
As a last experiment, we consider the case of a structured dataset. Recalling that the benign effects of diluting the training examples is manifest at high load, we should take into account a dataset with a large number of classes, if compared to the native size of the patterns. In order to explore this regime with a structured dataset, we consider as ground-truths a mixed set of Chinese characters and Japanese ideograms (drawn from Hiragana, Katakana and Kanji characters). In this wide range of possibilities, we choose 250 characters and generate Boolean representations of the selected ideograms as $25\times25$ matrices, which will play the role of ground-truths. A sample of these patterns is reported in the first row of Fig. \ref{fig:ideogrammi_results}. In this way, $N=625$ and $K=250$, so that $\alpha=0.4$. 
\par\medskip

\begin{figure}[tb]
    \centering
     \includegraphics[width=0.9\textwidth]{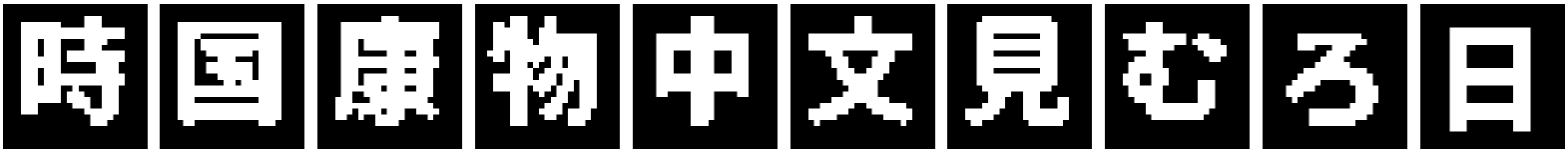}\vspace{0.2cm}\\
    \includegraphics[width=0.97\textwidth]{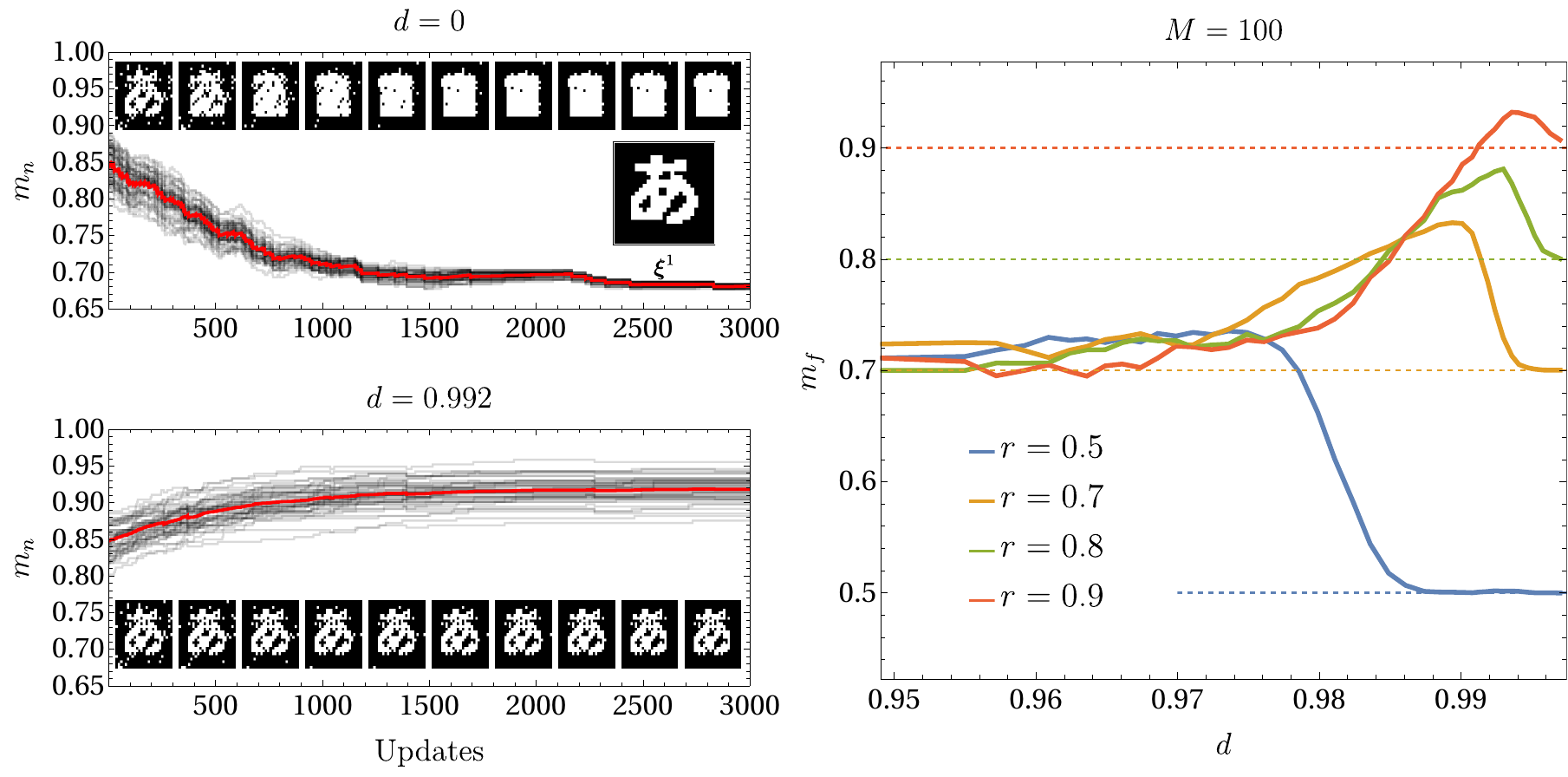}\\
        \includegraphics[width=\textwidth]{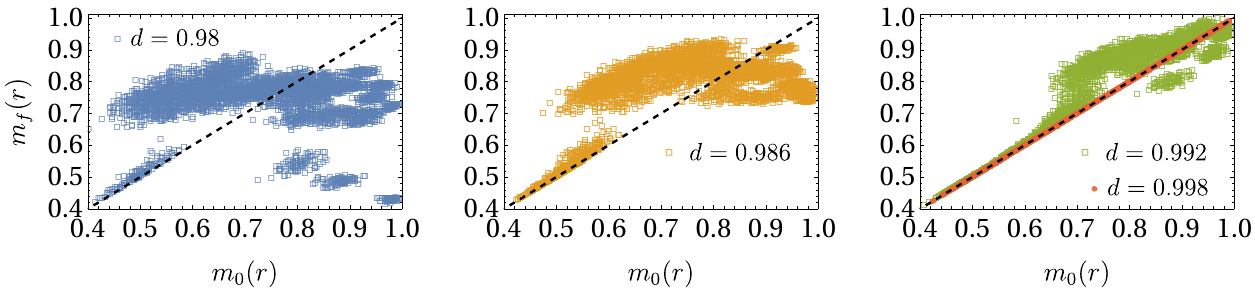}
    \caption{{\bfseries Summary of results for the unsupervised setting with structured patterns.} (First row) The figure shows a subset of the hidden patterns (Chinese character and Japanese ideograms) used to generate the training examples. (Second row, left plots) The plots show the relaxation of the system toward fixed points for $d=0$ (upper plot) and $d=0.992$ (lower plot). In both cases we input the {\it same} validation example (which is a corrupted version of the ground-truth $\bb \xi^1$ reported as inset in the upper plot). At each time step $n$, we compute the value of the Mattis magnetization $m_n= \frac1N \sum_i \sigma_i^{(n)} \xi^1_i$. We performed the same experiment for 50 different initial conditions, reporting the evolution of the magnetization in each run (gray solid lines) and the related average magnetization (red solid curve). We also collect snapshots of the network configurations at a distance of 300 updates (insets in both plots). The parameters are set as $N=625$, $K=250$, $r=0.85$, $M=100$. (Second row, right plot) The plot shows the dependence of the final magnetization $m_f$ as a function of $d$ for various values of the dataset quality $r$($=0.5,0.7,0.8,0.9$). The dashed lines are the threshold $m_f=r$, above which generalization is effective. The parameters are set as $N=625$, $K=250$, $M=100$. (Third row) Empirical retrieval maps for fixed values of $d$ varying the dataset quality $r$. Each point corresponds to a single retrieval experiment, where $m_0(r)$ is the initial overlap with the pattern and $m_f$ is the final magnetization. Respectively we have $d=0.98$ (left), $d=0.986$ (center), $d=0.992,0.998$ (right). The dashed black lines correspond to the trivial behavior, namely $m_f=m_0$. The parameters are set as $N=625$, $K=250$, $M=100$.}
    \label{fig:ideogrammi_results}
\end{figure}
Using these patterns, we generate a training dataset by corrupting them with a multiplicative Bernoulli noise with quality $r$ and dilution $d$. We begin by following neural dynamics of two networks, both with $M=100$ examples per class and quality $r=0.85$, and compare the cases with $d=0$ (no dilution) and $d=0.992$, see first column is second row in Fig. \ref{fig:ideogrammi_results}. Specifically, we provide the network with test examples and let the system evolve towards a fixed point with the dynamics \eqref{eq:dynamics}.\footnote{Unlike the synthetic random case, we here adopted a serial dynamics in order to better visualizing the evolution towards fixed points of the network in the structured scenario.} The Hebbian networks built with non-diluted and diluted training examples are fed with the same initial conditions, computing at each step the magnetization corresponding to the associated pattern. In the undiluted case, the magnetization of the network configuration with the corresponding ground-truth decreases over time, and starting from an initial configuration with $m_0=r$ (in average) at step 0, the system relaxes toward a final state with lower correlation with the ground-truth (in particular, it settles into a spurious combination of patterns with high attractive power). This behavior is confirmed by looking at the snapshots of the network configuration: the system loses the structure of the test example, and relaxes towards a final configuration characterized by a large ``blob'' of active pixels. We confirm that this is a strongly attractive configuration even by varying $M$ and $r$, observing qualitatively similar outcomes. Nevertheless, by inserting a strong dilution in the training dataset, the network becomes capable of producing a more faithful representation of the pattern, and the overlap with the latter and the network configuration increases during neural dynamics. As shown in the snapshot, the network filters out the noise in non-informative pixels (the dark background) while preserving the structure of the target, thus resulting in a better reconstruction of that pattern. This experiment highlights that a strong dilution is responsible for reducing the intrinsic noise in the update rule \eqref{eq:dynamics} and leading to better reconstruction performances of the ground-truths underlying the training dataset.\par\medskip
In the right column of the second row of Fig. \ref{fig:ideogrammi_results}, we fix the number of examples per class to $M=100$ and study the final magnetization $m_f$ as a function of $d$ when the network is fed with a test example for various values of the quality $r$. The plot reveals that, as along as the dilution is lower than a certain threshold ($\bar d\approx 0.975$), the neural dynamics tends to drive the system towards a spurious combination of the patterns (this is evidenced by the plateau below $d=\bar d$, which matches the final overlap of the spurious state shown in dynamical experiments without dilution). As dilution increases beyond this threshold, the behavior of the model becomes dependent on the dataset quality. For low $r$ (e.g. $r=0.5$), the final overlap settles on $m_f=r$, indicating a trivialization of the neural dynamics in which any configuration can become fixed points. Conversely, for higher quality, there exists an intermediate range of dilution where the model exhibits enhanced reconstruction capabilities: here, the system is more robust against relaxation toward spurious combinations and the final magnetization is higher than the quality of the training dataset ($m_f>r$), signaling non-trivial pattern retrieval. It is worth noticing that this convenient region however shrinks as $r$ increases. Within that range, we computed numerically the proximity to the diagonal structure $\pi(d)$ and the ratio between off-diagonal and diagonal entries, $R(d)$ -- see App. \ref{app:mechanism}; for all values of $r$, we find $\pi (d)\sim 0.1$ and $R(d) \sim 1$, indicating that -- even for extreme dilution in the training dataset -- the coupling matrix remains far from being trivial. Finally, for extreme value of dilution, regardless of $r$, neural dynamics again trivialize, as $m_f\approx r$. 

As a complementary analysis, we fix the dilution and numerically investigate the retrieval maps varying the quality. The results are reported in the third row of Fig. \ref{fig:ideogrammi_results} for $d=0.98,0.986,0.992,0.998$. Again, these plots reveal that there exists a range in the quality (or, equivalently, in $m_0(r)$ since the initial condition is a test example) where dilution is beneficial, and leads to a consistent improvement of the pattern reconstruction capabilities. Only for extreme dilution (namely $d=0.998$) neural dynamics trivialize, as $m_f= m_0(r)$, independently on the starting configuration.

\section{Discussion and outlooks} \label{sec:final}
We investigated the generalization capabilities of Hebbian neural networks built on noisy and sparse datasets, facing both supervised and unsupervised settings. Using random-matrix-theory tools, we derived the asymptotic spectral distribution of the coupling matrices constructed upon synthetic data and this allowed us to analytically investigate the 1-step Mattis magnetization with respect to ground-truth patterns.
Our findings reveal that, in supervised settings, data dilution generally impairs the generalization performance. In contrast, in unsupervised settings,
the presence of blank or missing entries in the data can, under certain conditions, improve the performance and this positive effect of dilution is detected especially when the data quality is high and the load is large. 
This outcome can be understood by directly looking at the impact of sparsification on the 1-step Mattis magnetization and at how sparsification downsizes the magnitude of noisy contributions hidden in neuron-neuron interactions without trivializing the relaxation towards fixed points. Overall, our findings show that dilution in training data reshape the energy landscape of the model, possibly making fixed points lie closer to the ground-truths compared to the undiluted case, improving generalization capabilities. 
Our theoretical results are supported by extensive numerical simulations which also demonstrate the robustness of the phenomenology when structured datasets are used. 
This evidence suggests sparsification strategies, where selective dilution may enhance generalization without sacrificing too much information.
%
Further, the close connection between Hebbian networks designed with diluted datasets and the dropout procedure within gradient-based training algorithm of associative memories highlights the relevance of diluting the available information in a more general context. 

\appendix

\section{Linking dropout and diluted datasets: a Machine Learning perspective} \label{app:dropout}
In this appendix, we aim to provide a relation between the operation of diluting the training dataset and the dropout technique \cite{srivastava2014dropout} designed to avoid overfitting when training neural networks. To this aim, we revisit Hebb's rule from a Machine Learning perspective. As is well-known \cite{Amit}, Hebbian prescription is designed to capture spin-spin correlation within a given set of observation. Thus, we consider a sample of vectors $\mathcal D= \{\hat {\bb \xi}^s\}_{s=1}^P$, where $s$ can label independent random vectors as in the storing case (thus $P=K$), or refer to training examples in an unsupervised setting (with $P=MK$). As we are interested mainly in the empirical setting, without loss of generality we will only focus on the unsupervised setting and, since we intend to insert dilution as an external operation, we will work with undiluted training examples, namely $\hat \xi^{\mu,A}_i = \hat \chi^{\mu,A}_i \xi^{\mu}_i$ with $P(\hat \chi^{\mu,A}_i =\pm1)=\frac{1\pm r}{2}$. The most natural loss-function to achieve the Hebb goal reads as
$$
\mathcal L (\bb J) = \frac1{2KM} \sum_{\mu,A=1}^{K,M} \lVert \bb J - \hat {\bb \xi}^{\mu,A} (\hat {\bb\xi}^{\mu,A})^T\lVert_{UF}^2=  \frac1{2KM} \sum_{\mu,A=1}^{K,M}\sum_{i,j=1}^N ( J_{i,j} -\hat \xi^{\mu,A}_i \hat \xi^{\mu,A}_j )^2,
$$
with $\lVert \cdot \lVert _{UF}$ being the (unnormalized) Frobenius norm. The optimization task can be solve by means of gradient flow:
$$
\dot J_{i,j}= - \partial_{J_{i,j}} \mathcal L (\bb J) =\frac1{KM} \sum_{\mu,A=1}^{K,M}\hat  \xi^{\mu,A}_i \hat \xi^{\mu,A}_j -J_{i,j}.
$$
The first contribution on the r.h.s. is the empirical expectation of the spin-spin correlation evaluated on the training set $\mc D$, while the second term plays the role of a regularization term (yielding to exponential decay of the network weights). Clearly, the fixed point of the gradient flow leads immediately to $\bb J = (KM)^{-1} \bb X^{\text{ex}} (\bb X^{\text{ex}})^T$,\footnote{Notice that the choice of the loss function as the empirical squared error leads to a different normalization factor in the definition of the Hebbian coupling matrix. This is however not a problem, since we only deals with zero-temperature dynamics. In the statistical-mechanical setting, this would only lead to a trivial rescaling of the temperature in the high-storage regime.} where (following the notation adopted in Subsec. \ref{sec:2.2}) $\bb X^{\text{ex}}$ is the $N\times KM$ matrix with the training examples on the columns. For our concerns, it will be helpful to consider the discretized version of the gradient flow, namely
$$
\Delta J_{ij}^{(k)}= - \epsilon \partial_{J_{i,j}}\mathcal L (\bb J) =  \frac\epsilon{KM} \sum_{\mu,A=1}^{K,M} \hat \xi^{\mu,A}_i \hat \xi^{\mu,A}_j -\epsilon J_{i,j}^{(k-1)},
$$
with $\Delta J_{ij}^{(k)}$ being the update of the network weights at training time $k$, and $\epsilon$ is the learning rate. Because of the simplicity of the training equations, we can recast everything in terms of the coupling matrix itself at time $k$. Specifically:
\begin{align*}
  J_{i,j}^{(k)}&= \frac{\epsilon}{KM} \sum_{\mu,A=1}^{K,M}\hat  \xi^{\mu,A}_i \hat \xi^{\mu,A}_j+(1-\epsilon) J_{i,j}^{(k-1)}= \frac{\epsilon}{KM} \sum_{s=0}^{k-1}(1-\epsilon)^s\cdot\sum_{\mu,A=1}^{K,M}\hat  \xi^{\mu,A}_i\hat  \xi^{\mu,A}_j+(1-\epsilon)^k J_{i,j}^{(0)}=\\
  &=\frac{1}{KM} [1-(1-\epsilon)^k]\sum_{\mu,A=1}^{K,M}\hat  \xi^{\mu,A}_i \hat \xi^{\mu,A}_j+(1-\epsilon)^k J_{i,j}^{(0)}.
\end{align*}
Since $0<\epsilon<1$, taking the $k\to\infty$ limit immediately leads to $J^{(\infty)}_{i,j}= \frac1{KM}\sum_{\mu,A}\hat  \xi^{\mu,A}_i\hat  \xi^{\mu,A}_j$, regardless of the initial condition $\bb J^{(0)}$. Thus, the Hebbian coupling matrix is the fixed point of the gradient descent algorithm with loss function $\mathcal L$.
\par\medskip
In order to link the dilution of training examples with the dropout technique, we slightly modify the previous algorithm implementing the following requirements: {\it i}) we present a single training example per update, and the whole training dataset is presented just once (thus, the total number of updates will be $k_{\text{max}}=KM$); {\it ii}) we use a cooling schedule $\epsilon = \epsilon_k= \frac{\epsilon_0}{1+\epsilon_0 k}$ (this is just for computational convenience, however has turned out to be effective in unlearning schemes, see for example \cite{fachechi2019dreaming}). With these ingredients, the training equations become
\begin{align*}
     J_{i,j}^{(k)}&=\big(1-\frac{\epsilon_0 }{1+\epsilon_0 k}\big)J_{i,j}^{(k-1)}+\frac{\epsilon_0}{1+\epsilon_0 k}\hat\xi^k_i\hat \xi^k_j,
\end{align*}
with $k=1,\dots,KM$ and using $\hat \xi^k_i$ instead of $\hat \xi^{\mu,A}_i$ for the sake of clarity.\footnote{We stress that presenting a single example per update is an extreme version of the mini-batch gradient descent, where spin-spin correlations are computed point-wisely: thus, we do not normalize the updates by $(KM)^{-1}$.} Expressing again the $k$-th value of the coupling matrix in terms of the initial condition and setting $k= k_{\text{max}}=KM$, we have
\begin{align*}
    J_{i,j}^{(\infty)}&=J_{i,j}^{(KM)}=\frac{\epsilon_0 }{1+\epsilon_0 KM} J_{i,j}^{(0)}+\frac{\epsilon_0}{1+\epsilon_0 KM}\sum_{\mu,A=1}^{K,M}\hat\xi^{\mu,A}_i\hat \xi^{\mu,A}_j.
\end{align*}
Now, if we set $\bb J^{(0)} =0$, $\epsilon_0 =\frac{\gamma}{KM}$ with $1\ll \gamma \ll KM$, we have $J_{i,j}^{(\infty)}\approx \frac1{KM}\sum_{\mu,A=1}^{K,M}\hat\xi^{\mu,A}_i\hat \xi^{\mu,A}_j $. The reason for choosing this version of the gradient descent is that dropout can be easily implemented in this setting. To do this, we apply the silencing mechanism {\it only} on the data-dependent part of the algorithm, as the contribution involving the $J_{i,j}^{(k)}$ will only cause a (slow) decay of the network weights. Furthermore, we assume that the diagonal entries of the coupling matrix remain set to zero throughout the entire training procedure. Thus, calling $\theta^{(k)}_{i,j}$ a set of i.i.d. $\text{Ber}(p)$, the dropout gradient descent reads as
\begin{align*}
     J_{i,j}^{(k)}&=\big(1-\frac{\epsilon_0 }{1+\epsilon_0 k}\big)J_{i,j}^{(k-1)}+\frac{\epsilon_0}{1+\epsilon_0 k}\theta^{(k)}_{i,j}\hat\xi^k_i\hat \xi^k_j.
\end{align*}
In order to maintain the symmetry of the coupling matrix,\footnote{This is crucial for ensuring the detailed balance principle to hold.} the most natural choice is that assume that the $\theta^{(k)}_{i,j}$ are factorized in the product of two $\text{Ber}(1-d)$, namely $\theta^{(k)}_{i,j}=\theta^{(k)}_{i}\theta^{(k)}_{j}$ with $P(\theta^{(k)}_i =0)=d=1-\sqrt{p}$. In this way, $\theta^{(k)}_{i,j}\hat\xi^k_i\hat \xi^k_j= (\theta^{(k)}_{i}\hat\xi^k_i)(\theta^{(k)}_{j}\hat \xi^k_j)= \xi^k_i \xi^k_j$, where we defined $\xi^k_i = \theta^{(k)}_i \hat \xi^k_i$. Along the same lines, the final value of the coupling matrix will be $J_{i,j}^{(\infty)}\approx \frac1{KM}\sum_{\mu,A=1}^{K,M}\xi^{\mu,A}_i \xi^{\mu,A}_j $, where the stored vectors are now the $\bb\xi^{\mu,A}$. Recalling that $\hat {\bb \xi}^{\mu,A}=\hat \chi^{\mu,A}_i \xi^\mu_i$ and that, in this setting $k\equiv (\mu,A)$, we have $\chi^{\mu,A}_i = \theta^{\mu,A}_i \hat \chi^{\mu,A}_i \xi^\mu_i=\chi^{\mu,A}_i \xi^\mu_i   $ upon identifying $\chi^{\mu,A}_i= \theta^{\mu,A}_i \hat \chi^{\mu,A}_i$. It is now straightforward to notice that the $\chi^{\mu,A}_i$ have the same statistics given in Eq. \eqref{eq:P_chi} used for defining the diluted training examples. Hence, the relation between the dropout and the framework defined in Sec. \ref{sec:2.2} is finally fulfilled.
\par
\medskip
\begin{figure}[h!]
    \centering  \includegraphics[width=0.8\textwidth]{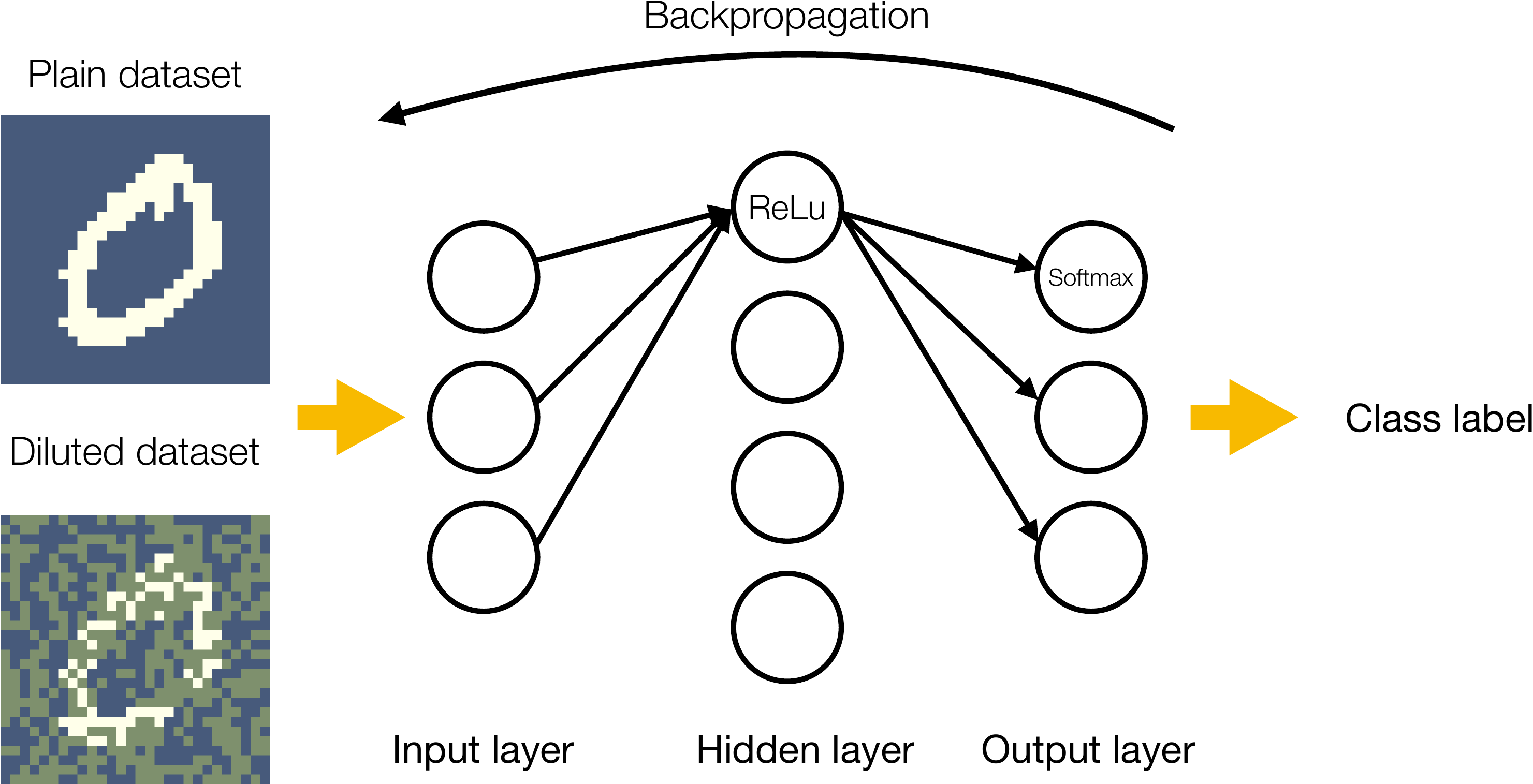}
    \caption{{\bfseries Schematic representation of the experiment with feed-forward neural network with (non-)diluted datasets.}}
    \label{fig:FFN_dilution_scheme}
\end{figure}
In order to stress the functioning of the dataset dilution as a dropout-like procedure, we train a feed-forward neural network with a single hidden layer with varying number of hidden neurons. We consider the MNIST dataset, binarized using local adaptive thresholding to emulate the Boolean nature of patterns. Dilution is introduced to the 60000 training images by applying a random (but fixed for each example) silencing mask setting a fraction $d$ of pixels to zero. We compare the non-diluted and $d=0.5$ case. The feed-forward neural networks consists in a input layer with 784 neurons, $h=10,20,\dots,50$ hidden neurons with ReLU activation function in the hidden layer, and an output one with Softmax activation function, see Fig. \ref{fig:FFN_dilution_scheme} for a schematic representation. Training is performed with backpropagation algorithm with categorical cross-entropy as loss function for 50 epochs with mini-batches of size 128. Importantly, evaluation is carried out on the non-diluted version of the training set, meaning that dilution is applied only during training—similarly to standard dropout. Test performance is computed on the original, unaltered test set. The resulting loss and accuracy values are shown in Fig.~\ref{fig:FFN_dilution}.

\begin{figure}[h!]
    \centering  \includegraphics[width=\textwidth]{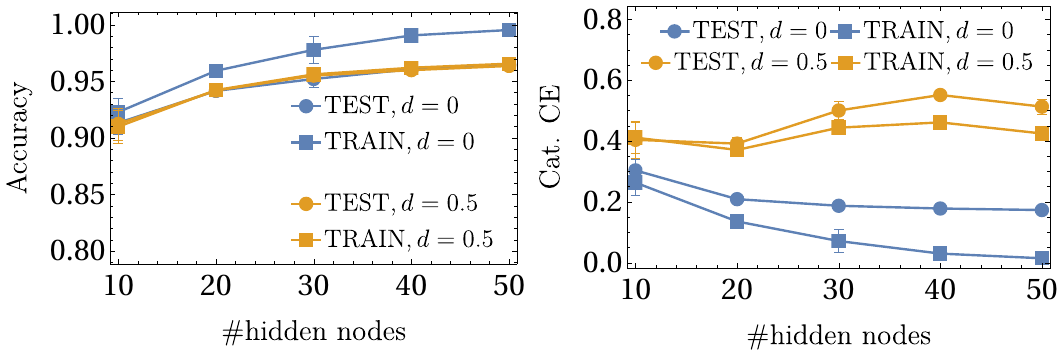}
    \caption{{\bfseries Numerical results for the training procedure of the feed-forward neural network.} The plots show the accuracies (left plot) and categorical cross-entropy (right plot) for training and validation datasets (resp. squares and circles) without dilution (blue curves) or with moderately diluted dataset (orange curves), as a function of the number of hidden neurons. The FFN is trained with backpropagation algorithm with categorical cross-entropy loss function for 50 epochs with minibatches of size 128. The training accuracy and loss are computed without dilution, as is done for the dropout technique.}
    \label{fig:FFN_dilution}
\end{figure}

 For the non-diluted case (blue curves) it is straightforward to notice that, as the number of hidden neurons (and, consequently, the trainable weights) the training accuracy (left plot) increases, and reaches values close to 1 for $h=50$, while the validation one attains its maximal value around 0.95. On the other hand, the training loss decreases to zero as the hidden layer gets wider and wider, while the validation one settles around a constant value of $\sim 0.2$. This signals an overspecialization of the system on the training dataset. Instead, introducing a moderate dilution ($d=0.5$) in the training dataset changes the qualitative behavior of the learning procedure: the training and validation accuracies overlaps for any width of the hidden layer, and the losses (despite being higher than the non-diluted case) are always comparable. Then, as with dropout, diluting the training dataset prevents overfitting. More precisely, during training the network still overspecialize on the training dataset, but -- in presence of dilution -- the latter carries less information w.r.t. the non-diluted case, so that when presenting complete data the network response is uniform, regardless if they are training or testing examples. Despite being operatively different, dropout and dataset dilution leads to similar mechanisms for preventing overfitting. In particular, the latter is more flexible, as targeted silencing strategies can be carried out in order to increase the performances of the model.
\section{Proofs} \label{app:proofs}
In this appendix, we collect the proofs of the propositions used in the main text. 
We start recalling the Marchenko-Pastur Theorem, see \cite{bai2010spectral} for a proof.
\begin{Theorem}[Marchenko-Pastur Theorem] \label{marchenkopasturtheorem}
Let $\bb X = \bb X^{(N)}$ be a sequence $N \times K$ real matrices, with $\lim_{N\to\infty}K/N = \alpha \in (0, \infty)$, and such that, for each $N$, its entries $ X^{(N)}_{i,\mu}$ are i.i.d. centered random variables with variance $\sigma^2<+\infty$, and $\exists \{\eta_N\}_N$ such that
\begin{equation} \label{boundconditions}
    \eta_N \to 0 \quad \text{and} \quad  X^{(N)}_{i,\mu} \leq \eta_N \sqrt{N} \text{, } \forall \text{ } N, i, \mu.
\end{equation}
Let $\bb S_N = \frac{1}{K} \bb X^{(N)} \bb X^{(N)^T}$ and $\mu_{\bb S}$ the empirical distribution of its eigenvalues:
\begin{equation} \label{empiricaldistribution}
    \mu_{\bb S}(x) = \frac{1}{N} \# \{j \leq N : \lambda^{\bb S}_j < x \}.
\end{equation}
Then, for $N \to \infty$,  $\mu_{\bb S}$ converges to a distribution $MP(\alpha, \sigma )$ with probability measure
\begin{align*}
        \mu_{MP}(A) =\begin{cases} (1-\alpha) \mathbf{1}_{0\in A} + \alpha \mu_{bulk}(A),& \text{if } 0\leq \alpha < 1\\
\alpha\mu_{bulk}(A),& \text{if } \alpha \geq 1 \end{cases},
    \end{align*}
    where
    \begin{align*}
        d\mu_{bulk}(x) = \frac{1}{2\pi \sigma^2 } \frac{\sqrt{(\lambda_{+} - x)(x - \lambda_{-})}}{x} \,\mathbf{1}_{x\in[\lambda_{-}, \lambda_{+}]}\, \mathrm d x,
    \end{align*}
    and $\lambda_\pm = \sigma^2 (1 \pm \sqrt{1/\alpha})^2$.
\end{Theorem}
In the following we will always retain a setting $\alpha <1$ and we will say that a random variable $Z$ has a \textit{shifted} (or modified) Marchenko-Pastur law $MP(\alpha, \sigma, s)$ with shift $s$ iff $Z-s \sim MP(\alpha, \sigma)$.

\subsection{Proof of Prop. \ref{prop:allspectra}} \label{app:proof_1}
\begin{proof}[Proof of Prop. \ref{prop:allspectra}, first point]
    Recalling $\boldsymbol{\Gamma} = \frac{1}{N} \bb X \bb X^T$, where $\bb X$ is a $N \times K$ random matrix such that $\mathbf X_{i,\mu} = \xi^\mu_i$ are i.i.d. $Rad(0)$ random variables, and $K/N \to \alpha \in (0,1)$, we can write $\boldsymbol{\Gamma} = \frac{1}{K} \sqrt\alpha \bb X \sqrt\alpha \bb X^T$. Then, $\boldsymbol{\Gamma}$ satisfies the hypothesis of Marchenko-Pastur theorem, since $\mathbb{E}[\sqrt\alpha \xi^\mu_i] = 0$ and $ \mathbb{E}[(\sqrt\alpha \xi^\mu_i)^2] = \alpha,$ hence the empirical distribution of the eigenvalues of $\boldsymbol{\Gamma}$ converges to a $MP\left(\alpha, \alpha ,0\right)$.
\end{proof}
In Fig. \ref{hebbbasic_spectrum}, we reported the comparison between the empirical spectral distributions for various values of $\alpha$ and its limiting distribution. 
\begin{figure}[h!]
    \centering  \includegraphics[width=\textwidth]{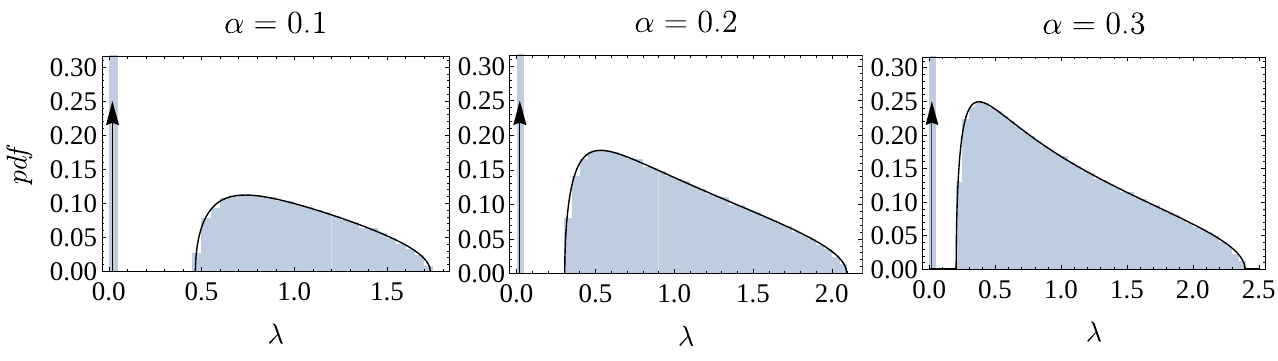}
    \caption{Histogram of the spectral distribution of $\boldsymbol \Gamma$ when $N=1000$, for different values of $\alpha$, against the analytical formula obtained by Prop. \ref{prop:allspectra}. The empirical distribution is realized by collecting the eigenvalues of a sample of 100 different realizations of the Hebbian matrices.}
    \label{hebbbasic_spectrum}
\end{figure}
As also remarked in the main text, also the supervised setting falls into the class of covariance matrices of Wishart form (with a suitable redefinition of the scale parameter in the Marchenko-Pastur distribution), thus the asymptotic spectral distribution of the supervised setting follows the same lines as the basic storing one:
\begin{proof}[Proof of Prop. \ref{prop:allspectra}, second point]
    Recalling $\boldsymbol{\Gamma}^s = \frac{1}{N} \Bar{\bb X} \Bar{\bb X}^T$, where $\Bar{\bb X}$ is a $N \times K$ random matrix such that $\bar {\mathbf X}_{i,\mu} =\bar \xi^\mu_i$ (see Eq. \ref{supervisedhebbmatrix}), we can write $\boldsymbol{\Gamma}^s = \frac{1}{K} \sqrt\alpha \Bar{\bb X} \sqrt\alpha \Bar {\bb X}^T$. Then, $\boldsymbol{\Gamma}^s$ satisfies the hypothesis of Marchenko-Pastur theorem, since $\mathbb{E}[\sqrt\alpha \xi^\mu_i] = 0$ and $ \mathbb{E}[(\sqrt\alpha \xi^\mu_i)^2] = \alpha\sigma^s,$ with
    $$
    \sigma^s = \sigma^s(r,d,M) =(1-d) \Big((1-d) r^2 + \frac{1-(1-d)r^2}{M}\Big),
    $$
    hence the empirical distribution of the eigenvalues of $\boldsymbol{\Gamma}^s$ converges to a $MP\left(\alpha, \alpha\sigma^s ,0\right)$.
\end{proof}
In Fig. \ref{hebbsupervised_spectrum}, we again provide a comparison between the empirical spectral distributions and their limiting ones for various values of the control parameters $r$ and $d$ at $\alpha=0.1$ and $M=50$. As is clear, the theoretical distribution is consistent with the numerical results. In both cases, the theoretical results are robust even for relatively small sizes of the networks ($N\sim 10^2$).
\begin{figure}[h!]
    \centering  \includegraphics[width=\textwidth]{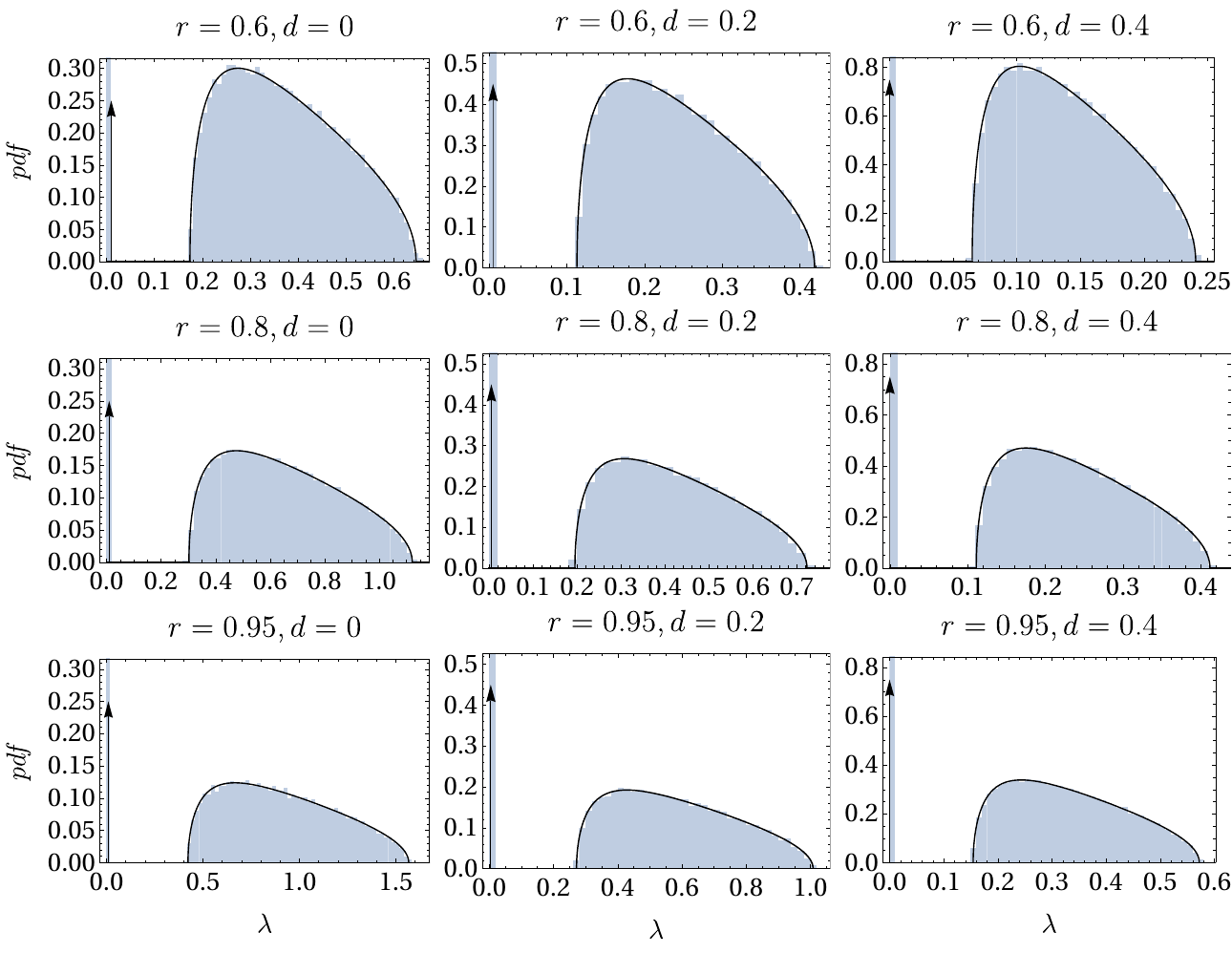}
    \caption{Histogram of the spectral distribution of $\boldsymbol \Gamma ^s$ when $N=1000$, $\alpha=0.1$ and $M=50$, for different values of $r$ and $d$. The empirical distribution is realized by collecting the eigenvalues of a sample of 100 different realizations of the Hebbian matrices.}
    \label{hebbsupervised_spectrum}
\end{figure}

\begin{proof}[Proof of Prop. \ref{prop:allspectra}, third point]
   Recall that the explicit expression of the Hebbian kernel in the unsupervised setting reads as
    \begin{equation} \label{eq:unsupunsup}
            \Gamma^u_{i,j} = \frac{1}{N} \sum_{\mu} \Big(\frac{1}{M}\sum_{A}\chi^{\mu,A}_i \chi^{\mu,A}_j\Big) \xi^\mu_i \xi^\mu_j.
    \end{equation}
    By virtue of the AFM (see App. \ref{afm}) and recalling that $(\xi^\mu_i )^2=1$, 
we can write (see also Eq. \eqref{eq:eqafm})
\begin{align*}
    \Gamma^u_{i,j} &= 
\frac{1}{N} \sum_{\mu} \phi^\mu_i \phi^\mu_j \xi^\mu_i \xi^\mu_j + \delta_{ij} \frac{1}{N}\sum_\mu \Big(\frac{1}{M}\sum_{A}(\chi^{\mu,A}_i)^2-(\phi^{\mu}_i)^2\Big),
\end{align*}
which is again the sum of a Wishart Matrix and a diagonal one. Using the strong law of large numbers, the diagonal can be evaluated straightforwardly in the thermodynamic limit, and it converges to 
\begin{equation*}
    \frac{1}{N}\sum_\mu \Big(\frac{1}{M}\sum_{A}(\chi^{\mu,A}_i)^2-(\phi^{\mu}_i)^2\Big) \to \alpha (1-d-\sigma^u(r,d,M)),
\end{equation*}
with $\sigma^u(r,d,M) = \mathbb{E}[(\phi^{\mu}_i)^2]$. Proceedings as in App. \ref{afm}, we can straightforwardly evaluate
\begin{align*}
    \sigma^u(r,d,M) &= \mathbb{E}[(\phi^{\mu}_i)^2] = \sqrt{\mathbb{E}[(\phi^{\mu}_i \phi^\mu_j)^2]}
    = \sqrt{\mathbb{E}\Big[\Big(\frac{1}{M}\sum^M_{A=1} \chi^{\mu, A}_i \chi^{\mu, A}_j\Big)^2\Big]} =\\
    &= \sqrt{\frac{1}{M^2} \Big(M \mathbb{E}[(\chi^{\mu, A}_i)^2] + M(M-1)(\mathbb{E}\chi^{\mu, A}_i)^2\Big)}=\\
    &= \sqrt{(1-d)^4r^4+(1-d)^2\frac{1-(1-d)^2r^4}{M}}.
\end{align*}
As for the contribution involving the Wishart matrix, we notice again that its factors are centered and with finite variance, hence its easy to show that the limiting spectral distribution is a $MP\left(\alpha, \alpha \sigma^u(r,d,M)\right)$. Putting all pieces together, the limiting spectral distribution of the Hebbian kernel in the unsupervised setting (within the AFM) is finally $MP(\alpha,\alpha\sigma^u (r,d,M), \alpha(1-d-\sigma^u(r,d,M)),$ which concludes the proof.
\end{proof}
A comparison between the empirical spectral distributions of the Hebbian kernel in the unsupervised setting and the limiting one derived by virtue of the AFM is shown in Fig.~ \ref{hebbunsupervised_spectrum}. Some comments are in order here: first, even at relatively low number of examples per class $M$, there is a substantial agreement of the primary bulk of the limiting spectral distribution with the numerical results; clearly, again the secondary bulk -- also present in the Gaussian scenario -- is still approximated with a $\delta$-peak, although the theoretical results accounts for its location. Further, increasing the dilution parameter $d$ at fixed $M$ (left to right in each plot), the AFM provides a less accurate predictions for the empirical spectral distribution (this is however not surprising, as high values of the dilution results in a lower effective number of examples contributing in the random sums $\frac1M \sum_{A=1}^M \chi^{\mu,A}_i \chi ^{\mu,A}_j$). Despite this, in general the theoretical results given by AFM provides a good starting point to investigate the algebraic properties of the Hebbian kernel in the unsupervised setting, and hence the retrieval capabilities of the related Hopfield network. 
\begin{figure}[h!]
    \centering  \includegraphics[width=\textwidth]{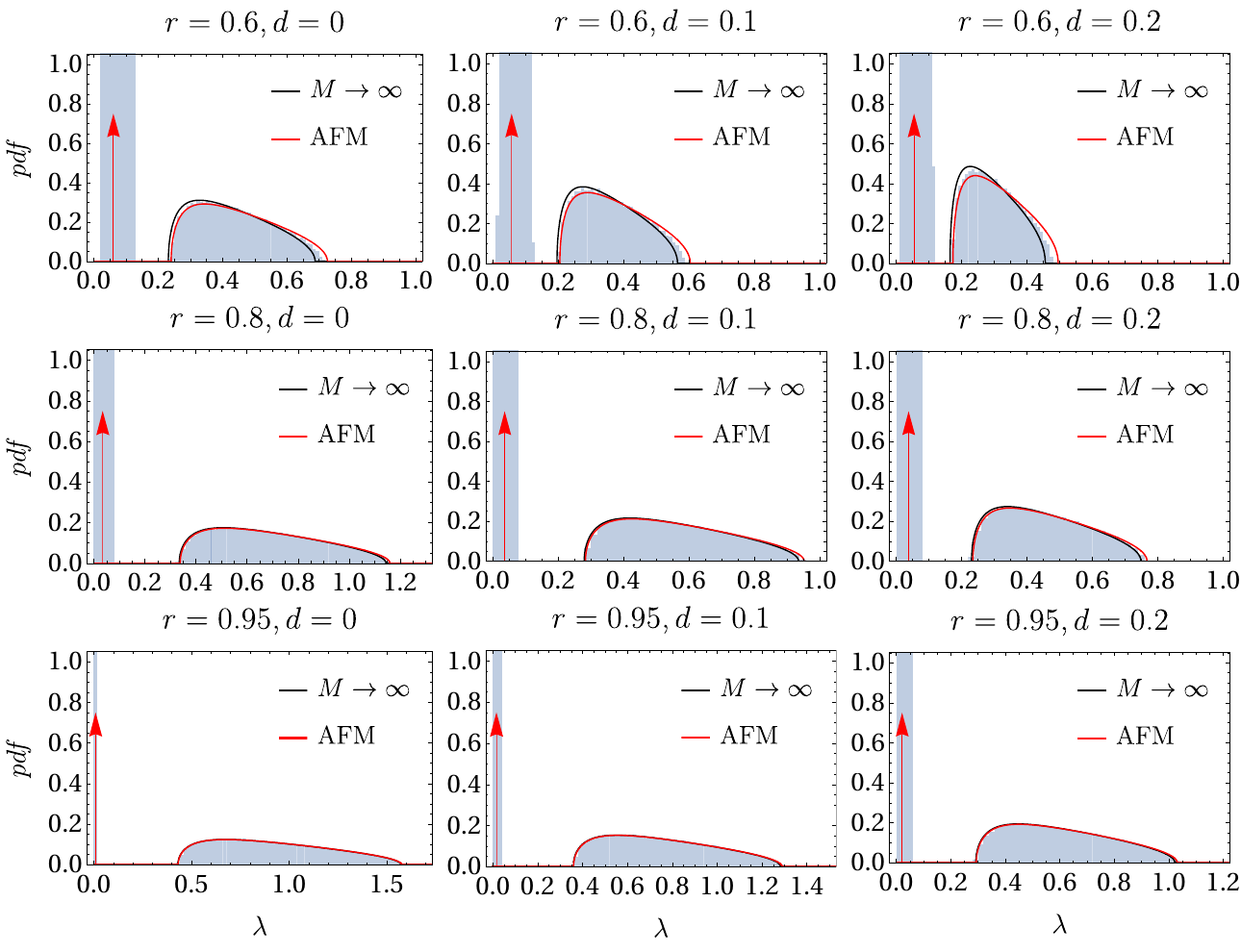}
    \caption{Histogram of the spectral distribution of $\boldsymbol \Gamma ^u$ when $N=1000$, $\alpha=0.1$ and $M=50$, for different values of $r$ and $d$. The empirical distribution is realized by collecting the eigenvalues of a sample of 200 different realizations of the Hebbian matrices. For the theoretical distributions, we only plotted the location of the $\delta$ peak for the AFM approximation.}
    \label{hebbunsupervised_spectrum}
\end{figure}
\subsection{Proof of Prop. \ref{attractiveness_prop}} \label{app: proof_4} 
Before detailing the proof of Prop. \ref{attractiveness_prop}, we give a definition for the first three moments of the shifted Marchenko-Pastur distribution $MP(\alpha, \alpha\sigma, s)$, allowing us to lighten the notations.
\begin{Definition} \label{moments}
    Let $\lambda \sim MP(\alpha, \alpha \sigma, s)$, then we define:
    \begin{itemize}
        \item $\kappa^1(\alpha, \sigma^2, s) = \mathbb{E}[\lambda] = \alpha \sigma^2 + s$;
        \item $\kappa^2(\alpha, \sigma^2, s) = \mathbb{E}[\lambda^2] = \alpha (1+\alpha) (\sigma^2)^2 + 2 \alpha \sigma^2 s + s^2$;
        \item $\kappa^3(\alpha, \sigma^2, s) = \mathbb{E}[\lambda^3] = \alpha (1+3\alpha+\alpha^2) (\sigma^2)^3 + 3 \alpha (1+\alpha) (\sigma^2)^2 s + 3 \alpha \sigma^2 s^2 + s^3$.
    \end{itemize}
\end{Definition}
We are now ready to prove Prop. \ref{attractiveness_prop}.
\begin{proof}
    \begin{align*}
        \mu^{s/u}_1 &= \frac{1}{NK} \mathbb{E}_{\boldsymbol{\eta}} \sum_{i\mu j} \Gamma^{s/u}_{ij} \eta_j \xi^\mu_i \xi^\mu_j = \frac{p}{\alpha N} \textrm{Tr}\boldsymbol\Gamma^{s/u} \boldsymbol\Gamma = \frac{p}{2\alpha N} (\textrm{Tr}(\boldsymbol\Gamma^{s/u})^2 + \textrm{Tr}\boldsymbol\Gamma^2 - \textrm{Tr}(\boldsymbol\Gamma^{s/u} - \boldsymbol{\Gamma})^2)\\
        &\xrightarrow{TDL} \frac{p}{2\alpha} (\mathbb{E}[\lambda^2_{\boldsymbol\Gamma^{s/u}}] + \mathbb{E}[\lambda^2_{\boldsymbol\Gamma}] - \mathbb{E}[\lambda^2_{\boldsymbol\Gamma - \boldsymbol\Gamma^{s/u}}]),
    \end{align*}
    where $TDL$ stands for the thermodynamic limit ($N\to\infty$ and $K/N\to \alpha$) and where we used the ciclicity of the trace together with $\bb A \bb B +\bb B \bb A =  (\bb A^2 +\bb B^2 - (\bb A-\bb B)^2)$. Similarly
    \begin{align*}
        \mu^{s/u}_2 &= \frac{1}{NK} \mathbb{E}_{\boldsymbol{\eta}} \sum_{i\mu jk} \Gamma^{s/u}_{ij} \Gamma^{s/u}_{ik} \eta_j \eta_k \xi^\mu_i \xi^\mu_k = \frac{1-p^2}{N} \textrm{Tr} \boldsymbol\Gamma^2 + \frac{p^2}{\alpha N} \textrm{Tr} (\bb \Gamma^{s/u})^2\boldsymbol\Gamma\\
        &= \frac{1-p^2}{N} \textrm{Tr} (\bb \Gamma^{s/u})^2 + \frac{p^2}{6\alpha N} (\textrm{Tr}(\boldsymbol\Gamma^{s/u}+\boldsymbol{\Gamma})^3 + \textrm{Tr}(\boldsymbol\Gamma^{s/u}-\boldsymbol{\Gamma})^3 - 2 \textrm{Tr} \boldsymbol\Gamma^2) \\
        & \xrightarrow{N \to \infty} (1-p^2) \mathbb{E}[\lambda^2_{\boldsymbol\Gamma^{s/u}}] + \frac{p^2}{6\alpha} (\mathbb{E}[\lambda^3_{\boldsymbol\Gamma - \boldsymbol\Gamma^{s/u}}] + \mathbb{E}[\lambda^3_{\boldsymbol\Gamma + \boldsymbol\Gamma^{s/u}}] - 2 \mathbb{E}[\lambda^3_{\boldsymbol\Gamma}]),
    \end{align*}
    where we used that $\bb A^2 \bb B +\bb A \bb B \bb A + \bb B\bb  A^2= \frac12 [(\bb A+ \bb B)^3+(\bb B - \bb A)^3 -2 \bb B ^3]$ and, due to the ciclicity of the trace, $\text{Tr}(\bb A^2 \bb B +\bb A \bb B \bb A + \bb B\bb A^2) = 3\text{Tr} \bb A ^2 \bb B.$
    The problem now reduces to calculating the asymptotic distribution for each contributions in the previous equations. Now, from the propositions in Sec. \ref{sec: spectra} we already know that 
    \begin{equation*}
        \lambda_{\boldsymbol{\Gamma}} \sim MP\left(\alpha, \alpha, 0\right), \quad \lambda_{\boldsymbol{\Gamma}^s} \sim MP\left(\alpha, \alpha \sigma^s, 0\right), \quad \lambda_{\boldsymbol{\Gamma}^u} \sim MP\left(\alpha, \alpha \sigma^u, s^u\right).
    \end{equation*}
    Now, both $\mu_1 ^{s/u}$ and $\mu_2 ^{s/u}$ depends on the matrix $\frac1N \text{Tr}(\bb \Gamma ^{s/u}-\bb \Gamma)^2$, namely the normalized Frobenius norm of the Hebbian kernel in the (un)supervised setting w.r.t. the storing case. This quantity will be thoroughly analyzed in App.~\ref{sec:Frobenius} and here, we simply recall the necessary results. In both cases, we used the AFM to see that
    \begin{equation*}
        \lambda_{\boldsymbol\Gamma - \boldsymbol\Gamma^{s/u}} \sim MP(\alpha, \alpha\sigma^{s/u}_{-}, s^{s/u}_{-}),
    \end{equation*}
    with
    \begin{itemize}
        \item $\sigma^s_- =\sqrt{1 - 2(1-d)^2r^2 + \big((1-d)^2r^2 + (1-d)\frac{1-(1-d)r^2}{M}\big)^2}$;
        \item $s^s_- =\alpha (1 -  \sigma^s - \sigma^s_-)$;
        \item $\sigma^u_- = \sqrt{1-2(1-d)^2r^2+(\sigma^u)^2} $;
        \item $s^u_- =  \alpha \left(d - \sigma^u_-\right)$,
    \end{itemize}
    and we recall that
    \begin{eqnarray*}
     \sigma^s(r,d,M)& =& (1-d) \Big((1-d) r^2 + \frac{1-(1-d)r^2}{M}\Big),\\
     \sigma^u(r,d,M) &= &\sqrt{(1-d)^4r^4+(1-d)^2\frac{1-(1-d)^2r^4}{M}}.
    \end{eqnarray*}
 Analogously, it is easy to show that
    \begin{equation*}
        \lambda_{\boldsymbol\Gamma + \boldsymbol\Gamma^{s/u}} \sim MP(\alpha, \alpha\sigma^{s/u}_{+}, s^{s/u}_{+}),
    \end{equation*}
    where
    \begin{itemize}
        \item $\sigma^s_{+} = \sqrt{1 + 2(1-d)^2r^2 + \left((1-d)^2r^2 + (1-d)\frac{1-(1-d)r^2}{M}\right)^2}$;
        \item $s^s_{+} = \alpha \left(1 + \sigma^s - \sigma^s_{+}\right)$;
        \item $\sigma^u_{+} = \sqrt{1 + 2(1-d)^2r^2+(\sigma^u)^2}$;
        \item $s^u_{+} = \alpha \left( 1 + (1-d) - \sigma^u_{+}\right)$,
    \end{itemize}
   Using these results, we can express $\mu_{1,2}^{s/u}$ in terms of the moments of Def. \ref{moments} and derive Eqs. (\ref{eq:mu1su}-\ref{eq:mu2su}). This concludes the proof.
\end{proof}

\section{The Approximate Factorization Method} \label{AFMapp}
In the canonical Marchenko-Pastur Theorem, the entries of the factor matrix $\bb{X}$ are assumed to be independent. Conversely, in the one-step analysis we dealt with cases (namely, the unsupervised setting) where this assumption does not hold anymore. In this appendix, we deepen the Approximate Factorization Method (AFM) -- already introduced in Def. \ref{def:afm} -- allowing to bypass mutual dependence of entries involved in the family of covariance matrices (covering the unsupervised Hebbian setting), and thus to compute the asymptotic spectral distribution. Further, we will compare the results to the benchmark approximation $M \to \infty$ extensively employed in \cite{regularizationdreaming, agliari2024spectral}.
\par\medskip
Let $M,N,K \in \mathbb N$, $P=KM$ and $\bb X$ be the matrix with entries ${X}_{i,(\mu, A)} = \chi^{\mu, A}_i \xi^\mu_i $ for $i=1,\dots,N$, $ \mu=1,\dots,K$, and $A=1,\dots,M$, where $\{\chi^{\mu, A}_i\}_{i,\mu,A}$ and $\{\xi^\mu_i\}_{i,\mu}$ are, respectively, two independent classes of i.i.d. random variables with finite moments. In particular, let us assume that $\mathbb{E}\xi^\mu_i=0$, and call $\mathbb{E}(\xi^\mu_i)^2=\theta^2$, $\mathbb{E}\chi^\mu_i=r$ and $\mathbb{E}(\chi^\mu_i)^2=\rho^2$. With these definitions, $\bb X$ is a $N\times P$ matrix, and the index $\mu$ labels a group of $M$ columns (each with length $N$), whose individual element is further specified by the index $A$. We consider the Wishart-like ensemble of random matrices\footnote{Here, we adopted the normalization by $P$ instead of $NM$ as in the definition of the unsupervised \eqref{unsupervisedhebbmatrix}, which allows us to avoid unessential rescalings of the eigenvalues.}
\begin{equation}
    \bb S = \frac{1}{P} \bb{X} \bb{X}^T.
\end{equation}
We will consider the spectral properties of this ensemble when $N\to\infty$ with $\lim_{N\to\infty}K / N =\alpha >0$ fixed (the thermodynamic limit). Straightforward computations shows that $
    \mathbb{E}\bb{X}_{i,(\mu, A)} = 0,$ while
\begin{align*}
    \mathrm{Cov}\left(\bb{X}_{i,(\mu, A)}, \bb{X}_{i,(\mu, B \neq A)}\right) =\theta^2 \rho^2 > 0,
\end{align*}
which breaks statistical independence of entries of the $\bb X$ matrix: thus, Marchenko-Pastur theorem cannot be applied.

\subsection{The \texorpdfstring{$M \to \infty$}{M to infinity} limit} \label{Mtoinftysubsec}
A first simplification in the study of Wishart-like correlation matrices occurs in the limit $M\to\infty$ (at fixed $N$, also referred as to the {\it big data} limit \cite{regularizationdreaming, agliari2024spectral}). Indeed, rewriting
\begin{align}\label{eq:unsup_kernel}
     S_{i,j} &=  \frac{1}{KM} \sum^K_{\mu=1} \sum^M_{A=1} {X}_{i,(\mu, A)} {X}_{j,(\mu, A)} = \frac{1}{K} \sum^K_{\mu=1} \Big( \frac{1}{M}\sum^M_{A=1} \chi^{\mu, A}_i \chi^{\mu, A}_j\Big) \xi^\mu_i \xi^\mu_j,
\end{align}
and using the strong law of large numbers, which provides
\begin{align*}
    \frac{1}{M}\sum^M_{A=1} \chi^{\mu, A}_i \chi^{\mu, A}_j \xrightarrow{M\to\infty} \mathbb{E}(\chi^{\mu, A}_i)^2 \delta_{i,j}  +  \mathbb{E}\chi^{\mu, A}_i \cdot \mathbb{E}\chi^{\mu, A}_j (1-\delta_{i,j})= \delta_{i,j} \rho^2 + (1-\delta_{i,j}) r^2,
\end{align*}
we obtain
\begin{align*}
     S_{i,j} 
    &= \frac{r^2}{K} \sum^K_{\mu=1}  \xi^\mu_i   \xi^\mu_j  + \delta_{i,j} (\rho^2 - r^2) \frac{1}{K} \sum^K_{\mu=1} \left(\xi^\mu_i\right)^2.
\end{align*}
In other words, in the big data limit, the Wishart-like matrix $\bb S$ splits as the sum of an actual Wishart matrix and a diagonal one, which can be analyzed separately in the thermodynamic limit. For the diagonal contribution, again using the strong law of large numbers, we get
\begin{equation*}
    \frac{1}{K} \sum^K_{\mu=1} (\xi^\mu_i)^2 \xrightarrow{K \to \infty} \mathbb{E}(\xi^\mu_i)^2 = \theta^2.
\end{equation*}
Hence, the diagonal matrix converges to $\theta^2 (\rho^2 - r^2) \bb{1}$. On the other hand, the Wishart matrix can be written as $\frac{1}{K}\bb Y \bb Y^T$ where $Y_{i,\mu} = r \xi^\mu_i$, whose entries are independent random variables with zero mean and second moment
\begin{equation*}
    \mathbb{E}(r \xi^\mu_i)^2 = r^2 \theta^2.
\end{equation*}
Thus, Marchenko-Pastur Theorem applies, and the empirical spectral distribution of $\frac1K \bb Y \bb Y^T$ converges to a $MP\left(\alpha, r^2 \theta^2\right)$, as $N \to \infty$. Going back to the matrix $\bb S$ in the big data limit, the diagonal contributions only gives a shift in the eigenvalues, and the asymptotic spectral distribution is modified accordingly, leading to the distribution (recall that, in this notation, the third argument of the Marchenko-Pastur law stands for the shift)
\begin{equation} \label{Mtoinftyeq}
    MP\left(\alpha, r^2 \theta^2, \theta^2 (\rho^2 - r^2)\right).
\end{equation}
This asymptotic is expected to fit the real distribution for large (but finite) $M$, yet it fails for (even relatively) small number of examples per class $M$. Further, adopting this limit in the unsupervised setting, it is not possible to capture essential scaling between $M$ and the relevant parameters characterizing the training dataset. This motivates the AFM exploited in the next paragraph.

\subsection{Finite (but large) $M$: the Approximate Factorization Method} \label{afm} 
As already stated before, the peculiarity of the unsupervised setting is due to the presence of the random sums $\frac{1}{M}\sum^M_{A=1} \chi^{\mu, A}_i \chi^{\mu, A}_j$, whose presence breaks the possibility to recover the usual Wishart setting. However, as exploited in the previous Subsection, when $M$ is large enough, simplicitations are expected to hold and, in the big data limit, the factorization in a Wishart-type matrix (expect for a trivial shift of the eigenvalues) is reproduced. Based on this observation, and thus
recover the Marchenko-Pastur Theorem \ref{marchenkopasturtheorem}, we assume that there is a class of i.i.d. random variables, namely $\{\phi^\mu_i\}_{i,\mu}$, such that
\begin{equation}  \label{factorization}
    \phi^\mu_i \phi^\mu_j \approx \frac{1}{M}\sum^M_{A=1} \chi^{\mu, A}_i \chi^{\mu, A}_j, \quad \forall i\neq j \text{ and } \forall \mu.
\end{equation}
from which the name to the Approximate Factorization Method. An immediate observation is that the factors $\phi^\mu_i$ are a function of the errors $\chi^{\mu, A}_i$, but clearly they are independent on patterns $\xi^\mu_i$. Within the AFM, for $i\neq j$, the expression of $ S_{i,j}$ reads as
\begin{equation*}
     S_{i,j} = \frac{1}{K} \sum^K_{\mu=1} \phi^\mu_i \phi^\mu_j \xi^\mu_i \xi^\mu_j.
\end{equation*}
Hence, taking in consideration also the diagonal contributions, we can rewrite the Wishart-like matrix
in the form
\begin{equation}\label{eq:eqafm}
    \bb S = \frac{1}{K} \bb \Phi \bb \Phi^T + \text{diag}\Big[\frac{1}{K} \sum^K_{\mu=1} \Big( \frac{1}{M}\sum^M_{A=1} (\chi^{\mu, A}_i)^2 - (\phi^\mu_i)^2\Big) (\xi^\mu_i)^2\Big],
\end{equation}
where $\bb \Phi$ is a $N \times K$ matrix with entries $\Phi^\mu_i = \phi^\mu_i \xi^\mu_i$. Remarkably, the matrix $\bb S$ is again the sum of a Wishart random matrix and a diagonal matrix. Let us focus on the first contribution. Concerning the $\Phi^\mu_i$ random variables, it is straightforward to notice that:
\begin{itemize}
    \item they are i.i.d.;
    \item they have zero mean;
    \item the variance is
    $\text{Var}(\Phi^\mu_i)=\mathbb{E}[(\phi^\mu_i\xi^\mu_i)^2] = \mathbb{E}[(\phi^\mu_i)^2] \mathbb{E}[(\xi^\mu_i)^2] = \theta^2 \mathbb{E}[(\phi^\mu_i)^2]$.
\end{itemize}
At this level, we only need to compute $\mathbb{E}[(\phi^\mu_i)^2]$. To this aim, we observe that
\begin{align*}
    \mathbb{E}[(\phi^\mu_i)^2]^2 \equiv \mathbb{E}[(\phi^\mu_i\phi^\mu_j)^2] = \mathbb{E}\Big[\Big(\frac{1}{M}\sum^M_{A=1}\chi^{\mu, A}_i \chi^{\mu, A}_j\Big)^2\Big] = r^4 + \frac{\rho^4-r^4}{M},
\end{align*}
which leads to
\begin{equation} \label{eq:sm_phi}
    \mathbb{E}[(\phi^\mu_i)^2] = \sqrt{r^4 + \frac{\rho^4-r^4}{M}}.
\end{equation}
As for the diagonal part, by virtue of the strong law of large numbers and by the independence of $\xi^\mu_i$ and $\chi^{\mu,A}_i$, in the thermodynamic limit we have
\begin{align*}
    \frac{1}{K} \sum^K_{\mu=1} \Big( \frac{1}{M}\sum^M_{A=1} (\chi^{\mu, A}_i)^2 - \left(\phi^\mu_i\right)^2\Big) (\xi^\mu_i)^2 \xrightarrow{K \to \infty}  \theta^2 (\rho^2 - \mathbb{E}[(\phi^\mu_i)^2]),
\end{align*}
where the expectation in the last line is given by Eq. \eqref{eq:sm_phi}. Hence, we can conclude that, in the thermodynamic limit and by virtue of the AFM, the spectral distribution of $\bb S$ converges asymptotically to the Marchenko-Pastur distribution 
\begin{equation} 
    MP\Big(\alpha, \theta^2 \sqrt{r^4 + \frac{\rho^4-r^4}{M}}, \theta^2 \big(\rho^2 - \sqrt{r^4 + \frac{\rho^4-r^4}{M}}\big)\Big).
\end{equation}
Clearly, the $M \to \infty$ exactly recovers the result \eqref{Mtoinftyeq}. Although the approximated nature of the pursued method, it capture the scaling between $M$, $r$ and $\rho$.

\subsection{A counterexample at finite (but small) $M$}

In this appendix, we give a counterexample showing that the AFM does not work at small $M$. Suppose, for instance, that the $\chi^{\mu, A}_{i}$ are i.i.d. Bernoulli random variables, and fix $M=3$. If we could always split the sum in the expression of $ S_{i,j}$ in two i.i.d. random variables we should have that, given $X \sim \text{Bin}(M=3, p)$, there are two i.i.d. random variables $Y$ and $Z$, such that $X \sim YZ$. Then, adopting the condensed notation $\mc P_X(k) = \mc P(X=k)$ we should have:
\begin{itemize}
    \item $(1-p)^3 = \mc P_X(0) = \mc P_Y(0) + \mc P_Z(0) - \mc P_Y(0) \mc P_Z(0) = 2\mc P_Y(0) - \mc P_Y(0)^2$;
    \item $3p(1-p)^2 = \mc P_X(1) = \mc P_Y(1) \mc P_Z(1) = \mc P_Y(1)^2$;
    \item $3p^2(1-p) = \mc P_X(2) = \mc P_Y(1) \mc P_Z(2) + \mc P_Y(2) \mc P_Z(1) = 2\mc P_Y(1)\mc P_Y(2)$;
    \item $p^3 = \mc P_X(3) =\mc P_Y(1) \mc P_Z(3) + \mc P_Y(3) \mc P_Z(1) = 2\mc P_Y(1)\mc P_Y(3)$;
\end{itemize}
where we used the fact that $Z$ and $Y$ must have the same support of $X$ and, of course, that they are independent and they have the same probability distribution. The calculations above lead us to:
\begin{itemize}
    \item $\mc P_Y(0) = 1 - \sqrt{1-(1-p)^3}$;
    \item $\mc P_Y(1) = \sqrt{3p}(1-p)$
    \item $\mc P_Y(2) = \frac{\sqrt{3}}{2}p^{\frac{3}{2}}$
    \item $\mc P_Y(3) = \frac{p^{\frac{5}{2}}}{2\sqrt{3}(1-p)}$
\end{itemize}
But, on the other hand, being $\mc P_Y$ a probability distribution, we are constrained by the condition:
\begin{equation*}
    1 = \sum^4_{i=0} \mc P_Y(i) = 1 - \sqrt{1-(1-p)^3} + \sqrt{3p}(1-p) + \frac{\sqrt{3}}{2}p^{\frac{3}{2}} + \frac{p^{\frac{5}{2}}}{2\sqrt{3}(1-p)},
\end{equation*}
which is true if and only if $p=0$: in this case, the factorization is not only inexact, but fails also as an approximation. The regime of applicability is clearly $M \gg 1$, in which case heuristic arguments in its favor can be carried out. For instance, when $M \gg 1$, by virtue of the Central Limit Theorem the sum $\frac{1}{M}\sum^M_{A=1} \chi^{\mu, A}_i \chi^{\mu, A}_j$ has a Gaussian behavior. Moreover, using the Delta-Method \cite{pinelis2016optimal}, one can prove that the centered and normalized product of two independent $Bin(L, p)$ converge in distribution to a Gaussian random variable when $L \to \infty$. This means that in the limit $M \to \infty$ there are two i.i.d. random variables that, for $L \to \infty$, converge to the sum above, and AFM holds.

\subsection{An algebraic comment on AFM}
As a side note, we make a further comment about the validity of the AFM. Let us assume generate $\chi^{\mu,A}_i$ and collect the entries $i=1,\dots,N$ for fixed $\mu$ and $A$ in vectors $\bb \chi^{\mu,A}$, lying in $\mathbb R^N$ (for instance, in the Gaussian scenario  $\chi^{\mu,A}_i \sim\mathcal N(r,\sqrt{\rho^2-r^2})$, we would have $\bb \chi^{\mu,A}= \bb{r} +\sqrt{\rho^2-r^2}\bb z^{\mu,A}$, with $\bb{r} = r(1,1,\dots,1)$ the mean vector and $ z_i^{\mu,A} \sim_{i.i.d.} \mathcal N(0,1)$ isotropic perturbations around $\bb{r}$). Let us now fix $\mu=1,\dots,K$ and generate $M<N$\footnote{Possibly, $1\ll M\ll N$.} vectors as before, then with large probability they will be linearly independent (since they are expressed as the sum of a constant vector $\bb{r}$ and, for large $N$, mutually orthogonal vectors $\bb z^{\mu,A}$). Now, the random sums (with fixed $\mu$)
$$
\frac1M \sum_{A=1}^M \chi^{\mu,A}_i \chi^{\mu,A}_j,
$$
display a projector-like structure, namely $\frac1M \sum_{A=1}\bb \chi ^{\mu,A} (\bb \chi ^{\mu,A})^T$, where the attribute ``-like'' is due to the fact that the vectors are not mutually orthogonal. Indeed, even in the $N\to\infty$ limit, the angle between any pair of vectors $\bb \chi^{\mu,A}$ and $\bb \chi^{\mu,B}$ is non-zero:
$$
\cos \theta (\bb \chi^{\mu,A},\bb \chi^{\mu,B})= \frac{\bb \chi^{\mu,A}\cdot \bb \chi^{\mu,B}}{\lVert \bb \chi^{\mu,A}\lVert \cdot \lVert \bb \chi^{\mu,B}\lVert}\underset{N\gg 1}\approx\frac{r^2}{\rho^2}+N^{-1/2} \sqrt{1-\frac{r^4}{\rho^4}}z^\mu_{A,B},
$$
with $z^{\mu}_{AB}\sim \mathcal N(0,1)$, since $\lVert \bb \chi^{\mu,A}\lVert \approx \rho \sqrt N$. Thus, the kernel $M^{-1} \sum_{A=1}^M \bb \chi^{\mu,A}(\bb \chi^{\mu,A})^T$ is associated, for fixed $\mu$, with a $M$-dimensional subspace of $\mathbb R^N$, which we call $V$. Now, when introducing the AFM, we assumed that there exists a single (random) vector $\bb \phi^\mu$ exactly reproducing the statistics of $M^{-1} \sum_{A=1}^M \bb \chi^{\mu,A}(\bb \chi^{\mu,A})^T$, so that we can replace (for $i\neq j$) the kernel $\frac1 {NM}\sum_{\mu}  \sum_A \chi^{\mu,A}_i \chi^{\mu,A}_j\xi^\mu_i \xi^\mu_j$ with the simplest one $\frac1N \sum_{\mu}  \phi^{\mu}_i \phi^{\mu}_j\xi^\mu_i \xi^\mu_j $. This is clearly not true in general, as $\bb \phi^\mu$ (for $\mu$ fixed) is associated with a 1-dimensional subspace of $\mathbb R^N$. However, things get simpler for high enough $M$, and the signature of the subspace $V$ in the structure of the Hebbian kernel $\bb \Gamma^u$ can be safely neglected in this regime. To see this, let us take any vector $\bb x\in \mathbb R^N$ and consider $(\bb \Gamma ^u \cdot \bb x)_i= \frac1N \sum_{\mu}\xi^\mu_i\big(\frac1{M} \sum_{jA} \chi^{\mu,A}_i \chi^{\mu,A}_j \xi^\mu_j x_j\big)$.\footnote{The constraint $j\neq i$ only gives trivial corrections, thus can be safely omitted.} For a given set $M$ of vectors $\{\bb \chi^{\mu,1},\dots \bb \chi^{\mu,M}\}$ and calling $v^\mu_j = \xi^\mu_j x_j$ (which is independent on the $\bb \chi^{\mu,A}$ vectors), we can always decompose $\bb v^\mu = \bb v^{\mu}_{\paral}+\bb v^{\mu}_\perp$, with $\bb v^{\mu}_{\paral}=\frac1{\sqrt M} \sum_{A=1}^M c_{\mu,A} \bb \chi^{\mu,A}$ and $\bb v^{\mu}_\perp\in V^\perp$, that is $\bb \chi^{\mu,A}\cdot \bb v^{\mu}_\perp = 0$ for all $A$. Now, different realizations of the vectors $\bb \chi^{\mu,A}$ translates in different decompositions, thus the coefficients $c_{\mu,A}$ are random variables depending on the $\bb \chi^{\mu,A}$ vectors; {\it ii)} since the vectors $\bb \chi^{\mu,A}$ are mutually independent random vectors, then so will be the coefficients $c_{\mu,A}$ (indeed, replacing for some $\bar A$ the $\bb \chi^{\mu,\bar A}$ with a new vector $\tilde{\bb \chi}^{\mu}$ will only affect the coefficient $c_{\bar A}$). But now:
\begin{align*}
   & \frac1{NM} \sum_{jA} \chi^{\mu,A}_i \chi^{\mu,A}_j \xi^\mu_j x_j=
    \frac1{NM} \sum_{jA} \chi^{\mu,A}_i \chi^{\mu,A}_j v^\mu_j =\frac1{NM} \sum_{jA} \chi^{\mu,A}_i \chi^{\mu,A}_j v^{\mu}_{\paral,j}=\\=&\,
    \frac1{NM} \Big(\frac1{\sqrt M}\sum_{A} c_{\mu,A} \chi^{\mu,A}_i \lVert\bb \chi^{\mu,A}\lVert^2+\frac1{\sqrt M}\sum_{A\neq B} c_{\mu,B} \chi^{\mu,A}_i \lVert\bb \chi^{\mu,A}\lVert \lVert \bb\chi^{\mu,B}\lVert \cos\theta _{AB}\Big)\approx\\
    \approx& \,\frac{\rho^2-r^2}M \frac1{\sqrt M}\sum_{A} c_{\mu,A} \chi^{\mu,A}_i +\frac{r^2}M\sum_{B} \frac{c_{\mu,B}}{\sqrt M}  \cdot\sum_A\chi^{\mu,A}_i = \frac{\rho^2-r^2}{M} v^{\mu}_{\paral,i}+{r^2}\sum_{B} \frac{c_{\mu,B}}{\sqrt M} \cdot \frac1M\sum_A\chi^{\mu,A}_i.
\end{align*}
It is straightforward to notice that the only contributions carrying information about $V$ is the first one, as we can approximate both $\sum_B c_{\mu,B}/\sqrt M$ and $\sum_A \chi^{\mu,A}_i$ by virtue of the CLT to get an example-free approximation, in particular $\frac1M \sum_{A=1}^M \chi^{\mu,A}_i\sim \bar \chi^\mu_i$ with $\bar \chi^\mu_i$ Gaussian-distributed. Assuming that $\sum_B c_{\mu,B}/\sqrt M$ has some finite limit $c_\mu$, then we would have $\frac1{NM} \sum_{jA} \chi^{\mu,A}_i \chi^{\mu,A}_j \xi^\mu_j x_j\approx r^2 c_\mu \bar \chi^\mu_i+\mc O(M^{-1})$. 
This finally implies
$$
(\bb \Gamma ^u\cdot \bb x)_i\approx r^2 \sum_{\mu} \xi^\mu_i c_\mu \bar \chi^\mu_i+\mc O\Big(\frac{\rho^2-r^2}M\Big).
$$
Then, for $M$ large enough, the first contribution can be neglected. On the other hand, taking a set of i.i.d. random variables $\{\phi^\mu_i\}$, and decomposing it again in the $\phi$-collinear and orthogonal parts, namely $v^\mu_j= \xi^\mu_i x_i = C_\mu \phi^\mu_j + w_{j,\perp}$, we would get $\frac1{N} \sum_{j} \phi^\mu_i \phi^\mu_j \xi^\mu_j x_j=C_\mu\sigma ^2 \phi^\mu_i$, where $\sigma^2=\mathbb E(\phi^\mu_i )^2$. This in turn leads to
$$
  (\bb \Gamma ^u\cdot\bb x)_i \approx  \sigma ^2\sum_{\mu}\xi^\mu_i C_\mu \phi^\mu_i.
$$
The two approaches exhibits compatibles structure, and the latter is expected to match the whole unsupervised Hebbian kernel when $(\rho^2-r^2)/M\ll1$. In this sense, AFM can be interpreted as a Central Limit Theorem (although it is not proven) for projector(-like) operators, and in particular will be very useful for our concerns since both spectral properties and retrieval capabilities are directly related to quantities of the form $(\bb \Gamma ^u\cdot\bb x)_i$.

\subsection{Time to face the reality: simulations} \label{approximations discussions}
To conclude this Appendix, we numerically explore the validity of the AFM. Here, we assume the Gaussian scenario, namely $\xi^\mu_i \sim \mathcal{N}\left(0, \theta^2\right)$ and $\chi^{\mu, A}_i \sim \mathcal{N}\left(r, \rho^2 - r^2\right)$. This will allow us to conduct the analysis for general values $\theta$ and $\rho$, which in the Rademacher setting exploited in the main text are instead fixed. 
In Figure \ref{AFM_M}, we report the comparison between the empirical spectral distribution of a sample of Wishart-like matrices and the theoretical predictions given by the $M\to\infty$ limit (App. \ref{Mtoinftysubsec}) and the AFM (App. \ref{afm}). As is clear, we observe that the AFM approximation starts to well-describe the empirical distribution at relatively low values of $M$. {Notice that, in the unsupervised setting, the empirical spectral distribution is characterized by two separate bulks, rather than a $\delta$-peak and a continuous distribution: both the approximation derived in the big data limit and within AFM fails in describing this peculiarity.} 
In Figs. \ref{AFM_variances} and \ref{AFM_thetas} we reported the results for different values of $\rho^2$ and $\theta^2$. In the former, we can see that, increasing $\rho$ with the other parameters fixed (that is, increasing the variance of the $\chi^{\mu,A}_i$ variables), the bulk approximated by a $\delta$ becomes wider, and both the approximations do not work anymore. Conversely, increasing $\theta$ has no such effect as expected. In our opinion, it is worthy to inspect the origin of the secondary bulk, however this falls outside the scope of the present paper, hence we cease to give further details on this.
\begin{figure}[h!]
    \centering  \includegraphics[width=\textwidth]{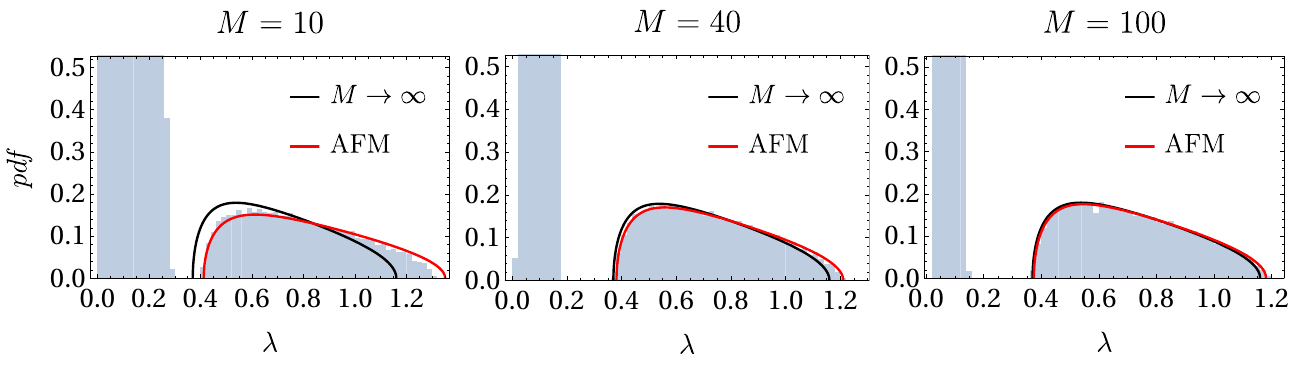}
    \caption{Histogram of the spectral distribution of $\bb S$ with Gaussian entries, when $N=1000$, $\alpha=0.1$ $\theta=0.5$, $r=0.5$, $\rho=0.75$, for different values of $M$, against the analytical formula obtained through the $M \to \infty$ approximation (black) and the AFM (red). The empirical distribution is realized by collecting the eigenvalues of a sample of 100 different realizations of the Hebbian matrices.}
    \label{AFM_M}
\end{figure}
\begin{figure}[h!]
    \centering  \includegraphics[width=\textwidth]{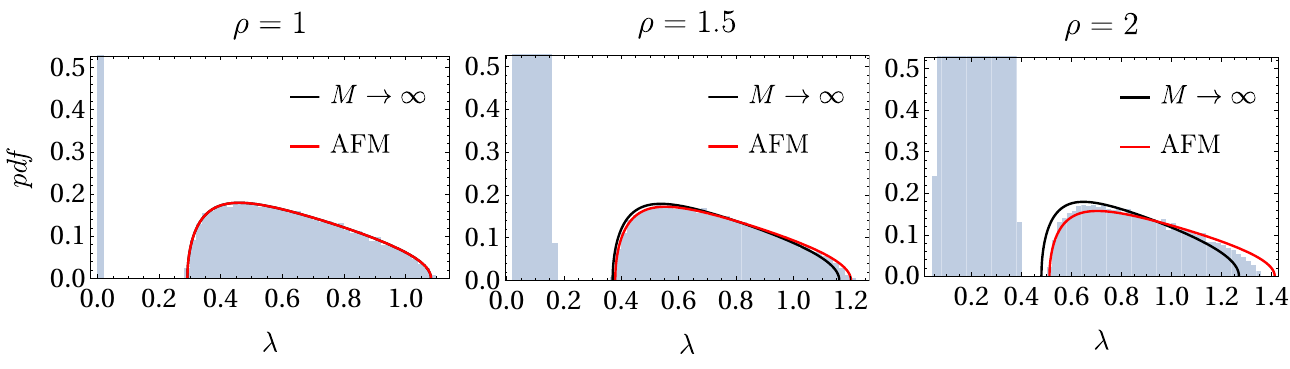}
    \caption{Histogram of the spectral distribution of $\bb S$ with Gaussian entries, when $N=1000$, $\alpha=0.1$, $\theta=0.25$, $r=1$, $M=50$, for different values of $\rho$, against the analytical formula obtained through the $M \to \infty$ approximation (black) and the AFM (red). The empirical distribution is realized by collecting the eigenvalues of a sample of 100 different realizations of the Hebbian matrices.}
    \label{AFM_variances}
\end{figure}
\begin{figure}[h!]
    \centering  \includegraphics[width=\textwidth]{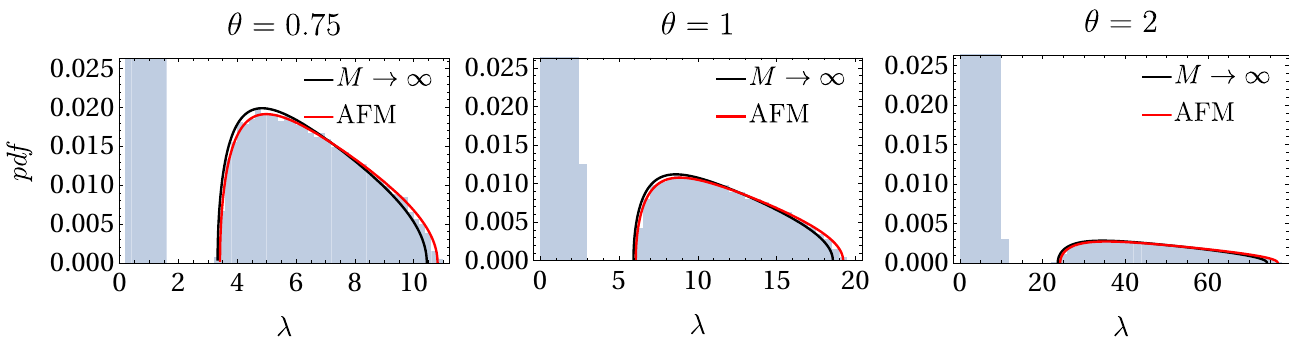}
    \caption{Histogram of the spectral distribution of $\bb S$ with Gaussian entries, when $N=1000$, $\alpha=0.1$, $\rho=1.25$, $r=1$, $M=50$, for different values of $\theta$, against the analytical formula obtained through the $M \to \infty$ approximation (black) and the AFM (red). The empirical distribution is realized by collecting the eigenvalues of a sample of 100 different realizations of the Hebbian matrices.}
    \label{AFM_thetas}
\end{figure}

\section{Frobenius norm and squared error} \label{sec:Frobenius}
We start recalling the definition of the Frobenius Norm $||\cdot||_F$:
\begin{Definition} \label{def:1}
    Let $\mathbf A$ be a $N \times N$ real matrix. The (normalized) \textit{Frobenius norm} of $\mathbf A$, is defined as:
    \begin{equation}
        ||\mathbf A||_F := \sqrt{\frac{1}{N}\sum^N_{i=1} \sum^N_{j=1} A^2_{i,j}}.
    \end{equation}
\end{Definition}
We recall some useful properties of the Frobenius norm (see e.g. \cite{horn2012matrix}):
    \begin{enumerate}
        \item $||\cdot||_F$ is a matrix norm;
        \item $||\cdot||_F$ is sub-multiplicative norm, thus, given two $N \times N$ real matrices $\mathbf A$ and $\mathbf B$, we have
        \begin{equation}
            ||\mathbf A \mathbf B||_F \leq ||\mathbf A||_F ||\mathbf B||_F.
        \end{equation}
        \item Given a $N \times N$ real matrix $\mathbf A$, we have that
        \begin{equation}
            ||\mathbf A||_F = \sqrt{\frac{1}{N}\textrm{Tr}\left(\mathbf A \mathbf A^T\right)} = \sqrt{\frac{1}{N}\sum^N_{i=1} \lambda_i^2},
        \end{equation}
        where $\lambda_i$, $i=1,\dots,N$, are the $N$ eigenvalues of $\mathbf A$.
    \end{enumerate}
The first two properties guarantee that $||\cdot||_F$ is a useful norm and then can be used to induce a matrix distance. The third one, explains why we have chosen exactly the Frobenius norm, i.e. because it exploits the spectral properties of the matrix which it measures. More precisely, if we take a symmetric random matrix, 
in the $N\to\infty$ limit we get
\begin{equation*}
    ||\mathbf A||_F = \sqrt{\frac{1}{N}\sum^N_{i=1} \lambda_i^2} \xrightarrow{N \to \infty} \sqrt{\mathbb{E}\left[\lambda_{\mathbf A}^2\right]},
\end{equation*}
where with $\lambda_{\mathbf A}$ we denote a stochastic variable distributed according to the limit of the empirical distribution of the eigenvalues of $\mathbf A$.
\par\medskip
When dealing with examples of unavailable patterns, either in supervised or unsupervised settings, it is natural to question whether our empirical models account for a good representation of the reality, namely whether $\boldsymbol \Gamma^s$ and $\boldsymbol \Gamma^u$ are close to $\boldsymbol \Gamma$, where we directly store the patterns. To this goal, we use the Frobenius norm to estimate the squared distance between the empirical and the ideal coupling matrices, which is interpreted as \emph{squared error} \cite{bernstein2020distance}. 
\begin{Definition}
    Let $\mathbf A$ and $\mathbf B$ be two $N \times N$ symmetric real random matrices. The Squared Error (SE) of $\mathbf B$ w.r.t. $\mathbf A$ is defined as:
    \begin{equation}
        \mathcal{E}\left(\mathbf A, \mathbf B\right) = ||\mathbf A - \mathbf B||^2_F.
    \end{equation}
\end{Definition}
In addition, we discuss how the parameters $\alpha$, $r$, $d$ and $M$ affect this error.
In order to compute $\mathcal{E}\left(\mathbf \Gamma^{s,u},\mathbf \Gamma \right)$ in the thermodynamic limit $N \to \infty$, we will have to compute the limit of the second non-centered moment of the empirical spectral distribution of $\mathbf \Gamma^{s,u} - \mathbf \Gamma$.


\subsection{\texorpdfstring{$\boldsymbol\Gamma^s$ vs $\boldsymbol\Gamma$}{Supervised vs Basic Storing}:  the validity of supervised Hebb's rule} 
In order to evaluate $\mathcal{E}\left(\boldsymbol\Gamma^s,\boldsymbol\Gamma\right)$ it is convenient to write explicitly the difference between $\boldsymbol\Gamma$ and $\boldsymbol\Gamma^s$
\begin{equation}
    \Gamma_{i,j} - \Gamma^s_{i,j} = \frac{1}{N} \sum^K_{\mu=1} ( 1 - \Bar\chi^\mu_i \Bar\chi^\mu_j ) \xi^\mu_i \xi^\mu_j,
\end{equation}
with $\bar \chi^\mu_i = \frac1M \sum_{A=1}^M \chi^{\mu,A}_i$. 
In the same spirit as the AFM, we now consider a class of i.i.d. random variables $\{\phi^\mu_i\}_{\mu,i}$, such that
\begin{equation}
    \phi^\mu_i \phi^\mu_j = 1 - \Bar\chi^\mu_i \Bar\chi^\mu_j \quad \forall i \neq j,
\end{equation}
and with the factors $\phi^\mu_i$ depending on the errors $\chi^{\mu,A}_i$, but
not on the concepts $\xi^\mu_i$. Then, we rewrite the difference between the two matrices in the usual form
\begin{align}
    \Gamma_{i,j} - \Gamma^s_{i,j} 
    &= \frac{1}{N} \sum^K_{\mu=1} (\phi^\mu_i \xi^\mu_i) (\phi^\mu_j \xi^\mu_j) + \frac{1}{N} \sum^K_{\mu=1} ( 1 - (\Bar\chi^\mu_i)^2 - (\Bar\phi^\mu_i)^2)\delta_{i,j},
\end{align}
namely, the sum of a Wishart Matrix and diagonal one. The latter in the thermodynamic limit simplifies by virtue of the strong law of large numbers, and the diagonal entries converge to a constant. Specifically
\begin{equation}
    \alpha \frac{1}{K}\sum^K_{\mu=1} (1 - (\Bar\chi^\mu_i)^2 -  (\phi^\mu_i)^2) \to \alpha \left(1-\mathbb{E}\left[(\Bar\chi^\mu_i)^2\right]-\mathbb{E}\left[(\phi^\mu_i)^2\right]\right):= s^s_-.
\end{equation}
In order to compute $s_-^s$, we have at first to compute the second moments of the $\phi^{\mu}_i$. Recalling that $\alpha \mathbb{E}[(\bar \xi^{\mu}_i )^2] = \sigma^s$ (given by Eq. \eqref{sigmasup}), we proceed again as follows:
\begin{align*}
    \mathbb{E}[(\phi^\mu_i)^2]^2 &= \mathbb{E}[(\phi^\mu_i \phi^\mu_j)^2]= \mathbb{E}[(1 - \Bar\chi^\mu_i \Bar\chi^\mu_j )^2]= 1 - 2(1-d)^2r^2 + \mathbb{E}[(\Bar\chi^\mu_i)^2]^2=\\
    &= 1 - 2(1-d)^2r^2 + \left(\sigma^s\right)^2 \\
    &= 1 - 2(1-d)^2r^2 + \Big((1-d)^2r^2 + (1-d)\frac{1-(1-d)r^2}{M}\Big)^2,
\end{align*}
thus
\begin{align*}
    \sigma^s_-:=\mathbb{E}[(\phi^\mu_i)^2] &= \sqrt{1 - 2(1-d)^2r^2 + \left(\sigma^s\right)^2} \\
    &= \sqrt{1 - 2(1-d)^2r^2 + \Big((1-d)^2r^2 + (1-d)\frac{1-(1-d)r^2}{M}\Big)^2}.
\end{align*}
We immediately have that
\begin{equation}
    s^s_- = \alpha (1 -  \sigma^s - \sigma^s_-).
\end{equation}
Concerning the contribution involving the Wishart matrix, we have $
    \mathbb{E}[\phi^\mu_i\xi^\mu_i] = \mathbb{E}[\phi^\mu_i] \mathbb{E}[\xi^\mu_i] = 0,
$ while $\mathbb E [(\phi^\mu_i\xi^\mu_i)^2]=\mathbb E[(\phi^\mu_i)^2]= \sigma^s_-$, since $\xi^\mu_i = \pm1$. Consequently, also in this approximation the spectrum of $\boldsymbol\Gamma-\boldsymbol\Gamma^s$ is described asymptotically by a shifted Marchenko-Pastur distribution, namely $MP\left(\alpha, \alpha\sigma^s_-, s^s_-\right)$, i.e. they have the same probability distribution as $\alpha\sigma^s_- \lambda + s^s_-$, where $\lambda$ follows a $MP\left(\alpha, 1\right)$ distribution. In the end, we can conclude that
\begin{align}\label{eq:SE_sup}
    \mathcal{E}\left(\boldsymbol\Gamma^s, \boldsymbol\Gamma\right) &
    = \kappa^2(\alpha, \sigma^s_-, s^s_-)= \notag \\ &=\alpha[1-2(1-d)^2r^2+(\sigma^s)^2+2 \alpha(\sigma^s-1)^2]= \notag
    \\&= \alpha [1-2(1-d)^2r^2+(1-d)^4r^4(1+\rho)^2 +2\alpha((1-d)^2r^2(1+\rho)-1)^2 ],
\end{align}
where $\rho = \frac{1 - (1-d)r^2}{M(1-d)r^2}$,\footnote{The parameter $\rho$ is used to quantify the information content in the synthetic dataset. In fact, proceeding similarly to \cite{AABD-NN2022, alemanno2023supervised}, we find $\mc P\left(\operatorname{sgn}\left(\sum_a \chi_i^{\mu a}\right) = -1\right) \approx 1 - \operatorname{erf}(1/\sqrt{2\rho})$, for $M \gg 1$. Thus, the conditional entropy $H(\xi_i^{\mu}|\{\xi_i^{\mu, A}\}_{A})$, which quantifies the amount of information needed to describe the original message $\xi_i^{\mu}$, given the examples related to it, increases monotonically with $\rho$. For example, in the supervised setting we can write $\sigma^s = (1-d)^2r^2 \left(1 + \rho\right)$.} and $\kappa^2(\alpha, \sigma^2, s)$ represents the second moment of a $MP(\alpha, \alpha\sigma^2, s)$, see Def. \ref{moments}.

\subsection{\texorpdfstring{$\boldsymbol\Gamma^{u}$ vs $\boldsymbol\Gamma$}{Unsupervised vs Basic Storing}:  the validity of unsupervised Hebb's rule} 
We now focus on the squared error $\mathcal E(\boldsymbol \Gamma^u, \boldsymbol \Gamma)$ between the basic storing coupling matrix $\boldsymbol\Gamma$ and its unsupervised counterpart $\boldsymbol\Gamma^u$. Let us rewrite the difference as
\begin{equation} \label{eq:45}
    \Gamma^u_{i,j}-\Gamma_{i,j} = \frac{1}{N} \sum^K_{\mu=1} \Big(1 - \frac{1}{M}\sum^M_{A=1}\chi^{\mu, A}_i \chi^{\mu, A}_j\Big) \xi^\mu_i \xi^\mu_j.
\end{equation}
Adopting the AFM, we then introduce a class of i.i.d. random variables $\{\phi^\mu_i\}_{\mu,i}$ such that, for $i \neq j$, we make the following replacement:
\begin{equation}
    1 - \frac{1}{M}\sum^M_{A=1}\chi^{\mu, A}_i \chi^{\mu, A}_j =\phi^\mu_i \phi^\mu_j .
\end{equation}
Again, these random variables are independent on $\xi^\mu_i$ as they only depends on the $\chi^\mu_i$. Then, we can recast \eqref{eq:45} as
\begin{align}
    \Gamma_{i,j}-\Gamma^u_{i,j} 
    &= \frac{1}{N} \sum^K_{\mu=1} (\phi^\mu_i\xi^\mu_i) (\phi^\mu_j  \xi^\mu_j) + \delta_{i,j} \frac{\alpha}{K}\sum_\mu (1-(\chi^{\mu, A}_i)^2 - (\phi^\mu_i)^2).
\end{align}
Needless to say, this is again the sum of a Wishart Matrix and a diagonal one. 
Moreover, focusing on the covariance matrix and thanks to the statistical independence, we have that $\mathbb{E}[\phi^\mu_i\xi^\mu_i] = \mathbb{E}[\phi^\mu_i]\mathbb{E}[\xi^\mu_i] = 0$ and $\mathbb{E}[(\phi^\mu_i\xi^\mu_i)^2]=\mathbb{E}[(\phi^\mu_i)^2]$.
As a consequence of the AFM, the spectrum of $\boldsymbol\Gamma - \boldsymbol\Gamma^u$ will follow again a Marchenko-Pastur distribution shifted by some parameter to be determined. In the thermodynamic limit by virtue of strong law of large numbers we have
\begin{equation}\label{eq:55}
     \frac{\alpha}{K}\sum_\mu (1-(\chi^{\mu, A}_i)^2 - (\phi^\mu_i)^2) \to \alpha ( d - \mathbb{E}[(\phi^\mu_i)^2]):=s^u_-.
\end{equation}
Thus, again the scale parameter and the shift of the distribution depends on the second moment of the $\phi^\mu_i$. This reads as
\begin{align} \label{eq:51}
    \sigma^u_- := \mathbb{E}[(\phi^\mu_i)^2] &= \sqrt{\mathbb{E}[(\phi^\mu_i\phi^\mu_j)^2]} = \sqrt{\mathbb{E}\Big[\Big(1 - \frac{1}{M}\sum^M_{A=1}\chi^{\mu, A}_i \chi^{\mu, A}_j\Big)^2\Big]}= \notag \\
    &= \sqrt{1-2\mathbb{E}[(\chi^{\mu,A}_i)^2]+\mathbb{E}\Big[\Big(\frac{1}{M}\sum^M_{A=1}\chi^{\mu, A}_i \chi^{\mu, A}_j\Big)^2\Big]} =\notag \\
    & = \sqrt{1-2(1-d)^2r^2+(\sigma^u)^2},
\end{align}
where $\sigma^u$ is defined as in Eq. \eqref{sigmau}. Combining Eqs. \eqref{eq:51} and \eqref{eq:51}, we get
\begin{equation}
    s^u_- = \alpha \left(d - \sigma^u_-\right).
\end{equation}
This implies that the limiting spectral distribution of the ensemble of matrices $\boldsymbol\Gamma^u-\boldsymbol\Gamma$ is a shifted Marchenko-Pastur distribution $MP\left( \alpha, \alpha\sigma^u_-, s^u_- \right)$.
After these considerations, we arrive to the conclusion that
\begin{align}\label{eq:SE_unsup}
\mathcal{E}\left(\boldsymbol\Gamma^u,\boldsymbol\Gamma\right) &
= \kappa^2(\alpha, \sigma^u_-, s^u_-) \notag \\ 
&= \alpha [1 - 2(1-d)^2r^2+(\sigma^u)^2 + \alpha d^2] \notag \\
&= \alpha \Big( 1 - 2 (1-d)^2r^2 + (1-d)^4 r^4 + (1-d)^2 \frac{1-(1-d)^2r^4}{M} + \alpha d^2\Big).
\end{align}
This analytical result highlights that $\mathcal{E}\left(\boldsymbol\Gamma^u,\boldsymbol\Gamma\right)$ is decreasing in $r,M$, increasing in $d$. The dependence on $M$ and $r$ is intuitive, because as $M$ or $r$ grow the dataset $\{\boldsymbol \eta^{\mu,A}\}_{\mu=1,...,K}^{A=1,...,M}$ gets more and more informative on the archetypes $\{\boldsymbol \xi^{\mu}\}_{\mu=1,...,K}$, in such a way that $\boldsymbol\Gamma^u$ provides a better approximation of $\boldsymbol\Gamma$. The role of the dilution parameter $d$ is discussed in Sec.~\ref{sec:1step}. Moreover, by deriving in $\alpha$ we get
\begin{eqnarray}
    \frac{\mathrm d}{\mathrm d \alpha} \mathcal{E}\left(\boldsymbol\Gamma^u,\boldsymbol\Gamma\right) = (2d) \cdot \alpha + k(r,d,M),
\end{eqnarray}
having posed $k(r,d,M) = 1 - 2 (1-d)^2r^2 + (1-d)^4 r^4 + (1-d)^2 \frac{1-(1-d)^2r^4}{M}$.Therefore, we can state that for $d>0$ there is a quadratic dependence on $\alpha$, while the dependence becomes linear when $d=0$.

\subsection{Validity of the models with examples: theory {\it versus} numerics} \label{sec:validity}

In this section, we give a comparison between the theoretical estimates and numerical results for the squared error between the empirical (e.g. supervised and unsupervised) settings {\it versus} the basic storing picture. We directly report the theoretical predictions reported in Eqs. \eqref{eq:SE_sup} and \eqref{eq:SE_unsup}, making explicit the dependence on the parameters $\alpha$, $d$ and $r$.
For the supervised case, we find
\begin{align}  \label{eq:35}
    \mathcal{E}\left(\boldsymbol\Gamma^s, \boldsymbol\Gamma\right) = \alpha [1-2(1-d)^2r^2+(1-d)^4r^4(1+\rho)^2 +2\alpha((1-d)^2r^2(1+\rho)-1)^2 ].
\end{align}
Likewise, in the unsupervised case we find
\begin{align}  \label{eq:38}
\mathcal{E}\left(\boldsymbol\Gamma^u,\boldsymbol\Gamma\right) = \alpha \big[ (1 - (1-d)^2r^2)^2 + (1-d)^2 \frac{1-(1-d)^2r^4}{M} + \alpha d^2\big].
\end{align}
\begin{figure}
    \centering  \includegraphics[width=0.9\textwidth]{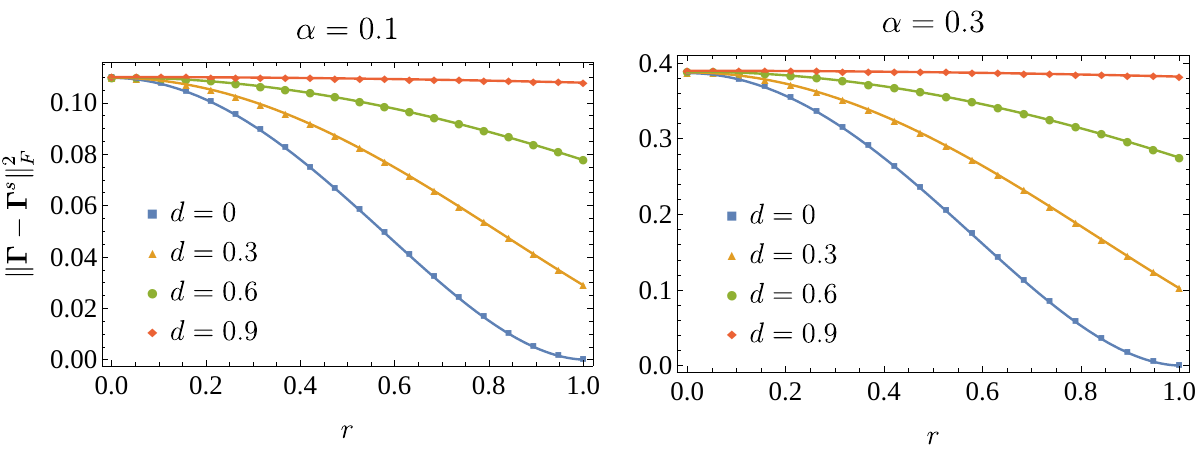}
    \centering  \includegraphics[width=0.9\textwidth]{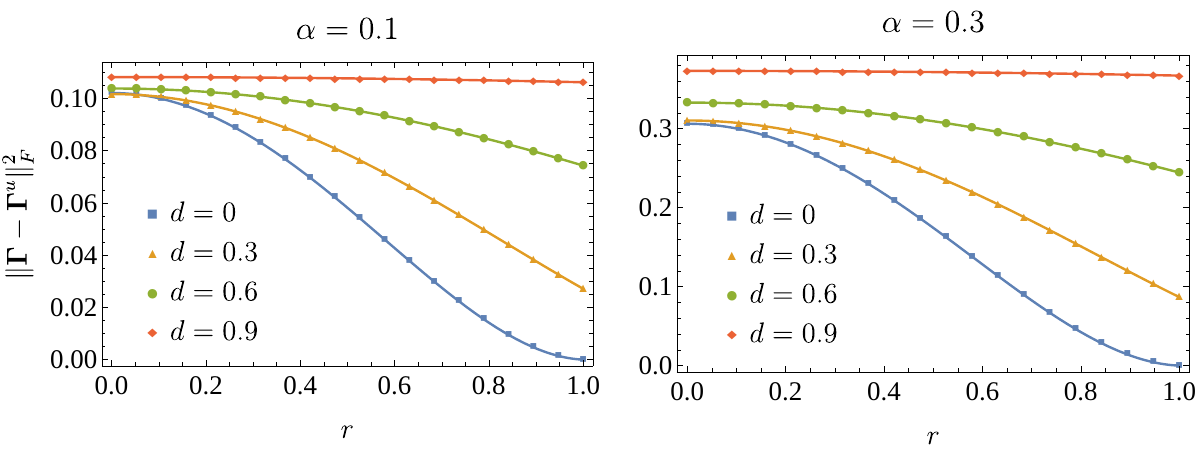}
    \caption{{\bfseries Comparison of the squared error between basic storing and supervised (upper panels) or unsupervised (lower panels) settings.} The plots compare numerical results (dots) and theoretical predictions \eqref{eq:35} and \eqref{eq:38} (solid curves) for the (normalized) Frobenius norm between $\bb \Gamma$ and, respectively, $\bb\Gamma^s$ and $\bb\Gamma^u$ for $\alpha=0.1$ (left) and $\alpha=0.3$ (right) for fixed $M=50$ and different choices for $d$, as a function of the dataset quality $r$. Numerical results are averaged over 30 different realizations of the coupling matrices. The network size is $N=1000$. Error bars are not reported, as the relative error is very low. 
    }
    \label{hebb_vs_BS}
\end{figure}

\par
These analytical results highlight that $\mathcal{E}\left(\boldsymbol\Gamma^s,\boldsymbol\Gamma\right)$ and $\mathcal{E}\left(\boldsymbol\Gamma^u,\boldsymbol\Gamma\right)$ decrease in $r,M$ and increase in $d$. The dependence on $M$ and $r$ is intuitive, because as $M$ or $r$ grow, the data set $\{\boldsymbol \xi^{\mu,A}\}_{\mu=1,...,K}^{A=1,...,M}$ becomes more and more informative on the archetypes $\{\boldsymbol \xi^{\mu}\}_{\mu=1,...,K}$, so that $\boldsymbol\Gamma^{s,u}$ provides a better approximation of $\boldsymbol\Gamma$. However, the larger $d$, the more insensitive the SE becomes to the quality of the examples $r$.
Moreover, we can state that both SEs grow with $\alpha$, the dependence being quadratic (but in the unsupervised case the dependence becomes linear when $d=0$). 
The numerical simulations shown in Fig.~\ref{hebb_vs_BS} (upper and lower panels, respectively) confirm these remarks and corroborate the AFM in the evaluation of the SE, as the analytical estimates accurately fit the empirical data (even for relatively small sizes, $N \sim 10^2$ and $M \sim 10$).

\section{A comparison with signal-to-noise analysis}\label{app:snr}
In this appendix, we relate and compare the approach developed in the present paper (and the work \cite{agliari2024spectral}) with a standard, non-rigorous technique used to anticipate the stability of a certain configuration and referred to as {\it signal-to-noise} \cite{Amit} (SNA). For simplicity and without loss of generality, we will focus in this case on the stability of patterns (also referred to as {\it invariance}) in the unsupervised setting; specifically, the configuration $\bb\sigma^{(0)} = \bb\xi^1$ is provided as the initial condition for the neural dynamics \eqref{eq:dynamics}, corresponding to $p=1$ in the attractiveness analysis as defined in Sec. \ref{sec:1step}. 
Denoting with $h_i(\boldsymbol{\sigma})= \sum_{j\neq i} \Gamma_{i,j} \sigma_j$ the local field acting on the $i$-th neuron, we say that the pattern $\boldsymbol \xi^1 $ is stable under the dynamics \eqref{eq:dynamics} iff $\Delta_i(\bb\xi^1) = h_i(\boldsymbol{\xi}^1)\xi_i^1>0$ for $i=1,\hdots, N$. The basic idea in signal-to-noise analysis is to split $\Delta_i$ into a sum of signal and noise, namely positive non-random terms from stochastic contributions responsible for driving the neural dynamics away from the desired target. For instance, in the usual Hopfield model the stability of a given pattern $\bb\xi^1$ can be written as $\Delta_i =\xi^1_i\sum_{j\neq i} \Gamma_{i,j} \xi^1_j= \sum_{j\neq i} \xi^1_i\big(\frac1N \xi^1_i \xi^1_j\big)\xi^1_j + \sum_{j \neq i }\xi^1_i\big(\frac1N \sum_{\mu >1} \xi^\mu_i \xi^\mu_j\big)\xi^1_j$, with the first contribution being of order $1$, while the second one being the sum of uncorrelated random variables, namely a random walk of $(K-1)(N-1)$ steps with equally probable unitary jumps, see \cite{Amit}. Hence, for large networks the pattern stability can be effectively approximated as $\Delta_i = S+R z_i$, with $z_i \sim \mathcal N(0,1)$. It is worth noticing that the randomness of the variables $\bb z$ encodes the variability in the pattern stability due to all possible realizations of the patterns $\bb \xi^\mu$ themselves, so that $\mathbb E_{\bb z} \equiv \mathbb E_{\bb\xi}$ for large $N$. This implies that the signal and the noise precisely capture the first and second moments of the stability $\Delta_i$, specifically $S = \mathbb E_{\bb \xi} \Delta_i$ and $R= \mathbb E_{\bb \xi} (\Delta_i ^2)- S^2$. Further, whether the signal-to-noise ratio (SNR) is low ($S/R\ll1$), the stability criterion $h_i (\bb\xi^1) \xi^1_i >0$ holds for all $i=1,\dots,N$ with large probability. Conversely, if $S/R = \mathcal O (1)$, the probability to have a misaligned field $h_i$ w.r.t. the pattern is non-negligible, and neural dynamics can drive the system away from the target $\bb\xi^1$. In probabilistic terms, $\mathcal P (\Delta_i \ge0)\equiv \mc P (S+R z_i \ge 0)=\frac12\big[1+\text{erf}\big(\frac{S}{\sqrt{2R^2}}\big)\big]$, which matches the expression from our spectral approach within the Gaussian approximation (Def. \ref{def:GA}) upon identifying $S= \mathbb E_{\bb \xi}\Delta_i=\mu_1 $ and $R^2=\mathbb E_{\bb \xi} (\Delta_i ^2)- S^2= \mu_2 -\mu_1 ^2$, see also \cite{agliari2024spectral}. In the unsupervised diluted setting, the pattern stability reads as $\Delta_i (\bb\xi^1) =  \sum_j \Gamma^u_{i,j} \xi^1_i \xi^1_j$, 
and in Fig.~\ref{fig:SNR_stability}, we reported the related histograms for various values of the dilution parameter, also reporting in the insets the frequency of events with $\Delta_i \ge 0$.
\begin{figure}[tb]
    \centering
   \includegraphics[width=\textwidth]{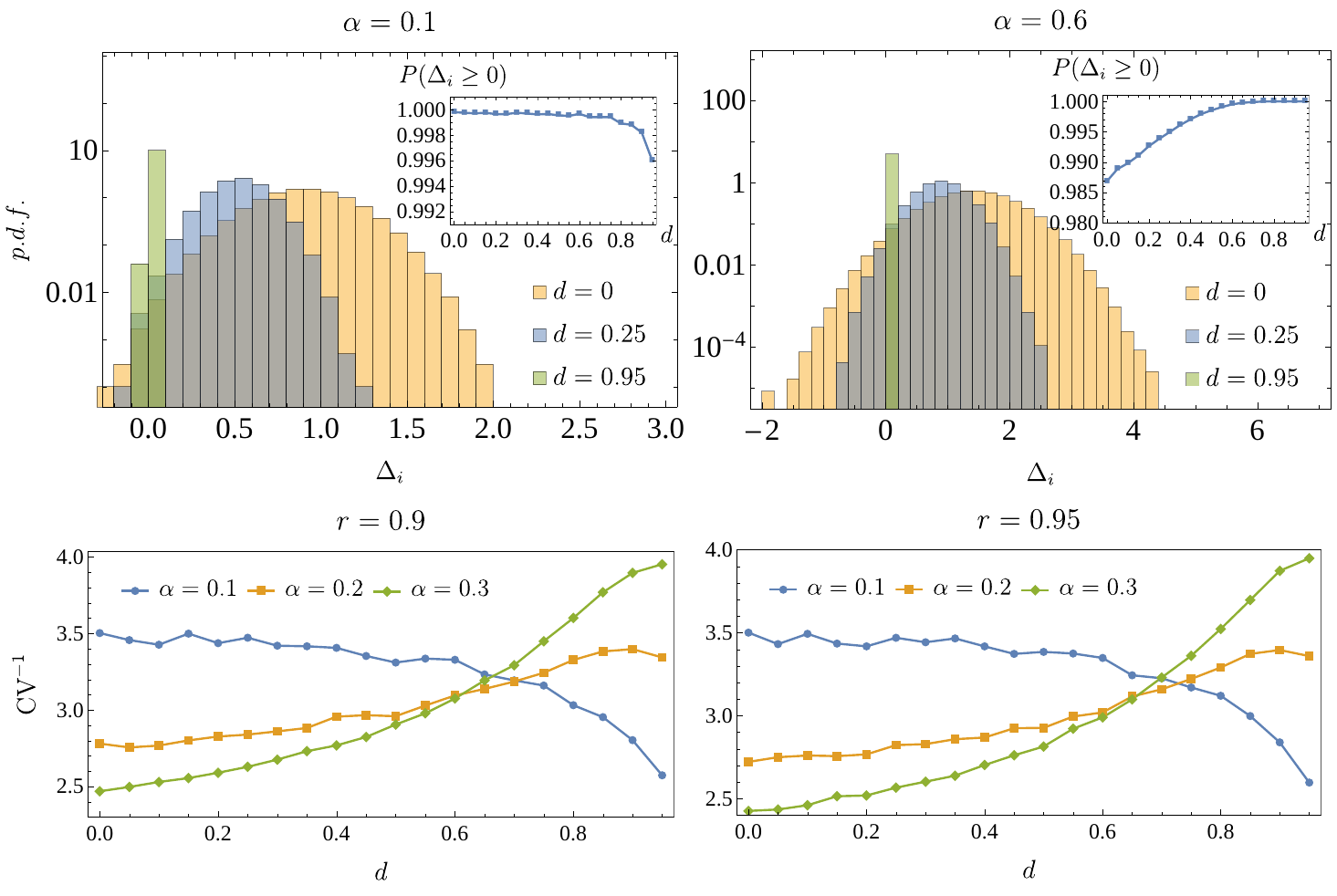}
    \caption{{\bfseries Details of empirical distributions of the stability in the unsupervised setting.} In the first row, we reported the empirical distribution of the stability $\Delta_i$ for $\alpha=0.1$ (left) and $\alpha=0.6$ (right) for various values of the dilution parameters ($d=0,0.25,0.95$). In the inset plots, we also reported the fraction $P(\Delta_i \ge 0)$ of events leading to positive pattern stability as a function of the dilution parameter for the same values of $\alpha$. Here, $M=50$ and $r=0.9$. In the second row, we reported the inverse coefficient of variation ($\text{CV}^{-1}$) as a function of $d$ for the empirical distributions of the pattern stability $\Delta_i$. The plots refer to $r=0.9$ (left) and $r=0.95$ (right), with a storage capacity of $\alpha=0.1$ (blue circles), $0.2$ (yellow squares), $0.3$ (green diamonds). The number of examples per class is $M=50$. The results are collect by setting $\bb\sigma^{(0)}$ to all of the ground-truths $\bb\xi^\mu$ and computing the pattern stability, for 50 different realizations of the training dataset (and, consequently, of the coupling matrix). In both cases, the network size is $N=1000$.}
    \label{fig:SNR_stability}
\end{figure}
\par\medskip
To analyze these results on a theoretical level, we proceed as in the usual Hopfield model splitting the pattern stability as
\begin{equation}
    \label{eq:delta_split}
    \Delta_i (\bb\xi^1) =  \sum_{j}\left(\frac1{NM} \sum_{A}\chi^{1,A}_i \chi^{1,A}_j \xi^1_i \xi^1_j\right) \xi^1_i \xi^1_j+\sum_{j}\left(\frac1{NM} \sum_{A,\mu >1}\chi^{\mu,A}_i \chi^{\mu,A}_j \xi^\mu_i \xi^\mu_j\right) \xi^1_i \xi^1_j, 
\end{equation}
where we separated the contribution of the pattern $\bb\xi^1$ to the unsupervised coupling matrix $\bb \Gamma^u$ from those of the remaining $K-1$ ground-truths. Since in our approach we included self-interactions (namely, $\Gamma^u_{i,i}\neq 0$), finite contributions to the signal can come also from the second term (this is clear by setting $i=j$ in the second quantity on the r.h.s.).\footnote{This actually holds also in the usual Hopfield model, since allowing for diagonal entries of the coupling matrix $\bb \Gamma$ would result in a higher stability of patterns, leading to a 1-step magnetization $m_1 = \text{erf}\big(\frac{1+\alpha}{\sqrt{2\alpha}}\big)$, rather than $m_1 = \text{erf}\big(\frac{1}{\sqrt{2\alpha}}\big)$ without self-interactions.}
Indeed, the terms making up the signal involve essentially the contributions of the coupling matrix due to examples at different spin sites $(i,j)$ but belonging to the same pattern (of the form $\chi^{1,A}_i\chi ^{1,A}_j$, giving contribution of order $(1-d)^2$ due to statistical independence) {\it and} the diagonal entries of $\bb\Gamma^{u}$ coming from all of the ground-truths (thus involving terms of the form $(\chi^{\mu,A}_i)^2$, whose magnitude is tuned by $1-d$). Conversely, the noise term will only receive contributions from uncorrelated pixels of the non-retrieved patterns (namely, $\chi^{\mu,A}_i \chi^{\mu,A}_j$ for $\mu >1$ and $i\neq j$), thus the factor $R$ will be of the order $1-d$ (i.e. the variance is of the order $(1-d)^2$). From this naive analysis, it is clear that first and second moments of the empirical spectral distribution will behave differently in terms of the dilution parameter. A cursory look at the histograms in Fig. \ref{fig:SNR_stability} confirms this expectations, as the general effect of dilution is to reduce both the signal (the mean) and the noise (the variance) of the pattern stability. Indeed, at low $\alpha$ and without dilution, the majority of the $\Delta_i$ are strictly positive, and increasing $d$ results in a simultaneous shift of the empirical distribution towards zero values and a drastic suppression of the variance; in this case, a higher $d$ results in a higher spin-flip chance (corresponding to $\Delta_i <0$), as also remarked by the inset in left upper plot. At higher $\alpha$ and zero dilution, the empirical distribution is instead very broad, resulting in a significant tail below zero. Increasing the dilution in this case results in a downsizing of the probability of observing $\Delta_i <0$, as highlighted by the inset plot in the right histogram. This effect is more clear by looking at the plots in the second row of Fig. \ref{fig:SNR_stability}, where we reported the inverse of coefficient  of variation (CV, which is essentially the inverse of the SNR) as a function of $d$ for various values of $\alpha$. At low $\alpha$, $\text{CV}^{-1}$ is a monotonously decreasing function of $d$, thus patterns are made less stable. For higher storage (above the threshold $\alpha_c=0.14$), diluting the training set is beneficial, and the argument of the error function will increase, thus signaling better stability performances (in terms of the 1-step magnetization $m_1$).
\begin{figure}[tb]
    \centering
   \includegraphics[width=\textwidth]{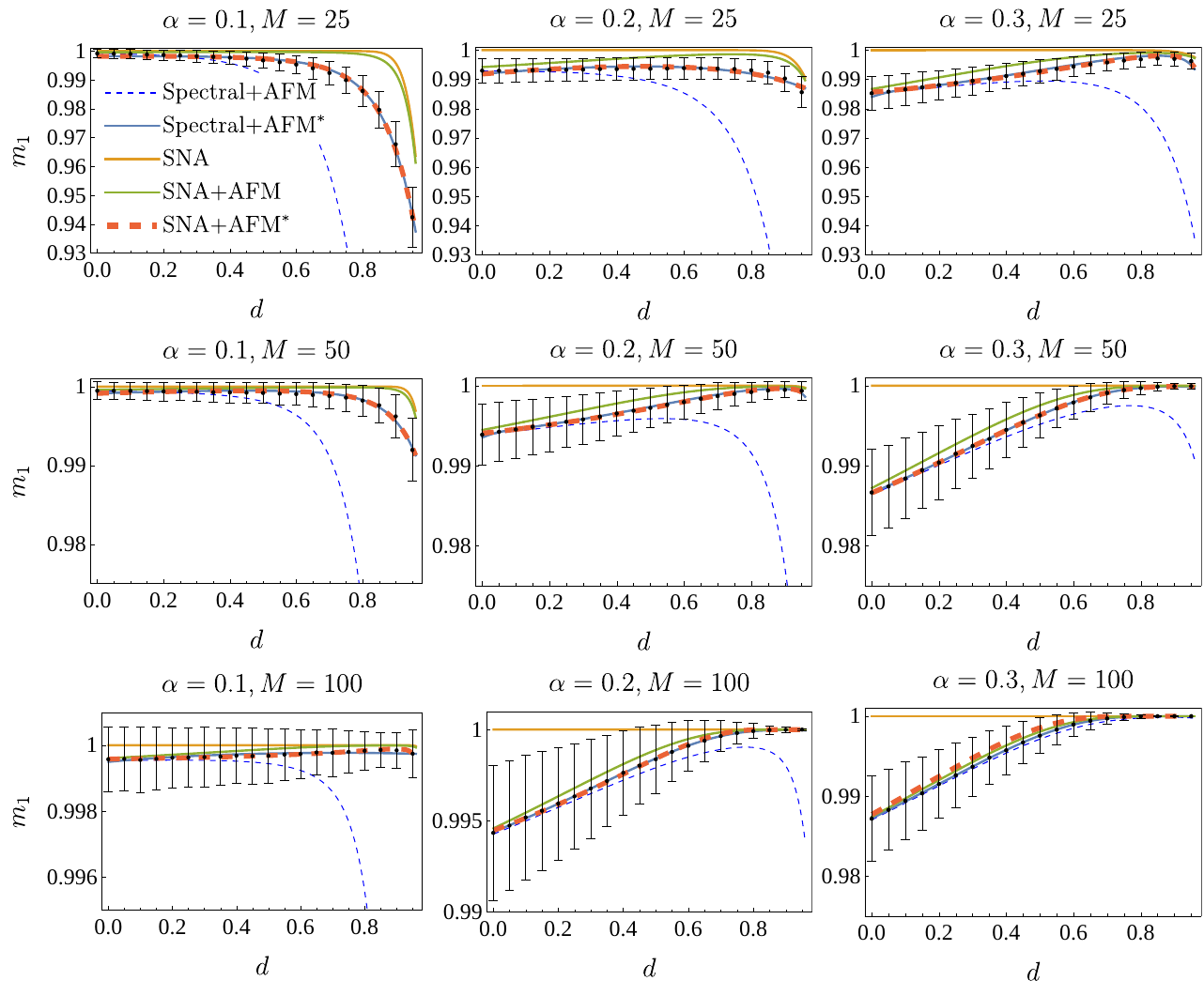}
    \caption{{\bfseries Comparison between signal-to-noise and spectral approach.} The plots report the empirical results for the 1-step magnetization $m_1$ as a function of $d$ (black markers) compared to five different theoretical predictions: spectral approach with AFM ($\text{Spectral+AFM}$, dashed blue line), the best-fitting version of the theoretical prediction of the spectral approach with AFM ($\text{Spectral+AFM}^*$, blue solide line), the naive SNR ($\text{SNR}$, yellow curve), SNR+AFM ($\text{SNR+AFM}$, green curve), and the best-fitting version of SNR+AFM ($\text{SNR+AFM}^*$, dashed thick red line). The nine plots refer to $\alpha=0.1,0.2,0.3$ (from left to right), and $M=25,50,100$ (from top to bottom). The empirical results at $r=0.9$ averaging over 50 different realizations of the training dataset (and, consequently, of the coupling matrix). The network size is $N=1000$.}
    \label{fig:SNR_stability_final}
\end{figure}
To conclude, we can perform a standard signal-to-noise analysis for the pattern stability $\Delta_i (\bb\xi^1)$ in the large $N$ limit, with the working hypothesis that isolated contributions in Eq. \eqref{eq:delta_split} are mutually uncorrelated. In particular, it is easy to verify that
$$
\sum_{j}\left(\frac1{NM} \sum_{A}\chi^{1,A}_i \chi^{1,A}_j \xi^1_i \xi^1_j\right) \xi^1_i \xi^1_j\sim r^2 (1-d)^2 +\mc O (N^{-1/2}),
$$
while
$$
\sum_{j}\big(\frac1{NM} \sum_{A,\mu >1}\chi^{\mu,A}_i \chi^{\mu,A}_j \xi^\mu_i \xi^\mu_j\big) \xi^1_i \xi^1_j\sim \alpha(1-d) +\sqrt{\frac{\alpha(1-d)^2}{M}}z_i,
$$
where again $z_i \sim \mathcal N(0,1)$. Thus, clearly $S\equiv \mathbb E_{\bb z} \Delta_i= r^2(1-d)^2+\alpha (1-d)$ and $R^2\equiv \mathbb E_{\bb z} \Delta_i-S^2 = \alpha M^{-1}(1-d)^2$. If the $\Delta_i$ at different spin site are independent, we can still approximate $m_1 = 2 \mathcal P(\Delta_i \ge 0 )-1=\text{erf}(S/\sqrt{2 R^2})$. This naive signal-to-noise analysis results in the same expression of the signal if compared to the spectral approach with the AFM, but fails in evaluating the second moment, as it gives a poor estimation of the noise (in particular, it is independent on $r$). We can improve this analysis by means of AFM
$$
\frac1M \sum_{A=1}^M \chi^{\mu,A}_i \chi^{\mu,A}_j \approx \phi^\mu_i \phi^\mu_j,
$$
for $i\neq j$. This procedure leads to the same signal $S=\alpha(1-d)+r^2(1-d)^2$, and a better evaluation of the noise, which reads as
$$
R\approx \sqrt{\alpha \Big((1-d)^4 r^4+(1-d)^2\frac{1-(1-d)^2 r^4}{M}\Big)}.
$$
This expression reproduces the one provided by naive SNA only for $r=0$, thus signaling that the latter only works for totally uncorrelated examples. The comparison between the numerical experiments and the theoretical predictions are reported in Fig. \ref{fig:SNR_stability_final}, where we reported the 1-step magnetization in five different settings: naive signal-to-noise (SNA), signal-to-noise with AFM (SNA+AFM and SNA+AFM$^*$, {\it vide infra}), spectral approach with AFM (Spectral+AFM and Spectral+AFM$^*$, again {\it vide infra}). SNA (yellow curve) and spectral-AFM (blue dashed line) fail in reproducing empirical data (black markers), while SNA+AFM properly working for sufficiently high $M$ and $\alpha$. However, as we also did in the main text, whenever AFM is used we can take the functional form of the theoretical 1-step magnetization seriously and fit the empirical data (these are SNA+AFM$^*$, dashed red curve, and Spectral+AFM$^*$, blue solid line). In this case, we consider $m_1$ as a function of $d$ and treat $\alpha$, $r$ and $M$ as tunable parameters.\footnote{Clearly, $\alpha$, $M$ and $r$ are parameters characterizing the training dataset, the best-fitting values have not physical meaning. The purpose here is to show that the theoretical predictions provided by SNR+AFM and the spectral approach have enough complexity to capture the qualitative behavior of the empirical results. For this reason, we do not report the best-fitting values for the best-fitting parameters.} With this procedure, both the SNR+AFM and the spectral approach matches perfectly the empirical results (with a coefficient of determination close to 1). These results hence suggests that AFM (besides making computations easier and providing an effective picture of the theory) allows to derive meaningful results capturing the (at least qualitatively correct) description of Hopfield-like networks in presence of non-trivial correlations between stored vectors.

\section{A more in-depth investigation on the mechanism of dilution}\label{app:mechanism}

In this Appendix, we further clarify the role of the dilution within the unsupervised setting, and aim to shed light on the mechanism underlying the positive effects of removing entries in the examples employed to build up the Hebbian coupling matrix. To this aim, we rely again on the deterministic evolution of the update rule \eqref{eq:dynamics} at zero temperature towards fixed point starting from the ground-truths themselves and from the test examples: the results are reported respectively in first and second column of Fig. \ref{fig:detail_stabgen} for various values of $M$ and $K$ as a function of the dilution parameter $d$. As already discussed in the main text, at low $K$ dilution is unnecessary or even harmful for the generalization purposes. Conversely, at high $K$ and sufficiently high $M$, dataset dilution becomes a beneficial strategy, and moderate $d$ ($\sim0.4\div 0.5$) exhibits a good trade-off between ground-truth stability and generalization performances of the model.
\begin{figure}[tb]
    \centering
    \includegraphics[width=0.8\linewidth]{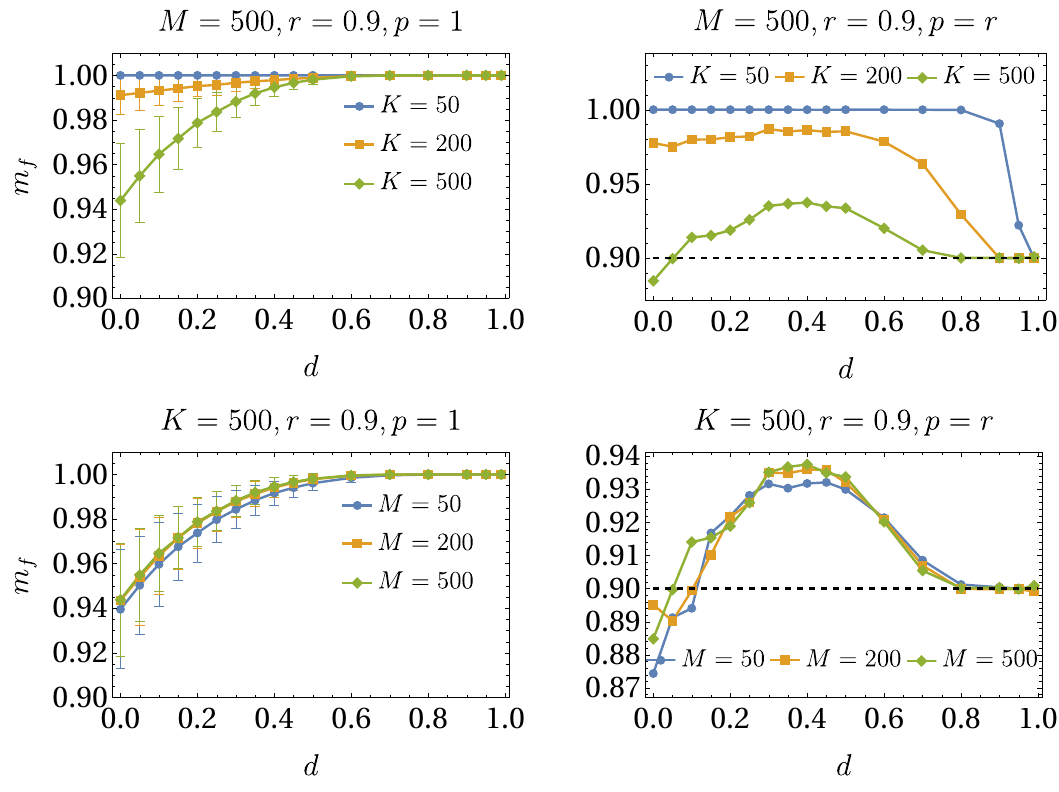}
    \caption{{\bfseries A detail of pattern stability and generalization capabilities of Hopfield model with diluted dataset.} The plots show a detail of numerical results about the pattern stability (left column) and generalization capabilities, namely retrieval capabilities of the model starting from a testing example (right column). The numerical results are reported as a function of $d$ for fixed $M=500$ and $r=0.9$ and various values of $K$ (first row), and fixed $K=500$ and $r=0.9$ for various values of $M$ (second row). The network size is fixed to $N=1000$, the numerical results are averaged over 20 different realizations of the training dataset for each value of $d$.}
    \label{fig:detail_stabgen}
\end{figure}\par\medskip

Let us focus on the attraction power of the ground-truths, select a pattern $\bb\xi^\mu$ and randomly flip a fraction $q$ of the bits. The perturbed configuration is then provided to the network as initial condition, which is allowed to evolve towards a fixed point. We then measure the overlap $m_f(q)$ between the final configuration and the pattern $\bb\xi^\mu$. The results are reported in Fig. \ref{fig:num_attract} for various values of $K$, $M$ and $d$ as a function of the initial overlap of the initial condition $m_0 (q)$. As expected, at low $K$ the undiluted case has best retrieval capabilities (as well as a large attraction basins), with the diluted counterpart matching its performances only for high number of stored training examples. This behavior aligns with standard theory of associative neural networks. However, as the number $K$ of patterns increases, we observe a qualitative shift: for $K/N$ slightly larger than the Hopfield critical storage capacity (e.g. $K=200$, second row), dilution begins to perform better at large $M$. When $K$ is far above the critical storage capacity of the Hopfield model (e.g. $K=500$, third row), diluted Hopfield model consistently outperforms the usual Hebbian prescription, and the performance gap increases with $M$. More remarkably, in this oversaturated regime the fraction of flipped entries leading to stable configurations (namely $m_f(q) = m_0 (q)$, corresponding to the intersection between the retrieval maps and the dashed black line) is always greater in the diluted case. The corresponding value $q^*$ identifies a sphere of radius $R^* = R(q^*)$ around the ground-truth that is stable under neural dynamics: configurations lying on this ball evolve into others on the sphere itself (not necessarily the initial condition), thus forming an invariant set under the dynamics. Also, the plots highlight that configurations inside the ball are typically pushed on its boundary -- resulting in a lower final overlap -- while those outside approach it. In other words, initial condition with sufficiently high correlation with a given ground-truth are pushed by neural dynamics towards invariant spheres by neural dynamics. This feature is typical of associative neural networks, see for instance \cite{mceliece1987capacity}, the crucial point here being that -- at high load -- the radius of the invariant sets decreases as dilution is introduced.

\begin{figure}[tb]
    \centering
    \includegraphics[width=0.95\linewidth]{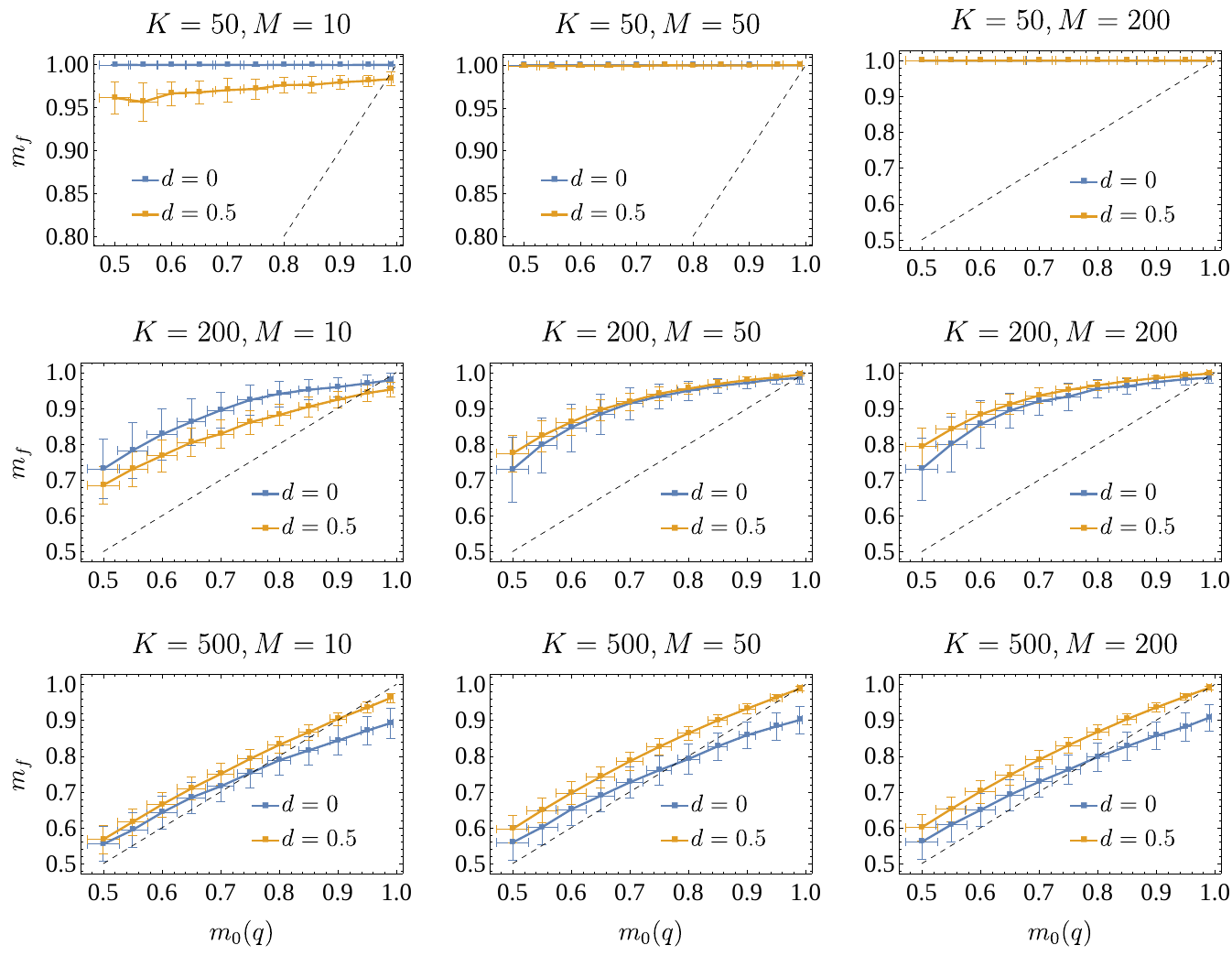}
  \caption{{\bfseries Retrieval maps in the unsupervised setting}. The plots report the results of the pattern attractiveness, namely the final magnetization $m_f$ under the relaxation to fixed points of the neural dynamics \eqref{eq:dynamics} starting from a perturbed version of the ground-truth with a fraction $q$ of random spin flips (with initial magnetization $m_0 (q)$). We report the dependence of $m_f$ on the initial overlap $m_0 (q)$ for various values of $K$ ($50$, first row; $200$, second row; $500$, third row), the number of examples per class $M$ ($10$, left; $50$, center; $200$, right), and the dilution parameter $d$ ($0$, blue squares; $0.5$, yellow ones). The network size is fixed to $N=1000$. The results are averaged over 500 different realizations of the couplings matrix for each point.
}
    \label{fig:num_attract}
\end{figure}

To better highlight this peculiar behavior, we focus on the extreme case $K=M=500$ and quality $r=0.9$, and analyze the behavior of $m_f-m_0$ as a function of the spin-flip fraction $q$ for various values of $d$. The results are reported in Fig. \ref{fig:attract_vs_q} (left plot). In this setting, invariant sets correspond to values of $q$ in the range $(0,0.5)$ where $m_f-m_0$ develops a zero.\footnote{The upper bound 0.5 for the spin flip fraction arise since, for $q>0.5$, the perturbed configuration becomes positively correlated to $-\bb \xi^\mu$ rather than $-\bb\xi^\mu$, and neural dynamics will drive the system towards the opposite of the pattern.} As shown in the figure, increasing the dilution not only improves the pattern reconstruction performances (as shown by the hierarchy of the curves $d=0,0.1,0.2$), but also shifts the zero of $m_f-m_0$ towards lower values of $q$. This signals again that the invariant sets in presence of dilution shrink the radius $R(q^*)$: the corresponding radius $R(q^*)$ is smaller for diluted datasets than for undiluted ones. Moreover, configurations with $q>q^*$ -- outside the ball delimited by the invariant set -- have higher magnetization than the initial overlap with the ground truth, with the situation being reversed starting inside, namely $q<q^*$: thus, configurations around the sphere of radius $R(q^*)$ are pushed towards it by neural dynamics. The phenomenon is schematized in the inset of the plots in Fig. \ref{fig:attract_vs_q} (left plot), and turns out to be qualitatively similar to the one observed in \cite{agliari2024spectral}.
\begin{figure}[h!]
    \centering
    \includegraphics[width=0.45\linewidth]{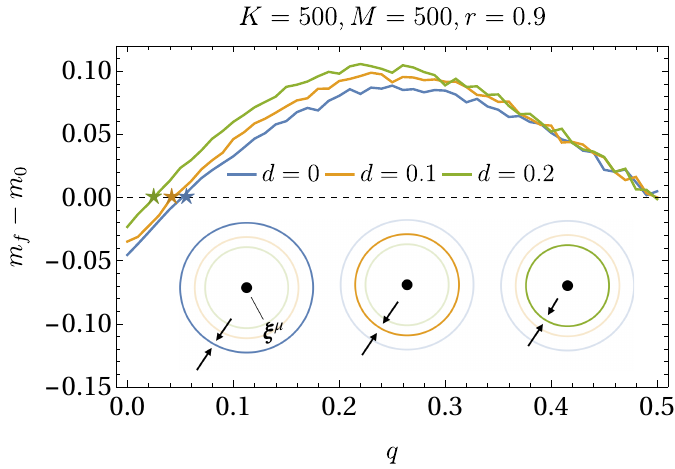}
    \includegraphics[width=0.45\linewidth]{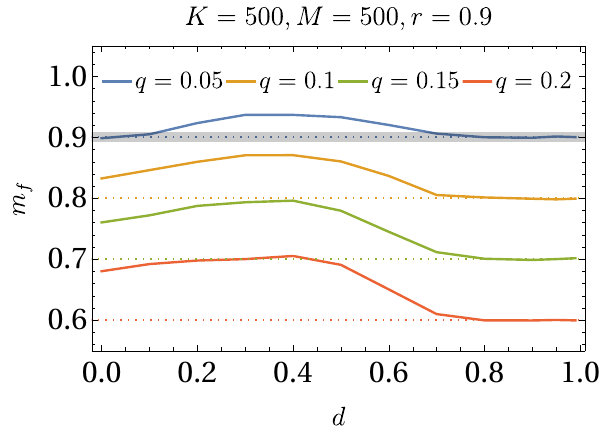}
    \caption{{\bfseries Effects of dilution on pattern reconstruction capabilities of the model.} The plots show the effect of dilution on the reconstruction capabilities of the Hopfield model starting from an initial condition with a fraction $q$ of spin-flips w.r.t. the corresponding ground-truth as $d$ is tuned. (Left) Numerical results for the difference $m_f-m_0$ between the final (i.e. after convergence to fixed point of the neural dynamics) and the initial overlap of the network configuration with the pattern are obtained for $K=M=500$ and $r=0.9$ and different values of the dilution, with $q\in[0,0.5]$ (where the initial condition exhibits a non-zero positive correlation with the pattern). The stars refer to the value of $q$ for each value of $d$ at which $m_f=m_0$, that is invariant sets of configurations. The inset shows a schematic representation of these invariant sets, with the arrow meaning that configurations inside or outside these balls are pushed towards these invariant sets (despite not relaxing in them). (Right) Numerical results for the final magnetization $m_f$ as a function of the dilution starting from initial conditions consisting in flipping a fraction $q$ of the bits in a given pattern, for various values of $q$. The horizontal dotted lines correspond to $1-2q$, namely the overlap of the initial condition with the associated pattern, while the gray thick line represents the radius of the validation sphere, namely at $q=0.05$, $1-2q =r$. At moderate dilution ($d\sim 0.4\div 0.5$) the model achieves the best reconstruction capabilities. In the experiments, the network size is fixed to $N=1000$, $K=M=500$ and $r=0.9$; numerical results are averaged over 20 different realizations of the training dataset.}
    \label{fig:attract_vs_q}
\end{figure}

As a side analysis, we also study the dependence of the model retrieval capabilities as a function of $d$ for different values of the spin-flip fraction $q$. The results -- which focus on extreme values of number of classes and examples, namely $K=M=500$, and (moderately) high quality $r=0.9$ -- are reported in Fig. \ref{fig:attract_vs_q} (right plot). For $q=0.05$, corresponding to an initial condition lying on the validation sphere\footnote{
For a fixed pattern $\bb\xi^\mu$, we define the {\it validation sphere} associated to that pattern as the set of all possible associated test examples for given quality $r$. For large $N$, the Hamming distance between the pattern and the validation example of the form $\bb x= \bb\eta \odot \bb\xi^\mu$ with $\eta_i \sim_{i.i.d.} Rad(r)$ concentrates at $R=(1-r)/2$. Hence, test examples is located on the Hamming sphere centered in $\bb\xi^\mu$ and radius $R$.} (the correlation between the initial condition with the ground-truth is $1-2q=r=0.9$), the undiluted model is not able to generalize, as the final condition falls on the validation sphere itself. Introducing dilution results in better reconstruction capabilities of the ground-truth, up to a critical value of $d\approx0.4$, and further diluting training points leads to a loss of the retrieval capabilities: for extreme dilution (namely $d\gtrsim0.8$), the final overlap again equals the initial correlation of the neural configurations. A similar behavior is met for higher value of $q$, i.e. initial condition farther from the hidden pattern: after a beneficial effect of dilution, above $d\gtrsim 0.8$ the model is no longer able to accomplish pattern reconstruction. Since this features is independent on $q$, it signals a flaw hidden in the coupling matrix, consistent with the hypothesis that -- for extreme dilution -- the latter trivializes to a diagonal (actually, a multiple of the identity) matrix. As a result, every initial configuration becomes fixed point for the neural dynamics, and the network loses its pattern reconstruction capabilities.


The above analysis highlights the importance to check the consequences of dilution on the algebraic properties of the coupling matrix for better understanding the mechanism through which an extreme dilution leads to trivialization of the neural dynamics. The diagonal entries in the coupling matrix are $
\Gamma^u_{i,i}= \frac1{NM}\sum_{\mu A} (\xi^\mu_i )^2 (\chi^{\mu,A}_i )^2=\frac1{NM}\sum_{\mu A}(\chi^{\mu,A}_i )^2$ since $(\xi^\mu_i )^2=1$, while the off-diagonal terms are
$\Gamma^u_{i,j}= \frac1{NM}\sum_{\mu} \xi^\mu_i \xi^\mu_j \sum_{A}\chi^{\mu,A}_i \chi^{\mu,A}_j$. In the first case, $\mathcal P ((\chi^{\mu,A}_i)^2\neq 0) = 1-d$, thus, for $KM$ large enough, the total contribution of the sum is $\approx MK(1-d)$, therefore $\Gamma^u_{i,i} \approx \alpha (1-d)$ (with fluctuations of the order $(NM)^{-1/2}$). As for the off-diagonal entries, due to spatial independence of the training points, we have $P(\chi^{\mu,A}_i \chi^{\mu,A}_j \neq 0)= (1-d)^2$, thus the contribution of the sum over the examples is $\approx MK(1-d)^2$. Hence, dilution implies an uneven reduction of autointeractions and synaptic strength at different sites. For very high $d$, the leading behavior of the coupling matrix will then be $\bb \Gamma^u \approx \alpha (1-d) \bb 1+\mathcal O (1-d)^2$. To quantify the impact of the dilution mechanism, we introduce two metrics. First, we recall the definition of local pattern stability: $\Delta_i = \sum_{j=1}^N \Gamma^u_{i,j}\xi^\mu_i= \Gamma^u_{i,i}\xi^\mu_i + \sum_{j\neq i} \Gamma^u_{i,j}\xi^\mu_j$. In order for this quantity to be non-trivial, we should ensure that contributions coming from the diagonal and off-diagonal entries are of the same order for all $\mu$, namely $\sum_{j\neq i }\Gamma^u_{ij} /\Gamma^u_{ii} = \mc O(1)$. If instead $\sum_{j\neq i }\Gamma^u_{ij} /\Gamma^u_{ii} \ll1$, the contribution of the diagonal is large as a result of dilution, making neural dynamics trivial. Then, we study the quantity
\begin{equation}
    \label{eq:R}
    R(d) = \Big\vert\frac{\sum_{j\neq i }\Gamma^u_{ij}(d)}{\Gamma^u_{ii}(d)}\Big\vert,
\end{equation}
as a function of $d$. This ratio provides a measure of the relative strength of off-diagonal to diagonal contributions in the coupling matrix and is thus indicative of the degree of interaction among different spins. An alternative way to quantify the deviation from the diagonal matrix is the following: the eigenvalues of a diagonal matrix are the entries on its diagonal, with associated eigenvectors $\bb v ^{(a)}$ with zeros everywhere but the $a$-th components, which is equal to 1. However, since entries in the coupling matrix goes at least as $c( 1-d)$, we define the renormalized coupling matrix $\tilde {\bb \Gamma}$ given by $\bb \Gamma^u = (1-d) \tilde {\bb \Gamma}$, with diagonal entries being of order 1 for all $d\in [0,1]$.\footnote{Clearly, this rescaling has no effect on the quantity $R(d)$.} Then, deviations from the diagonal structure can be decoded from the quantity 
\begin{equation}
\label{eq:pi}
   \pi(d) = \max_{a}\lVert\tilde {\bb \Gamma}\cdot \bb v^a -\tilde \Gamma_{a,a} \bb v^a\lVert , 
\end{equation}
which we refer to as the proximality to the diagonal structure. Clearly, if $\pi(d)\ll1$, the coupling matrix is nearly diagonal, otherwise it consistently deviates from this structure. The numerical results for $M=500$, $r=0.9$ and various values of $K$ as a function of $d$ are reported in Fig. \ref{fig:proximality}. As is expected, increasing dilution leads to a decrease in both the proximality and the relative ratio of off-diagonal and diagonal entries in the coupling matrix. Values close to zero for both the metrics are reached for extreme dilution, specifically $d\gtrsim0.8$, where trivialization occurs. This behavior corroborates the findings on reconstruction power reported in Fig. \ref{fig:attr_vs_d}, thus confirming that every configurations there becomes fixed point for neural dynamics. It is also worth noticing that, in the regime where dilution improves reconstruction capabilities, the coupling matrix remains far from the diagonal structure. Then, at intermediate values of dilution, the noise is effectively mitigated, the neural dynamics is attracted by invariant sets close to the patterns (as discussed in Fig. \ref{fig:attract_vs_q}) without compromising the non-trivial structure of the network interactions. In this regime, the neural dynamics continues to operate meaningfully, avoiding the degeneracy associated with overly sparse (i.e., near-diagonal) couplings.

\begin{figure}
    \centering
    \includegraphics[width=\linewidth]{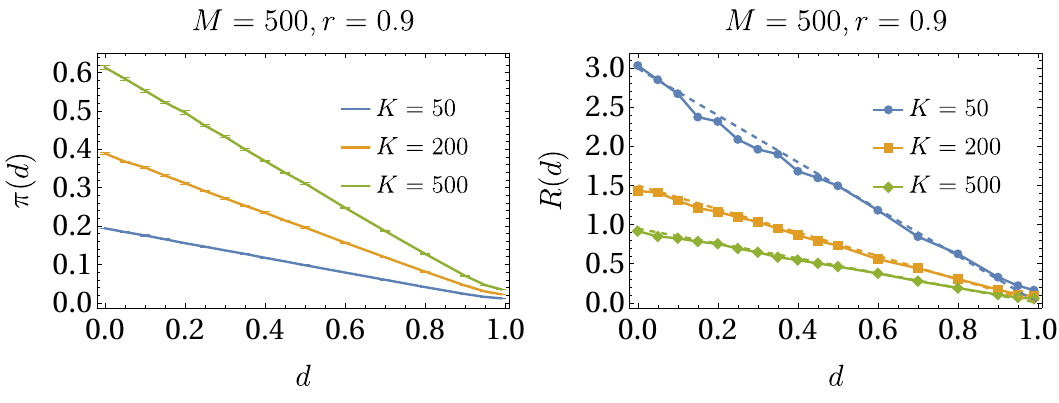}
    \caption{{\bfseries Analysis of the algebraic structure of the coupling matrix as a function of the dilution.} The two plots reports the behavior of the proximality $\pi(d)$ (left, see Eq. \eqref{eq:pi}) and the off-diagonal/diagonal ratio $R(d)$ (right, see Eq. \eqref{eq:R}) for $M=500$, $r=0.9$ and various values of $K=50,200,500$. In the right plot, the dashed lines corresponds to linear fit of the form $R_{fit} = A(1-d)$, in agreement with the scalings $\Gamma^u_{i,i}\propto (1-d)$ and $\Gamma^u_{i,j} \propto (1-d)^2$ for $i\neq j$. The network size is fixed to $N=1000$, the numerical results are averaged over 20 different realization of the training dataset for each value of $d$. Both the metrics behaves as $\sim 1-d$.}
    \label{fig:proximality}
\end{figure}

\bibliographystyle{ieeetr}
\bibliography{NeuralNetworks}

\begin{thebibliography}{10}

\bibitem{little2019statistical}
R.~J. Little and D.~B. Rubin, {\em Statistical analysis with missing data}.
\newblock John Wiley \& Sons, 2019.

\bibitem{shorten2019survey}
C.~Shorten and T.~M. Khoshgoftaar, ``A survey on image data augmentation for
  deep learning,'' {\em Journal of big data}, vol.~6, no.~1, pp.~1--48, 2019.

\bibitem{feng2021survey}
S.~Y. Feng, V.~Gangal, J.~Wei, S.~Chandar, S.~Vosoughi, T.~Mitamura, and
  E.~Hovy, ``A survey of data augmentation approaches for nlp,'' {\em arXiv
  preprint arXiv:2105.03075}, 2021.

\bibitem{huber1992robust}
P.~J. Huber, ``Robust estimation of a location parameter,'' in {\em
  Breakthroughs in statistics: Methodology and distribution}, pp.~492--518,
  Springer, 1992.

\bibitem{barron2019general}
J.~T. Barron, ``A general and adaptive robust loss function,'' in {\em
  Proceedings of the IEEE/CVF conference on computer vision and pattern
  recognition}, pp.~4331--4339, 2019.

\bibitem{breiman2001random}
L.~Breiman, ``Random forests,'' {\em Machine learning}, vol.~45, pp.~5--32,
  2001.

\bibitem{prechelt2002early}
L.~Prechelt, ``Early stopping-but when?,'' in {\em Neural Networks: Tricks of
  the trade}, pp.~55--69, Springer, 2002.

\bibitem{srivastava2014dropout}
N.~Srivastava, G.~Hinton, A.~Krizhevsky, I.~Sutskever, and R.~Salakhutdinov,
  ``Dropout: a simple way to prevent neural networks from overfitting,'' {\em
  The journal of machine learning research}, vol.~15, no.~1, pp.~1929--1958,
  2014.

\bibitem{bishop1995training}
C.~M. Bishop, ``Training with noise is equivalent to tikhonov regularization,''
  {\em Neural computation}, vol.~7, no.~1, pp.~108--116, 1995.

\bibitem{goodfellow2014explaining}
I.~J. Goodfellow, J.~Shlens, and C.~Szegedy, ``Explaining and harnessing
  adversarial examples,'' {\em arXiv preprint arXiv:1412.6572}, 2014.

\bibitem{chen2020simple}
T.~Chen, S.~Kornblith, M.~Norouzi, and G.~Hinton, ``A simple framework for
  contrastive learning of visual representations,'' in {\em International
  conference on machine learning}, pp.~1597--1607, PmLR, 2020.

\bibitem{devlin2019bert}
J.~Devlin, M.-W. Chang, K.~Lee, and K.~Toutanova, ``Bert: Pre-training of deep
  bidirectional transformers for language understanding,'' in {\em Proceedings
  of the 2019 conference of the North American chapter of the association for
  computational linguistics: human language technologies, volume 1 (long and
  short papers)}, pp.~4171--4186, 2019.

\bibitem{EmergencySN}
E.~Agliari, F.~Alemanno, A.~Barra, and G.~De~Marzo, ``The emergence of a
  concept in shallow neural networks,'' {\em Neural Networks}, vol.~148,
  pp.~232--253, 2022.

\bibitem{Leuzzi_2022}
L.~Leuzzi, A.~Patti, and F.~Ricci-Tersenghi, ``A quantitative analysis of a
  generalized hopfield model that stores and retrieves mismatched memory
  patterns,'' {\em Journal of Statistical Mechanics: Theory and Experiment},
  vol.~2022, p.~073301, aug 2022.

\bibitem{alemanno2023supervised}
F.~Alemanno, M.~Aquaro, I.~Kanter, A.~Barra, and E.~Agliari, ``Supervised
  {H}ebbian learning,'' {\em Europhysics Letters}, vol.~141, no.~1, p.~11001,
  2023.

\bibitem{negri2024memorization}
B.~Pham, G.~Raya, M.~Negri, M.~J. Zaki, L.~Ambrogioni, and D.~Krotov,
  ``Memorization to generalization: The emergence of diffusion models from
  associative memory,'' 2024.

\bibitem{negri2024random}
M.~Negri, C.~Lauditi, G.~Perugini, C.~Lucibello, and E.~M. Malatesta, ``Random
  {F}eatures {H}opfield {N}etworks generalize retrieval to previously unseen
  examples,'' in {\em Associative Memory {\&} Hopfield Networks in 2023}, 2023.

\bibitem{catania2025theoretical}
G.~Catania, A.~Decelle, C.~Furtlehner, and B.~Seoane, ``A theoretical framework
  for overfitting in energy-based modeling,'' {\em arXiv preprint
  arXiv:2501.19158}, 2025.

\bibitem{Fontanari-1990}
J.~Fontanari, ``Generalization in a {H}opfield network,'' {\em Journal of
  Physics France}, vol.~51, pp.~2421--2430, 1990.

\bibitem{Litinskii}
L.~Litinskii, ``Generalization in the {H}opfield model,'' in {\em IJCNN'01.
  International Joint Conference on Neural Networks. Proceedings (Cat.
  No.01CH37222)}, vol.~1, pp.~65--70 vol.1, 2001.

\bibitem{Theumann}
P.~R. Krebs and W.~K. Theumann, ``Generalization in a hopfield network with
  noise,'' {\em Journal of Physics A: Mathematical and General}, vol.~26,
  p.~3983, aug 1993.

\bibitem{regularizationdreaming}
E.~Agliari, M.~Aquaro, F.~Alemanno, and A.~Fachechi, ``Regularization,
  early-stopping and dreaming: a hopfield-like setup to address generalization
  and overfitting,'' {\em Neural Networks}, vol.~177, p.~106389, 2024.

\bibitem{fachechi2019dreaming}
A.~Fachechi, E.~Agliari, and A.~Barra, ``Dreaming neural networks: forgetting
  spurious memories and reinforcing pure ones,'' {\em Neural Networks},
  vol.~112, pp.~24--40, 2019.

\bibitem{leonelli2021effective}
F.~E. Leonelli, E.~Agliari, L.~Albanese, and A.~Barra, ``On the effective
  initialisation for restricted {B}oltzmann machines via duality with
  {H}opfield model,'' {\em Neural Networks}, vol.~143, pp.~314--326, 2021.

\bibitem{AABD-NN2022}
E.~Agliari, F.~Alemanno, A.~Barra, and G.~De~Marzo, ``The emergence of a
  concept in shallow neural networks,'' {\em Neural Networks}, vol.~148,
  pp.~232--253, 2022.

\bibitem{AAKBA-EPL2023}
M.~Aquaro, F.~Alemanno, I.~Kanter, A.~Barra, and E.~Agliari, ``Supervised
  {H}ebbian learning,'' {\em Europhysics Letters - Perspective}, vol.~141,
  p.~11001, 2023.

\bibitem{Picco}
A.~Bovier, V.~Gayrard, and P.~Picco, ``Artiﬁcial neural networks. a review
  from physical and mathematical points of view,'' {\em Annales de l’I. H.
  P., section A}, vol.~64, no.~3, pp.~289--307, 1996.

\bibitem{Coolen}
A.~Coolen, R.~K\"{u}hn, and P.~Sollich, {\em Theory of neural information
  processing systems}.
\newblock Oxford University Press, 2005.

\bibitem{Kohonen-1972}
T.~Kohonen and M.~Ruohonen, ``Representation of {A}ssociated {D}ata by {M}atrix
  {O}perators,'' {\em IEEE Transaztions on Computers}, 1973.

\bibitem{Gardner}
E.~Gardner, ``Multiconnected neural network models,'' {\em Journal of Physics
  A: General Physics}, vol.~20, 1987.

\bibitem{FAB-NN2019}
A.~Fachechi, E.~Agliari, and A.~Barra, ``Dreaming neural networks: Forgetting
  spurious memories and reinforcing pure ones,'' {\em Neural Networks},
  vol.~112, pp.~24--40, 2019.

\bibitem{AABF-JStat2019}
E.~Agliari, F.~Alemanno, A.~Barra, and A.~Fachechi, ``Dreaming neural networks:
  rigorous results,'' {\em Journal of Statistical Mechanics}, p.~083503, 2019.

\bibitem{LAD}
E.~Agliari, F.~Alemanno, M.~Aquaro, A.~Barra, F.~Durante, and I.~Kanter,
  ``Hebbian dreaming for small datasets,'' {\em Neural Networks}, vol.~173,
  p.~106174, 2024.

\bibitem{Hebb-1949}
D.~Hebb, {\em The Organization of Behavior: A Neuropsychological Theory.}
\newblock New York, NY: John Wiley \& Sons, 1949.

\bibitem{hopfield1982neural}
J.~J. Hopfield, ``Neural networks and physical systems with emergent collective
  computational abilities.,'' {\em Proceedings of the national academy of
  sciences}, vol.~79, no.~8, pp.~2554--2558, 1982.

\bibitem{AGS}
D.~J. Amit, H.~Gutfreund, and H.~Sompolinsky, ``Storing infinite numbers of
  patterns in a spin-glass model of neural networks,'' {\em Physical Review
  Letters}, vol.~55, pp.~1530--1533, 1985.

\bibitem{agliari2012multitasking}
E.~Agliari, A.~Barra, A.~Galluzzi, F.~Guerra, and F.~Moauro, ``Multitasking
  associative networks,'' {\em Physical review letters}, vol.~109, no.~26,
  p.~268101, 2012.

\bibitem{agliari2013parallel}
E.~Agliari, A.~Barra, A.~De~Antoni, and A.~Galluzzi, ``Parallel retrieval of
  correlated patterns: {F}rom hopfield networks to {B}oltzmann machines,'' {\em
  Neural Networks}, vol.~38, pp.~52--63, 2013.

\bibitem{AABCT-JPA2013a}
E.~Agliari, A.~Annibale, A.~Barra, A.~Coolen, and D.~Tantari, ``Immune
  networks: multitasking capabilities at medium load,'' {\em Journal of Physics
  A}, vol.~46, p.~335101, 2013.

\bibitem{AABCT-JPA2013b}
E.~Agliari, A.~Annibale, A.~Barra, A.~Coolen, and D.~Tantari, ``Immune
  networks: multitasking capabilities near saturation,'' {\em Journal of
  Physics A}, vol.~46, p.~415003, 2013.

\bibitem{AABR-JStat2023}
E.~Agliari, A.~Alessandrelli, A.~Barra, and F.~Ricci-Tersenghi, ``Parallel
  learning by multitasking neural networks,'' {\em J. Stat. Mech.}, p.~113401,
  2023.

\bibitem{agliari2024spectral}
E.~Agliari, A.~Fachechi, and D.~Luongo, ``A spectral approach to {H}ebbian-like
  neural networks,'' {\em Applied Mathematics and Computation}, vol.~474,
  p.~128689, 2024.

\bibitem{Amit}
D.~Amit, {\em Modeling brain function: The world of attractor neural networks}.
\newblock Cambridge university press, 1989.

\bibitem{Vinci-2025}
G.~V. Vinci, A.~Galluzzi, and M.~Mattia, ``Beyond catastrophic forgetting in
  associative networks with self-interactions,'' {\em arXiv preprint
  arXiv:2504.04560}, 2025.

\bibitem{bai2010spectral}
Z.~Bai and J.~W. Silverstein, {\em Spectral analysis of large dimensional
  random matrices}, vol.~20.
\newblock Springer, 2010.

\bibitem{pinelis2016optimal}
I.~Pinelis and R.~Molzon, ``Optimal-order bounds on the rate of convergence to
  normality in the multivariate delta method,'' 2016.

\bibitem{horn2012matrix}
R.~A. Horn and C.~R. Johnson, {\em Matrix analysis}.
\newblock Cambridge university press, 2012.

\bibitem{bernstein2020distance}
J.~Bernstein, A.~Vahdat, Y.~Yue, and M.-Y. Liu, ``On the distance between two
  neural networks and the stability of learning,'' {\em Advances in Neural
  Information Processing Systems}, vol.~33, pp.~21370--21381, 2020.

\bibitem{mceliece1987capacity}
R.~McEliece, E.~Posner, E.~Rodemich, and S.~Venkatesh, ``The capacity of the
  hopfield associative memory,'' {\em IEEE transactions on Information Theory},
  vol.~33, no.~4, pp.~461--482, 1987.

\end{thebibliography}
\end{document}